\documentclass{article}

\usepackage{PRIMEarxiv}

\usepackage[utf8]{inputenc} 
\usepackage[T1]{fontenc}    
\usepackage{hyperref}       
\usepackage{url}            
\usepackage{booktabs}       
\usepackage{amsfonts}       
\usepackage{amsthm}  
\usepackage{nicefrac}       
\usepackage{microtype}      
\usepackage{lipsum}
\usepackage{graphicx}

\usepackage{booktabs} 
\RequirePackage{newtxtext}
\RequirePackage{newtxmath}

\usepackage{booktabs}        
\usepackage{array}           
\usepackage{threeparttable}  

\usepackage{algorithm}
\usepackage{algpseudocode}

\usepackage{subcaption}
\usepackage{caption}

\hypersetup{
    colorlinks=true,    
    linkcolor=red,      
    citecolor=blue,      
    urlcolor=blue        
}

\graphicspath{{Fig/}}

\title{A Self-scaled Approximate $\ell_0$ Regularization Robust Model for Outlier Detection
}

\author{
  Pengyang Song, Jue Wang \\
  AMSS Center for Forecasting Science \\
  Chinese Academy of Sciences \\
  Beijing \\
  \texttt{\{Jue Wang\}wjue@amss.ac.cn} \\
}

\begin{document}
\maketitle

\newtheorem{theorem}{Theorem}
\newtheorem{lemma}{Lemma}            
\newtheorem{proposition}[theorem]{Proposition}
\newtheorem{corollary}[theorem]{Corollary}
\newtheorem{definition}{Definition}
\newtheorem{remark}{Remark}
\newtheorem{example}[theorem]{Example}
\newtheorem{assumption}[theorem]{Assumption}

\begin{abstract}
Robust regression models in the presence of outliers have significant practical relevance in areas such as signal processing,
financial econometrics, and energy management. Many existing robust regression methods, either grounded in statistical theory
or sparse signal recovery, typically rely on the explicit or implicit assumption of outlier sparsity to filter anomalies and
recover the underlying signal or data. However, these methods often suffer from limited robustness or high computational complexity,
rendering them inefficient for large-scale problems.
In this work, we propose a novel robust regression model based on a Self-scaled Approximate $\ell_0$ Regularization Model (SARM) scheme.
By introducing a self-scaling mechanism into the regularization term, the proposed model mitigates the negative impact of uneven
or excessively large outlier magnitudes on robustness. We also develop an alternating minimization algorithm grounded in Proximal Operators and Block Coordinate Descent. We rigorously prove the algorithm convergence, and under the Restricted Isometry Property (RIP) framework, we derive a theoretical upper bound on the estimation error associated with SARM.
Empirical comparisons with several state-of-the-art robust regression methods demonstrate that SARM not only achieves
superior robustness but also significantly improves computational efficiency. Motivated by both the theoretical error bound and
empirical observations, we further design a Two-Stage SARM (TSSARM) framework, which better utilizes
sample information when the singular values of the design matrix are widely spread, thereby enhancing robustness under certain conditions.
Finally, we validate our approach on a real-world load forecasting task.
The experimental results show that our method substantially enhances the robustness of load forecasting against adversarial data
attacks, which is increasingly critical in the era of heightened data security concerns.
\end{abstract}

\keywords{First keyword \and Second keyword \and More}

\section{Introduction}

For a general linear regression model with $m$ training sample pairs $\{(x_i,y_i): x_i \in \mathbb{R}^n, y_i \in \mathbb{R}\}$, it is typically assumed that $x_i$ and $y_i$ satisfy the following relationship \cite{liu2018robust}:
\begin{equation} \label{normal}
y_i = x_i^T w  + e_i, i = 1,2,\cdots, m
\end{equation}
or in matrix form: $y = X w  + e$, where $y = (y_1, y_2, \cdots y_m)^T (\in \mathbb{R}^m)$ and $X = (x_1, x_2, \cdots x_m)^T (\in \mathbb{R}^{m\times n})$ are all known. $w \in \mathbb{R}^n$ is the parameter that need to be estimated and $\epsilon = (\epsilon_1, \epsilon_2, \cdots \epsilon_m)^T \in \mathbb{R}^m$ represents the observation error.
Without loss of generality, the matrix $X$ is assumed to be full rank.
However, in practical applications of regression models, training samples may contain some labels $y_i$ that are corrupted by large observation errors \cite{papageorgiou2015robust}, or we can call them outliers. 
In this case, it primarily assumed that the samples follow the relationship given below \cite{liu2018robust, suggala2019adaptive}:
\begin{equation} \label{attackted}
y = X w  + e + z,
\end{equation}
where $z \in \mathbb{R}^m$ denotes outliers and satisfies the following two assumptions:
\begin{enumerate}
    \item Outliers typically exhibit large observation errors, i.e. $min\{ |z_i|, z_i \neq 0 \} > \|\epsilon \|_{\infty}$.
    \item Only a subset of the training sample pairs contains outliers, i.e. $\|z\|_0 < m$, where $\| \cdot\|_0$ denotes the cardinality of non-zero entries.
\end{enumerate}
These outliers may result from unexpected events, such as when images are affected by accidental occlusions or illumination changes \cite{chen2021delta}, leading to anomalous pixels. Alternatively, they could arise from deliberate manipulation, for example, load demand data may be subject to data attacks aimed at disrupting power system scheduling \cite{luo2018robust, luo2023robust}. 
These outliers, originating from different sources, may significantly affect the estimation performance of Ordinary Least Squares (OLS) \cite{de2021review}.
In this context, robust regression models, a class of regression models that can resist the impact of outliers on the estimation of parameters, hold significant practical importance in numerous fields, including power systems \cite{hong2022data}, image processing \cite{zoubir2018robust, buccini2020large}, robot perception \cite{antonante2022outlier} and so on.

To the best of our knowledge, there are three main directions in the development of robust regression models. 

The first direction is Robust Statistics \cite{maronna2019robust}, represented by M-estimators and L-estimators \cite{de2021review}.
Specifically, several notable methods include Least Absolute Deviations (LAD), Iteratively Reweighted Least Squares (IRLS) \cite{beaton1974fitting, holland1977robust},  Least Median of Squares (LMedS) \cite{rousseeuw1984least}, Least Trimmed Squares (LTS) \cite{maronna2019robust}, and Random Sample Consensus (RANSAC) \cite{fischler1981random}. These methods each still have certain drawbacks. For example, RANSAC, LMedS and LTS face challenges in handling high-dimension problems \cite{liu2018robust, mitra2012analysis} and LAD is still influenced by outlier observation errors to some extent.

The second research direction approaches the problem from the perspective of Sparse Signal Recovery (SSR) \cite{mitra2012analysis}. 
By further assuming the sparsity of $z$ in (\ref{attackted}), Restricted Isometry Property (RIP) \cite{candes2005decoding} and techniques from Compressed Sensing can be applied.
Representative algorithms in this category include GARD \cite{papageorgiou2015robust, kallummil2019noise}, which employs greedy strategies, and AROSI \cite{liu2018robust}, which utilizes $\ell_0$-norm regularization. These two methods, especially GARD, still pose potential risks in terms of computational complexity when applied to large-scale datasets.

The third research direction is to reduce the impact of outliers by designing Non-convex Bounded Loss Functions \cite{chen2021delta,buccini2020large,fu2023robust}.
Among strategies for constructing such loss functions, truncation is one of the most popular approaches and TLRM is a representative example \cite{huang2024large}. The non-convex nature often introduces some optimization challenges for these models.

When corruption rate is low, most of the aforementioned methods generally perform well. However, when outliers stem from deliberate adversarial attacks—for example, cyber attacks targeting power grid systems \cite{luo2023robust, hong2022data}—corruption rate in the data may become significantly higher, and then many robust models tend to fail.
Moreover, some models are either designed based on LAD or require LAD as an initial estimator \cite{liu2018robust, she2011outlier}. Since LAD is typically solved via linear programming, such models often have high computational complexity when dealing with large-scale problems. On the other hand, models that do not rely on LAD tend to fail earlier as corruption rate increases, i.e. they have a smaller breakdown point.
In addition, we further observe that many existing robust models tend to focus on model structure or algorithm design, while relatively few studies consider enhancing robustness from the perspective of the training samples. Information such as sample feature dimensionality is often closely related to model robustness and should not be overlooked.


Based on the above considerations, our contributions can be summarized as follows:
\begin{enumerate}
    \item We propose a Self-scaled Approximate $\ell_0$ Regularization Model (SARM) and design an tailored alternating minimization algorithm. SARM not only demonstrates strong robustness but also exhibits low computational complexity. The algorithmic structure of SARM also makes it highly compatible with GPU parallel computing, resulting in better computational efficiency on large-scale datasets.
    \item We rigorously prove the algorithm convergence, and under the Restricted Isometry Property framework, we derive a theoretical upper bound on the estimation error associated with SARM.
    \item Motivated by both the theoretical error bound and empirical observations, we further design a Two-Stage SARM (TSSARM) framework, which better utilizes sample information when the singular values of the design matrix satisfy certain conditions. We also provide a model selection scheme for SARM and TSSARM.
\end{enumerate}

We compare our method with several state-of-the-art robust models on ideal simulated data, and further validate its performance on a real-world power forecasting case. Experimental results demonstrate the effectiveness of the proposed model.

The remainder of this paper is organized as follows: Section \ref{SARM} describes the structure of the SARM model and the related algorithms. Section \ref{Theorem} presents the convergence analysis of the SARM-related algorithms and the theoretical error bound of SARM. Section \ref{Simulation} showcases the results of the simulation experiments. Section \ref{TSSARM} discusses the structure of the TSSARM model, the simulation results, and the model selection strategy. Section \ref{Applied} provides a case study on electrical load demand forecasting. The Appendix contains the detailed proof for Section \ref{Theorem}.

\subsection{Notations}

For clarity, we present the key notations and definitions that will be employed in this paper.

Assume that $\mathbf{a} = [a_1, a_2, \cdots a_d]^T \in \mathbb{R}^d$ and $\mathbf{b} = [b_1, b_2, \cdots b_d]^T \in \mathbb{R}^d$ are vectors, and
$\sigma : \mathbb{R}^d \to (-\infty, +\infty]$ be a proper and lower semicontinuous function.
\begin{enumerate}
\item The division $\frac{\mathbf{a}}{\mathbf{b}}$ is defined as:
\begin{equation}  \label{division}
 \frac{\mathbf{a}}{\mathbf{b}} = \left[\frac{a_1}{b_1}, \frac{a_1}{b_1}, \cdots, \frac{a_d}{b_d} \right]^T.
\end{equation}
\item If $f(x)$ is a single-valued function, $f(\mathbf{a})$ is defined as
\begin{equation} \label{multi-value function}
 f(\mathbf{a}) = \left[f(a_1) , f(a_2) , \cdots, f(a_d)  \right]^T.
\end{equation}
\item The concatenation of vectors \(\mathbf{a}\) and \(\mathbf{b}\) is defined as:
\begin{align} \label{concatenation}
    [\mathbf{a}; \mathbf{b}] = [a_1, \cdots a_d, b_1, \cdots b_d]^T
\end{align} 
\item An alternative representation of a vector $\mathbf{a}$ is defined as:
\begin{align} \label{component}
    [a_i] = [a_1, a_2, \cdots , a_d]^T = \mathbf{a}
\end{align} 
    \item For a given $\mathbf{b} \in \operatorname{dom} \sigma$, the \emph{Fréchet subdifferential} of $\sigma$ at $\mathbf{b}$, 
    written $\hat{\partial} \sigma(\mathbf{b})$, is the set of all vectors $\mathbf{u} \in \mathbb{R}^d$ which satisfy:
    \begin{align}
    \liminf_{\mathbf{a} \to \mathbf{b},\, \mathbf{a} \neq \mathbf{b}} \frac{\sigma(\mathbf{a}) - \sigma(\mathbf{b}) - \langle \mathbf{u} , \mathbf{a} - \mathbf{b} \rangle}{\|\mathbf{a} - \mathbf{b}\|} \geq 0.  \notag
    \end{align} 
    \item The \emph{limiting subdifferential}, or simply the \emph{subdifferential}, of $\sigma$ at $\mathbf{x} \in \mathbb{R}^d$, written $\partial \sigma(\mathbf{x})$, 
    is defined as \cite{mordukhovich2006variational}:
    \begin{align}
    \partial \sigma(x) := \left\{ \mathbf{u} \in \mathbb{R}^d : \exists\, \mathbf{x}_k \to \mathbf{x},\ \sigma(\mathbf{x}_k) \to \sigma(\mathbf{x}),\ \text{and } \mathbf{u}_k \in \hat{\partial} \sigma(x_k) \to \mathbf{u} \text{ as } k \to \infty \right\}. \notag
    \end{align} 
    \item For any set $D \subset \mathbb{R}^d$ and any point $\mathbf{a} \in \mathbb{R}^d$, the distance from $\mathbf{a}$ to $D$ is defined as
    \begin{align} \label{dist}
      \operatorname{dist}(\mathbf{a}, D) := \inf \{ \|\mathbf{b} - \mathbf{a}\| : \mathbf{b} \in D \}. \notag
    \end{align}
    If $S = \emptyset$, we define $\operatorname{dist}(\mathbf{b}, D) = \infty$ for all $\mathbf{b}$.
\end{enumerate}
A brief discussion on some basic properties of the \emph{Fréchet subdifferential} and the \emph{subdifferential} can be found in 
Definition 1 and the corresponding Remark 1 and  Proposition 1 of \cite{bolte2014proximal}.

The following presents the definitions of other notations:
\begin{enumerate}
    \item \textbf{Breakdown point}: It is defined as the maximum fraction of corrupted data beyond which the estimator can no longer guarantee recovery of the parameter $w$ with small error \cite{suggala2019adaptive,hampel1971general}.
    \item \textbf{$\ell_0$-Norm ($\|\cdot\|_0$)}:  It denotes the cardinality of non-zero entries of a vector. For a vector $\boldsymbol{a} \in \mathbb{R}^m$, $\|\boldsymbol{a}\|_0 := |\{i \in \{1, 2, \dots, m\} : a_i \ne 0\}|$.
    \item \textbf{$\ell_p$-Norm ($\|\cdot\|_p$)}:  For a vector $\boldsymbol{a} \in \mathbb{R}^m$, $\|a\|_p = \left( \sum_{i=1}^n |a_i|^p \right)^{1/p}, \quad \text{for } p \geq 1$.
    \item \textbf{Critical Points}: Points whose subdifferential contains 0 are called (limiting-)critical points. The set of critical points of $\sigma$ is denoted by $\operatorname{crit} \sigma$.
\end{enumerate}

    

\section{The Self-scaled Approximate $\ell_0$ Regularization Model} \label{SARM}

\subsection{The $\ell_0$ Regularization Model and The Limitations} \label{motivation}

We first illustrate our motivation from the limitations of the $\ell_0$ regularization model.
The $\ell_0$ regularization model is usually expressed as:
\begin{equation} \label{L01}
\min_{w \in \mathbb{R}^n, z \in \mathbb{R}^m} \quad Loss( y - X w - z ) + \delta \| z \|_0,
\end{equation}
and a closely related form of it is:
\begin{align} \label{L02}
& \min_{w \in \mathbb{R}^n, z \in \mathbb{R}^m} \quad \quad \| z \|_0 \\
& s.t. \quad Loss( y - X w - z ) \leq \varepsilon, \notag
\end{align}
where $Loss()$ usually denotes $\| \cdot \|_2^2$.
A considerable number of effective robust regression methods, such as GARD \cite{papageorgiou2015robust}, can be directly or indirectly transformed into these forms. 
AROSI \cite{liu2018robust} shares a similar structure, if $Loss()$ in (\ref{L01}) is replaced by $\| \cdot \|_1$.
Truncated loss functions are also closely related to $\ell_0$ regularization.
Essentially, the $\ell_0$-norm regularization model perfectly aligns with (\ref{attackted}) and its associated two assumptions. 

However, the main challenge faced by researchers lies in the fact that the $\ell_0$-norm is neither continuous nor differentiable, which makes this Optimization Problem (\ref{L01}) or (\ref{L02}) challenging to solve.
A classical approach is convex relaxation, where the $\ell_0$ regularization term is replaced by an $\ell_1$ regularization. Unfortunately, similar to LAD, the $\ell_1$-based relaxation exhibits sensitivity to the magnitude of outlier errors, which can be observed from the error bounds provided in \cite{mitra2012analysis}, and a similar argument is also presented in \cite{wang2015self}. 

\begin{algorithm}
\caption{An Optimization Algorithm for $\ell_0$-regularization Problem (\ref{L01})}
\begin{algorithmic}[1]
\Require $y, X, \delta > 0$
\State Initialization: $k = 0$, $w^{(0)} = 0$, $S_0 = \{1, \ldots, m\}$
\While{algorithm not converged}
    \State Iteration $k+1$
    \State (Update $w$)
    $w^{k+1} = \underset{x}{\arg\min}  Loss(y_{S_k} - X_{S_k}w)$
    \If{$Loss(y_{S_k} - X_{S_k}w^{k+1}) = Loss(y_{S_k} - X_{S_k}w^{k})$}
        \State Further update $w^{k+1} = w^{k}$
    \EndIf
    \State (Update $z$ and $S$)
    $z_i^{k+1} = 
    \begin{cases}
        0, & \text{if } Loss(y_i - x_iw^{k+1}) \leq \alpha \\
        y_i - x_iw^{k+1}, & \text{otherwise}
    \end{cases}$
    \State $S_{k+1} := \{i: z_i^{k+1} = 0\}$
    \State $k := k + 1$
\EndWhile
\Ensure Final solution $w^{k}$ and $z^{k}$
\end{algorithmic}
\label{AROSI}
\end{algorithm}

An alternative class of methods addresses Problem (\ref{L01}) directly. Due to the inherent properties of the $\ell_0$ regularization term, these methods exhibit insensitivity to the scale of outlier errors.
A typical and effective optimization algorithm comes from AROSI and a generalized form of it is shown in Algorithm \ref{AROSI}. 
When the truncated loss regression problem is approached from the perspective of the Majorization-Minimization Algorithm (MMA), a similar algorithm can also be derived \cite{huang2024large}. 
In simple terms, Algorithm \ref{AROSI} gradually filters out outliers during each iteration until convergence. 

In terms of computational complexity and robustness, Algorithm \ref{AROSI} demonstrates great performance, outperforming many robust models, such as GARD, IRLS, and LAD \cite{liu2018robust}. However, if the initial parameter estimate $w^{(1)}$ deviates significantly from the true $w$, Algorithm \ref{AROSI} may fail to correct itself through subsequent iterations.
This is also referred to as the "swamping effect" \cite{she2011outlier}, which means that outliers cause non-outlier samples to be mistakenly identified as outliers.

By examining the objective function of Optimization Problem (\ref{L01}), we identify certain structural aspects of the objective function that contribute to the aforementioned issue. 
Due to the optimization challenges posed by the $\ell_0$ regularization term, researchers are often forced to resort to alternately updating $w$ and $z$. However, when updating $z$, Algorithm \ref{AROSI} typically removes samples whose residuals exceed a certain threshold in a "hard" manner.
When the swamping effect is severe, many non-outlier samples are likely to be mistakenly discarded, leading to model failure.
Some of the more effective methods either leverage robust models such as LAD to obtain a more stable initial estimate \cite{she2011outlier} or directly replace $Loss()$ with $\| \cdot \|_1$ in order to prevent the excessive removal of non-outliers \cite{liu2018robust}, but the high computational complexity inherently introduced by the $\ell_1$ loss function (i.e., LAD) is unavoidable. 
Therefore, directly considering the original form of $\ell_0$ regularization may not be ideal. 

From the issues discussed above, we identify two key aspects that require special attention:
\begin{enumerate}
    \item The structural advantages of $\ell_0$ regularization should be taken into account to avoid sensitivity to the magnitude of outlier errors.
    \item A surrogate formulation of the $\ell_0$ regularization should be designed to be more optimization-friendly and more robust against the swamping effect.
    \item A more reliable initial estimate should be taken into account.
\end{enumerate}

The design of SARM is primarily inspired by the first two considerations, and to some extent, SARM alleviates the need for an initial estimate obtained via LAD. As for the third consideration, we will revisit the issue of initial estimates from a different perspective in our later development of TSSARM.

\subsection{SARM: Model Formulation and Algorithm Design}

Our main modification lies in the regularization term of model (\ref{L01}).
Specifically, we construct the following Self-scaled Approximate $\ell_0$ Regularization Model (SARM):
\begin{equation} \label{SARMOPT}
\min_{w \in \mathbb{R}^n, z \in \mathbb{R}^m} \frac{1}{2}\|y-Xw-z\|_2^2 + \delta \left\|\frac{z}{y-Xw}\right\|_1.
\end{equation}
where $\frac{z}{y-Xw} $ is defined as (\ref{division}). 
We first relax the $\ell_0$-norm to an $\ell_1$-norm, and then use the residual $Xw - y$ to scale it, thereby approximating the $\ell_0$-norm regularization term. Outliers are typically considered to have much larger observation errors and $z$ is non-zero only at the outliers, therefore for the true $w$ and $z$, $ \left\|\frac{z}{y-Xw}\right\|_1  = \left\|\frac{z}{z+\epsilon}\right\|_1 \approx \| z \|_0$. 

In fact, (\ref{SARMOPT}) is sufficient for designing a robust model and the potential division-by-zero problem can be resolved by adding a small positive constant to the denominator. However, for the sake of the convergence analysis presented later, it remains necessary to further smooth the regularization term.
We introduce a smoothing function:
\begin{equation} \label{SMOOTHF}
S(x) =
\begin{cases}
  \frac{1}{ 2 \sqrt{\delta} } x^2 + \frac{ \sqrt{\delta} }{2}, & \text{if } |x| < \sqrt{\delta} \\
  | x |, & \text{if } |x| \geq \sqrt{\delta},
\end{cases} 
\end{equation}
and reformulate Optimization Problem (\ref{SARMOPT}) to:
\begin{equation} \label{SARMOPT2}
\min_{w \in \mathbb{R}^n, z \in \mathbb{R}^m} \frac{1}{2}\|y-Xw-z\|_2^2 + \delta \left\|\frac{z}{S(y-Xw)}\right\|_1.
\end{equation}
The graph of Function (\ref{SMOOTHF}) is shown in Figure \ref{fig1}. 

\begin{figure}[htbp]
  \centering
  \includegraphics[width=0.6\textwidth]{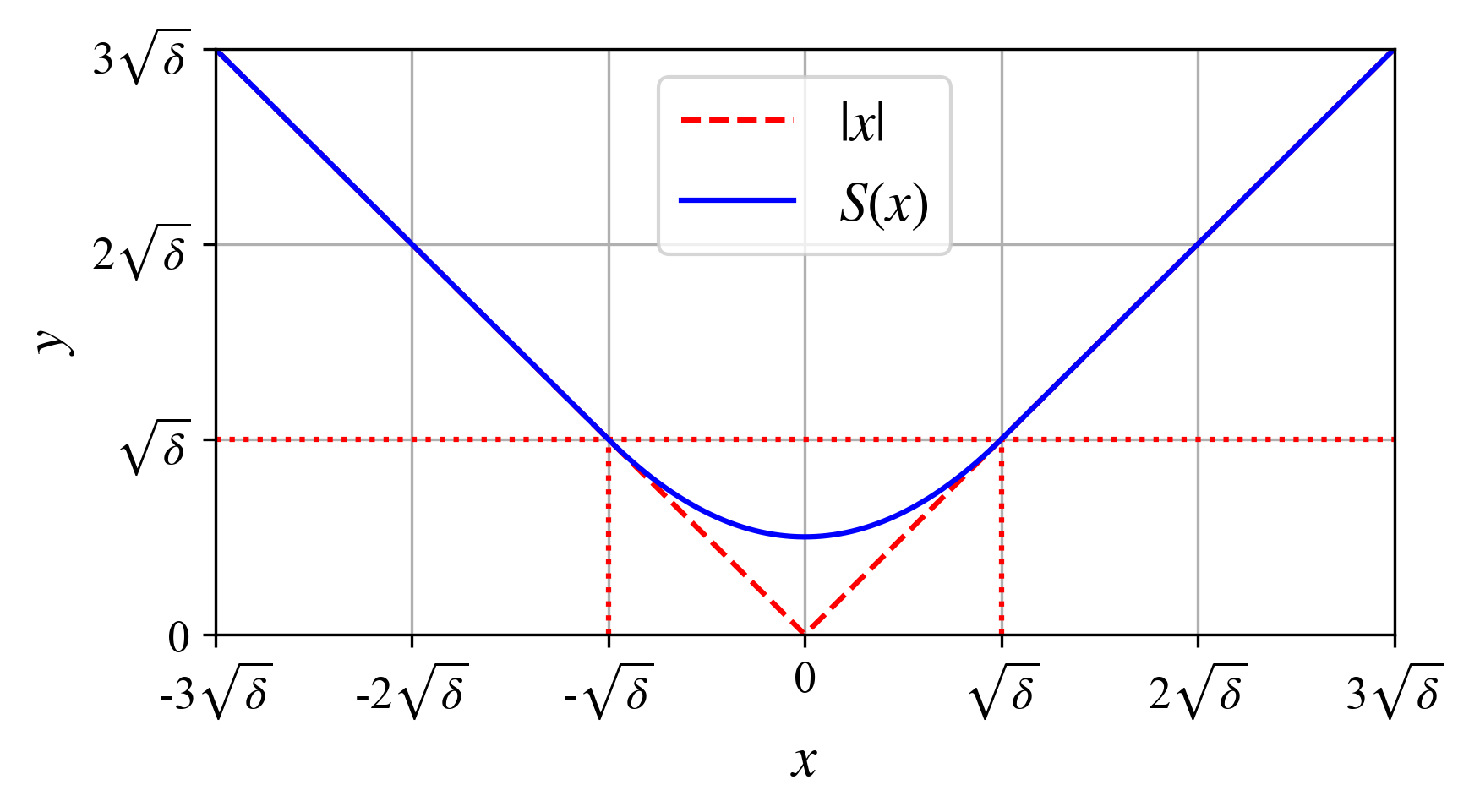}
  \caption{The graph of Function (\ref{SMOOTHF}).}
  \label{fig1}
\end{figure}

Next, we describe the iterative algorithm for (\ref{SARMOPT2}). 
The variables $z$ and $w$ are coupled, which makes direct optimization challenging. A straightforward approach is to decouple the variables $z$ and $w$ via alternating optimization. 
Specifically, at the $k$-th iteration, we first fix $z = z^{(k-1)}$ and perform one step of Gradient Descent on $w$ to obtain $w^{k}$. 

Then, we fix $w = w^{k}$ and proceed to optimize $z$. We denote the residual $y - Xw^{k} = r^{k}$, and the optimization problem is:
\begin{align} 
 & \min_{z \in \mathbb{R}^m}  \frac{1}{2}\|r^{k}-z\|_2^2 + \delta \left\|\frac{z}{S \left(r^{k}\right)}\right\|_1  \notag \\
 =& \min_{z \in \mathbb{R}^n}\sum_{i=1}^{m}  \left( \frac{1}{2}(r^{k}_i-z_i)^2 + \delta \left|\frac{z_i}{S(r^{k}_i)}\right| \right)  \notag \\
=& \sum_{i=1}^{m} \min_{z_i \in \mathbb{R}}  \left( \frac{1}{2}(r^{k}_i-z_i)^2 + \delta \left|\frac{z_i}{S(r^{k}_i)}\right| \right)  \notag \\
=& \sum_{i=1}^{m} \min_{z_i \in \mathbb{R}} \frac{1}{2}(r^{k}_i-z_i)^2 + t_i h(z_i),  \notag
\end{align}
where  $t_i = \frac{\delta}{S(r^{k}_i)} $ and $h(x) =|x|$.
Note that all $r^{k}_i$ are constants, and thus the optimal solutions $z = (z_1, z_2, \cdots, z_m)^T$ can be directly obtained via the proximal operator:
\begin{align}
z_i &={\arg\min}_{z_i \in \mathbb{R}} \frac{1}{2}(r^{k}_i-z_i)^2 + t_i h(z_i),\notag \\
&= prox_{t_i h}(r^{k}_i) \notag \\
&= sign(r^{k}_i) max\{ r^{k}_i - t_i, 0 \} \notag \\
&= sign(r^{k}_i) max \left\{ | r^{k}_i | - \frac{\delta}{S(r^{k}_i)}, 0 \right\} \label{prox1}\\
& =     \begin{cases}
        0, & \text{if } | r^{k}_i |  \leq \sqrt{\delta} \\
         sign(r^{k}_i) \cdot \left(|r^{k}_i|  - \frac{\delta}{S(r^{k}_i)} \right), & \text{if } | r^{k}_i  | > \sqrt{\delta}
    \end{cases}  \label{prox2} \\
& =     \begin{cases}
        0, & \text{if } | r^{k}_i |  \leq \sqrt{\delta} \\
         r^{k}_i  - \frac{\delta}{r^{k}_i } , & \text{if } | r^{k}_i|  > \sqrt{\delta},
    \end{cases} 
 \label{altz}
\end{align}
where (\ref{prox1}) $\iff$ (\ref{prox2}) because $|r^{k}_i|  - \frac{\delta}{S(r^{k}_i)} > 0 \iff |r^{k}_i| \cdot{S(r^{k}_i)} > {\delta} \iff | r^{k}_i|  > \sqrt{\delta}$. 
Thus, the full iteration framework is established.

However, it is worth noting that before performing the iterations, the input matrix $X$ needs to be preconditioned, i.e.,
\begin{align} \label{pre}
X_p = XL^{-1},
\end{align}
where $L \in \mathbb{R}^{n \times n}$ an upper triangular matrix obtained from the Cholesky decomposition of the symmetric matrix,  i.e. $X^T X = L^T L$. This procedure is aimed at normalizing the singular values of $X_p$ to 1, thereby facilitating faster convergence of the optimization process. 
Meanwhile, it can be verified that $X_p^T X_p = I$ and $ \|X_p \|_2 = 1$.
Finally, we apply a linear transformation to the obtained parameter estimate $\hat{w}_p$: $\hat{w} = L^{-1}\hat{w}_p$. $\hat{w}$ is the desired estimate of the true $w$ because $L$ is full rank and:
\begin{align}
&\min_{w \in \mathbb{R}^n, z \in \mathbb{R}^m} \frac{1}{2}\|y-Xw-z\|_2^2 + \delta \left\|\frac{z}{S(y-Xw)}\right\|_1  \notag \\
= &\min_{w_p \in \mathbb{R}^n, z \in \mathbb{R}^m} \frac{1}{2}\|y-XL^{-1}w_p-z\|_2^2 + \delta \left\|\frac{z}{S(y-XL^{-1}w_p)}\right\|_1 .  \notag
\end{align}
For convenience, we will use $X$ to denote $X_{p}$ in the following. 

The complete procedure is shown in Algorithm \ref{SARMAlgorithm}. In what follows, we further elaborate on the update step in line \ref{update w}. 

When performing gradient descent on $w$, the update rule is given by:
\begin{align}
w^{k+1} &= w^{k} - \alpha \nabla_w H (w^{k}, z^{k}) \notag  \notag \\
&= w^{k} - \left. \alpha \nabla_w (\frac{1}{2}\|y-Xw-z^{k}\|_2^2)\right|_{w = w^{k}} - \left. \alpha \cdot \delta \cdot \sum_{i=1}^{m}\nabla_w \left(\left|\frac{z^{k}_i}{S(x_i^T w-y_i)}\right| \right)\right|_{w = w^{k}}  \notag \\
&= w^{k} - \alpha X^T (Xw^{k}-y - z^{k}) -\left. \alpha \cdot \delta \cdot \sum_{i=1}^{m}\left(\frac{ \left|z^{k}_i\right| \cdot S'(y_i-x_i^Tw)}{S^2(y_i-x_i^Tw)} \right) \cdot x_i \right|_{w = w^{k}}  \notag \\
&= w^{k} - \alpha X^T (r^{k} - z^{k}) - \alpha \cdot \delta \cdot \sum_{i=1}^{m}\left(\frac{ \left|z^{k}_i\right| \cdot S'(y_i - x_i^Tw^{k})}{S^2(y_i - x_i^Tw^{k})} \right) \cdot x_i  \label{w1} \\
&= w^{k} - \alpha X^T (r^{k} - z^{k}) - \alpha \cdot \delta X^T \left(\frac{z^{k}}{S^2(r^{k})} \right). \label{w2} 
\end{align}
(\ref{w1}) $\iff$ (\ref{w2}) due to the fact that 
\begin{align}
& \left|{z^{k}_i}\right| S'(y_i - x_i^Tw^{k}) \label{111}\\
=&      \begin{cases}
        0 \cdot S'(y_i - x_i^Tw^{k}), & \text{if } | r^{k}_i |  \leq \sqrt{\delta}\\
       \left|{z^{k}_i}\right| sign(y_i - x_i^Tw^{k}), & \text{if } | r^{k}_i  | > \sqrt{\delta}
    \end{cases} \label{113} \\
=&      \begin{cases}
        0, & \text{if } | r^{k}_i |  \leq \sqrt{\delta} \\
       {z^{k}_i}, & \text{if } | r^{k}_i  | > \sqrt{\delta}
    \end{cases} \label{114} \\
=&   z^{k}_i,  \notag
\end{align}
and (\ref{113}) $\iff$ (\ref{114}) because $r^{k}_i - \frac{\delta}{ r^{k}_i}$ and $r^{k}_i = y_i - x_i^Tw^{k}$ have the same sign when $| r^{k}_i  | > \sqrt{\delta}$. 

Although Algorithm \ref{SARMAlgorithm} is derived under the smoothed Optimization Problem (\ref{SARMOPT2}), a straightforward verification reveals that the same update rules can be obtained by directly handling the original Problem (\ref{SARMOPT}). The primary purpose of applying the smoothing function (\ref{SMOOTHF}) lies in facilitating the convergence analysis.

Next, we examine Algorithm \ref{SARMAlgorithm}. By further expanding the update rule for $w$, we obtain:
\begin{align}
     w^{k+1} &= w^{k} + \alpha X^T (y - Xw^{k} - z^{k}) - \alpha \cdot \delta \cdot X^T \frac{z^{k}}{S^2(r^{k})}  \notag \\
     &=(I-\alpha X^TX)w^{k} + \alpha X^T(y + z^{k})- \alpha \cdot \delta \cdot X^T \frac{z^{k}}{S^2(r^{k})} \label{e1}\\
     &=(1-\alpha)w^{k} + \alpha(X^TX)^{-1}X^T(y + z^{k})- \alpha \cdot \delta \cdot X^T \frac{z^{k}}{S^2(r^{k})},\label{e2}
\end{align}
where $I$ denotes the identity matrix. (\ref{e1}) $\iff$ (\ref{e2}) holds because the preconditioned matrix $X$ satisfies $X^TX = I$.
$(X^T X)^{-1} X^T (y + z^{k})$ corresponds to the parameter estimate obtained by performing linear regression with $y + z^{k}$ as the response variable.
This suggests that each update of $w$ involves a combination of two effects: a closed-form component obtained via linear regression on the adjusted response variable $y + z^{k}$, and a descent direction contributed by the gradient of the regularization term.
$\alpha$ determines the extent to which the previous parameter estimate is preserved. Typically, we set it to 1.

We further examine the update scheme of $z$. In each iteration, 
The update of $z$ improves the estimation of the parameter $w$ by influencing the adjusted labels:
\begin{align}
    y_i - z^{k}_i &= y_i - \left(r^{k}_i - \frac{\delta}{r^{k}_i} \right) \label{original observation} \\
    &=y_i - \left(y_i - x_i^Tw^{k} - \frac{\delta}{r^{k}_i} \right)  \notag\\
    &=x_i^Tw^{k} + \frac{\delta}{r^{k}_i} \label{fitted value}, 
\end{align}
where $r^{k}_i \geq \sqrt{\delta}$ and $r^{k}_i - \frac{\delta}{r^{k}_i} \geq 0$.
Algorithm \ref{SARMAlgorithm} does not remove samples whose residual exceeds a certain threshold, nor does it set $z^{k+1}_i = r^{k}_i$, which is similar to removing the sample. 
In contrast, Algorithm \ref{SARMAlgorithm} applies varying degrees of adjustment to potential outliers by different residual magnitudes. Specifically, the larger the residual, the closer the adjusted label $y_i + z^{k}_i$ is to the fitted value $x_i^Tw^{k}$; the smaller the residual, the closer the adjusted label $y_i + z^{k}_i$ remains to the original observation $y_i$, which can be observed according to (\ref{fitted value}) and (\ref{original observation}), respectively. 
Samples with large residuals are clearly more likely to be true outliers, whereas those with small residuals are often misidentified non-outliers. Therefore, the above design helps mitigate the influence of the swamping effect.

Next, we discuss the computational advantages of Algorithm \ref{SARMAlgorithm}. 
Cholesky decomposition is computationally efficient, with a complexity of only \(O(\frac{n^3}{3})\) \cite{golub2013matrix}. Consequently, the computational complexity of the preprocessing phase is \(O(mn^2 + \frac{n^3}{3})\). During the iterative process, the computational complexity per iteration is \(O(mn)\), which is much more low.

\begin{algorithm}
\caption{Algorithm For SARM}
\begin{algorithmic}[1]
\Require $y, X_{origin}, \alpha, \delta > 0$
\State Initialization: $k = 0$, $w^{(0)} = 0$, $z^{(0)} = 0$
\State Perform Cholesky decomposition on $X_{origin}^T X_{origin}$, and obtain an upper triangular matrix $L$
\State $X = X_{origin}L^{-1}$
\While{$H(w,z)$ not converged}
    \State $r^{k} = y - Xw^{k}$
    \State (Update $w$):
    $w^{k+1} = w^{k} - \alpha X^T (r^{k} - z^{k}) - \alpha \cdot \delta \cdot X^T \frac{z^{k}}{S^2(r^{k})}$ \label{update w}
    \State (Update $z$):
        $z^{k+1}_i = \begin{cases}
        0, & \text{if } | r^{k}_i |  \leq \sqrt{\delta} \\
        r^{k}_i - \frac{\delta}{r^{k}_i}, & \text{if } | r^{k}_i  | > \sqrt{\delta}
    \end{cases}$
    \State $k := k + 1$
\EndWhile
\State $w^{k} = L^{-1} w^{k}$
\Ensure Final solution $w^{k}$ and $z^{k}$
\end{algorithmic}
\label{SARMAlgorithm}
\end{algorithm}

\section{Theoretical Analysis} \label{Theorem} 

This section addresses two fundamental theoretical aspects: the convergence properties of Algorithm \ref{SARMAlgorithm} and the theoretical error bound associated with the SARM model.
Our convergence proof is established within the framework of non-convex optimization, primarily relying on a quadratic upper bound satisfied during the iterative process of Algorithm \ref{SARMAlgorithm}, as well as the Kurdyka–Łojasiewicz (KL) property.
The error bound proof is based on the Restricted Isometry Property (RIP) condition, along with a set of strong underlying assumptions.

\subsection{Convergence Analysis}

For convenience, we first clarify a notation of the objective function given a $\delta$:
\begin{align}
    H(w, z) &= \frac{1}{2}\|y-Xw-z\|_2^2 + \delta \left\|\frac{z}{S(y-Xw)}\right\|_1 \label{objective}
\end{align}
and assume that $w^{k+1}, w^{k}, z^{k+1}, z^{k}$ are obtained from the $k$-th and $(k+1)$-th iteration of Algorithm \ref{SARMAlgorithm}.
We also suppose that $\alpha$ in Algorithm \ref{SARMAlgorithm} satisfied $0 \leq \alpha \leq \frac{2}{L}$. 
Note that \( L \in \mathbb{R}\) here does not refer to the Lipschitz constant; it will be defined latter in Lemma \ref{Quadratic Upper Bound}.

The convergence analysis is structured in three main steps:

\textbf{STEP1. (Sufficient Decrease Condition).} 
\begin{theorem} \label{Sufficient decrease}
For any $w^{k+1}, w^{k}, z^{k+1}, z^{k}$, there exists $\rho_1 \in \mathbb{R}$ such that:
\begin{align}
    \rho_1 \left(\left\|[w^{k+1}; z^{k+1}] - [w^{k}; z^{k}]\right\|^2_2 \right) \leq H \left( w^{k+1}, z^{k+1} \right) - H\left(w^{k}, z^{k}\right), \quad k=1,2, \cdots   \notag
\end{align}
\end{theorem}

\textbf{STEP2. (A Subgradient Lower Bound).}
\begin{theorem} \label{Subgradient Lower Bound}
For any $w^{k+1}, w^{k}, z^{k+1}, z^{k}$, there exists $ A^{k+1} = \nabla_w H(w^{k+1}, z^{k+1}) $, $ B^{k+1} \in \partial_z H(w^{k+1}, z^{k+1}) $ and $\rho_2 \in \mathbb{R}$ such that:
\begin{align}
\left\|[A^{k+1}; B^{k+1}]\right\|_2 \leq \rho_2 \left\|[w^{k+1}; z^{k+1}] - [w^{k}; z^{k}] \right\|_2 , \quad k=1,2, \cdots  \notag
\end{align}
\end{theorem}

\textbf{STEP3. (Apply the Kurdyka-Łojasiewicz (KL) Property).} 

Assume that the sequence \(\{[w^{k}; z^{k}]\} \subset \mathbb{R}^{m+n}\) generated by Algorithm \ref{SARMAlgorithm} is bounded. 
Then proof the objective function \(H(w, z)\) is a KL function. 
Then, following the standard KL framework \cite{bolte2014proximal, attouch2010proximal,attouch2013convergence,liu2020optimization}, 
it can be shown that \(\{[w^{k}; z^{k}]\}\) satisfies a finite length property, and therefore, it is a Cauchy sequence.
\begin{theorem} \label{Apply KL Property}
For any $w^{k+1}, w^{k}, z^{k+1}, z^{k}$, if \(\{[w^{k}; z^{k}]\} \subset \mathbb{R}^{m+n}\) obtained by Algorithm \ref{SARMAlgorithm} is bounded, 
then the following  results holds:
\begin{enumerate}
    \item The sequence \(\{[w^{k}; z^{k}]\}\) has finite length, i.e.
    \begin{align}
    \sum_{k=1}^{\infty} \left\| [w^{k+1}; z^{k+1}] - [w^{k}; z^{k}] \right\|_2 < + \infty  \notag
    \end{align}
    \item \(\{[w^{k}; z^{k}]\}\)  converges to a critical point \(\{[w^{*}; z^{*}]\}\) of $H(w, z)$.
\end{enumerate}
\end{theorem}
The complete proof is provided in the appendix.

\subsection{Error Bound Analysis}

In this subsection, we provide a theoretical error bound for SARM. Prior to that, we make several simple modifications to the outlier representation in (\ref{attackted}). 
We denote the true \( w \) and \( z \) as \( w^{true} \) and \( z^{true} \), and still assume that the samples follow the relationship:
\begin{equation} \label{modified attacked}
y = X w^{true}  + e + z^{true}.
\end{equation}
In addition, we assume that the random noise \( e \) and the outlier variable \( z \) satisfy:
\begin{equation} \label{condition1}
 \{ i| e_i \neq 0 \} \cap \{ i| z_i^{true} \neq 0 \} = \emptyset, 
\end{equation}
which is imposed to ensure a clear separation between the informative (clean) samples and the outliers. 
Moreover, since the observation errors associated with outliers are typically large, we assume that:
\begin{align} \label{condition2}
\min_{w \in \mathbb{R}^n} \frac{1}{2}\|y-Xw\|_2^2 \geq \varepsilon.
\end{align}
Condition (\ref{condition2}) is used to facilitate the subsequent derivation of the error bound, and it is typically a natural assumption in practice.
$ \| e\|_2^2$ is denoted as \( \varepsilon \) and the sparsity of $z$, i.e. $|\{ i| z_i^{true} \neq 0 \}| $ is denoted as $ k^{true}$.

We also present a definition that is crucial in sparse robust models: the Restricted Isometry Property (R.I.P.).
\begin{definition} \label{RIP}
    \textbf{(Restricted Isometry Property)} The Restricted Isometry Property of order \( s \) with constant \( \mu_s \in (0,1) \) for a preconditoned matrix \( X \in \mathbb{R}^{m \times n} \) states that for all \( s \)-sparse vectors \( [w;z] \in \mathbb{R}^{m+n} \), the following inequality holds:
\begin{align}
(1 - \mu_s) \| [w;z] \|_2^2 \leq \| [I,X^T] [w;z] \|_2^2 \leq (1 + \mu_s) \| [w;z] \|_2^2.
\end{align}
\end{definition}

Based on these preliminaries, we have the following theorem.
\begin{theorem} \label{upper bound}
Suppose that for any \( \delta > 0 \), the optimization problem (\ref{SARMOPT2}) has a unique global solution \( [ w^*_{\delta}; z^*_{\delta} ] \) and there exists a sufficiently large compact set \( U \subset \mathbb{R}^{m+n} \) such that $ [ w^*_{\delta}; z^*_{\delta} ] \in U, \forall \delta>0$.
Then there exists a a lower bound \( \delta^{lb} >0 \), such that when \( \delta \geq \delta^{lb} \),  $[w^*_{\delta};z^*_{\delta}]$ satisfies \( \frac{1}{2}\|y-Xw^*_{\delta}-z^*_{\delta}\|_2^2 \geq \varepsilon \). Moreover, for any $\delta \geq \delta^{lb}$, the following inequality holds:
\begin{align}
    \left\|w^* -  w^{true}\right\|_2 \leq \min_{\lambda \in (1, +\infty)} \frac{ \left( \lambda - \frac{1}{\lambda} + 1 \right ) \cdot \sqrt{m\delta} + \varepsilon}{\sqrt{1- \mu_{s} }}, \label{error upper bound expression}
\end{align}
where  
\begin{align}
s = \left(\frac{\lambda^2}{ \lambda^2 - 1 } + 1 \right) k^{true} + n.
\end{align}
\end{theorem}

\begin{remark}
    The above theorem requires several remarks. First, the assumption that \(U\) is a compact set is made to satisfy the conditions required for applying Berge's theorem. 
    Second, assuming global optimality is a standard practice in the theoretical analysis of non-convex optimization. 
    Third, if the optimal solution is not unique, it becomes theoretically difficult to quantify the distance between the estimated parameters \( [ w^*_{\delta}; z^*_{\delta} ] \) and the true parameters \( [w^{true}, z^{true} ]\).
    Finally, it should be acknowledged that, because of the theoretical difficulties associated with non-convex optimization, the assumptions made to derive the error bound are considerably stronger than those used to prove convergence.
\end{remark}

\section{Simulation Experiments} \label{Simulation}

Simulation experiments are conducted in this section to evaluate and compare the performance of SARM against several state-of-the-art robust regression approaches, with a focus on both robustness and computational efficiency.
Our study concentrates on the scenario where the feature dimension is less than the sample size ($n<m$).
In the experiments, we assume the standard deviation $\sigma$ of the inlier noise is known as prior information, which is a common assumption in robust regression research \cite{liu2018robust, papageorgiou2015robust}.

\subsection{Comparative Models} \label{Comparative Models}
The main comparative models considered include:
\begin{itemize} 
    \item [(1)] \textbf{Ideal}: Ideal solution where $z$ is assumed to be exactly known,i.e.:
    \begin{align}
        w_{Ideal} = \underset{w \in \mathbb{R}^n}{\arg\min}\| y - Xw - z\|_2^2 \notag
    \end{align}
    \item [(2)] \textbf{Oracle}: Oracle solution \cite{candes2008highly} under the oracle knowledge of the outlier support $\mathcal{S} = supp(z) = \left\{  i| z_i \neq 0\right\}$:
    \begin{align}
        w_{Oracle} = \underset{w \in \mathbb{R}^n}{\arg\min}\| y_{\mathcal{S}^C} - X_{\mathcal{S}^C}w\|_2^2, \notag
    \end{align}
    where $\mathcal{S}^C$ denotes the complement of $\mathcal{S}$.
    \item [(3)] \textbf{MLR}: Multiple Linear Regression (MLR) solved via Ordinary Least Squares (OLS):
    \begin{align}
        w_{MLR} = \underset{w \in \mathbb{R}^n}{\arg\min}\| y - Xw\|_2^2, \notag
    \end{align}
    \item [(4)] \textbf{IRLS}: Generalized M-estimators with Bisquare weighting function \cite{maronna2019robust,beaton1974fitting}.  It is solved  Iteratively Reweighted Least Squares (IRLS) \cite{holland1977robust}, and we set $c = 4.685$, which is the default value.
    \item [(5)] \textbf{L1}:  Least Absolute Deviation (LAD, $\ell_1$ estimator):
    \begin{align}
        w_{\ell_1} = \underset{w \in \mathbb{R}^n}{\arg\min}\| y - Xw\|_1. \notag
    \end{align}
    It can be expressed as a Linear Programming (LP) by introducing auxiliary variables $v_1, v_2 \in \mathbb{R}^m$:
    \begin{equation}
    \begin{aligned}
        &\underset{w \in \mathbb{R}^n,v_1, v_2 \in \mathbb{R}^m}{\arg\min} \quad v_1 + v_2 \\
         s.t. y - &Xw = v_1-v_2, v_1, v_2 \geq 0.
    \end{aligned}
    \end{equation}
    \item [(6)] \textbf{TRLM}: Truncated Loss Regression via Majorization-minimization algorithm (TLRM) \cite{huang2024large}. The TRLM algorithm applied to linear regression has a similar form as Algorithm \ref{AROSI}. $Loss$ is specified as $\|\cdot\|_2^2$. \(\alpha\) is set to $5\sigma$.
    \item [(7)] \textbf{SOCP}: Second-Order Cone Programming (SOCP):
    \begin{equation}
    \begin{aligned}
        \underset{w \in \mathbb{R}^n,z \in \mathbb{R}^m}{\arg\min}\| z\|_1, \quad s.t. \| y-Xw-z\|_2 \leq \sqrt{m}\sigma
    \end{aligned}
    \end{equation}
    It can be expressed in the form mentioned in \cite{papageorgiou2015robust} Section 2.2. We solve it by ECOS (the Embedded Conic
Solver) \cite{domahidi2013ecos}.
    \item [(8)] \textbf{L1reg}: $\ell_1$ regularization problem \cite{jin2010algorithms} (equivalent to the Huber regression model):
    \begin{align}
        w_{\ell_1reg} = \underset{w \in \mathbb{R}^n}{\arg\min}\| y - Xw - z\|_2^2 + \lambda\|z \|_1, \notag
    \end{align}
    where $\lambda = \frac{\sigma\sqrt{2log(m)}}{3}$ (consistent with \cite{jin2010algorithms}).
        we solve it by transforming it into:
    \begin{equation}
    \begin{aligned}
        \underset{w \in \mathbb{R}^n,z_1 \in \mathbb{R}^m, z_2 \in \mathbb{R}^m}{\arg\min} \| y - &Xw - z_1 + z_2\|_2^2 + \lambda z_1 + \lambda z_2, \\
        s.t. &z_1 \geq 0 ,z_2 \geq 0.
    \end{aligned}
    \end{equation}
    We accelerate the optimization by replacing the solver from ADMM (used in \cite{papageorgiou2015robust}) to ECOS \cite{domahidi2013ecos}.
    \item [(9)] \textbf{GARD}: Greedy Algorithm for Robust Denoising (GARD) \cite{papageorgiou2015robust}. It solves
    \begin{align}
    \min_{w \in \mathbb{R}^n,z \in \mathbb{R}^m} \|z\|_0, \quad s.t. \quad \|y -[X,I][w;z]\|_2 \leq \sqrt{m}\sigma,
    \end{align}
    by a greedy Algorithm, which identifies and eliminates potential outliers one by one using Orthogonal Matching Pursuit (OMP).
    \item [(10)] \textbf{IPOD}:  Thresholding-based Iterative Procedure for Outlier Detection (IPOD, or ${\Theta}$-IPOD) \cite{she2011outlier}. It enhances robustness through alternating iterations and hard thresholding, relying on a preliminary robust regression to provide an initial estimate. Following \cite{liu2018robust}, we use $\ell_1$ estimator as the initial estimate.
    \item [(11)] \textbf{AROSI}:  The Algorithm for Robust Outlier Support Identification (AROSI) \cite{liu2018robust}. It solves (\ref{L01}) by setting the loss function as \( \| \cdot \|_1 \) in Algorithm \ref{AROSI}. \(\alpha\) is set to $5\sigma$. The solution method is the same as for L1.
\end{itemize}

For SARM, the parameter \(\delta\) is set to $6\sigma^2$.
The convergence threshold for all models is set to 1e-6.
The theoretical complexities of SARM and the aforementioned algorithms are summarized in Table \ref{tab:timecomplexity}. In general, SARM demonstrates some theoretical advantages. However, since the number of iterations in practice often varies, the specific computational efficiency of the model remains to be experimentally validated.

\begin{table}[htbp]
\centering
\caption{Theoretical Computational Complexity Comparison}
\setlength{\tabcolsep}{12pt} 
\renewcommand{\arraystretch}{1.3} 
\begin{tabular}{|c|c|c|}
\hline
\textbf{Algorithm} & \textbf{Preprocessing Complexity} & \textbf{Iteration Complexity} \\ 
\hline
MLR    & None                               & $\mathcal{O}(m n^2 + n^3)$ \\ 
\hline
IRLS   & None                               & $\mathcal{O}(m n^2 + n^3)$ per step \\  
\hline
L1     & None                               & $\mathcal{O}(m^3)$ per step \cite{liu2018robust, nesterov2018lectures} \\ 
\hline
TRLM   & None                               & $\mathcal{O}(m n^2 + n^3)$ per step \\ 
\hline
SOCP   & None                               & $\mathcal{O}((m + n)^{2.5} m)$ \cite{papageorgiou2015robust} \\  
\hline
L1reg  & None                               & $\mathcal{O}(m^3)$ per step \cite{liu2018robust} \\ 
\hline
GARD   & None                               & $\mathcal{O}\left(\frac{n^3}{3} + \frac{K^3}{2} + (m + 3K) n^2 + 3 K m n \right)$ \cite{papageorgiou2015robust} \\ 
\hline
IPOD   & $\mathcal{O}(m^3 + m n^2)$         & $\mathcal{O}(m n)$ per step \cite{liu2018robust, she2011outlier} \\ 
\hline
AROSI  & None                               & $\mathcal{O}(m^3)$ per step \cite{liu2018robust}  \\ 
\hline
SARM   & $\mathcal{O}\left(\frac{n^3}{3} + m n^2\right)$ & $\mathcal{O}(m n)$ per step \\ 
\hline
\end{tabular}
\begin{tablenotes}
    \footnotesize
    \item The complexities of L1 and L1reg are based on Interior Point Method (IPM) \cite{nesterov2018lectures}. Due to sparsity, the actual computational complexity is lower than the theoretical value.
    \item For MLR, GARD, and SOCP, total complexity is shown; for other methods, complexity depends on iteration count.
    \item For GARD, \(K\) denotes the number of iterations, i.e., the number of outliers identified.
\end{tablenotes}
\label{tab:timecomplexity}
\end{table}

\subsection{Implementation Details and Experimental Setup}

All experiments are conducted on a server equipped with Intel(R) Xeon(R) Silver 4410Y CPUs, 128.0 GB usable RAM. We use 50-core multithreaded computing for the implementation of the experiments.

The simulation experiments adopt settings drawn from various existing studies to demonstrate the broad applicability of SARM.

Our simulation experimental setup is as follows:
\begin{enumerate}
  \item \textbf{Step 1}: Generate a random matrix \(X \in \mathbb{R}^{m \times n} \).
  \item \textbf{Step 2}: Generate \(w \in \mathbb{R}^n\) by i.i.d. \( \mathcal{N}(0, \sigma_w^2) \). Compute \(Xw\).
  \item \textbf{Step 3}: Generate inlier noise \(e = (e_1, \ldots, e_m)\) and add $e$ to \(Xw\), with \(e_i\) i.i.d. \(\mathcal{N}(0, \sigma^2)\).
  \item \textbf{Step 4}: Select a proportion \( p \) to define the number of outlier samples $k= round(p \cdot m)$;
  \item \textbf{Step 5}: Randomly and uniformly select \(k\) outlier locations and apply corruptions. Obtain $y$;
  \item \textbf{Step 6}: Estimate \(w\) utilizing different robust models.
\end{enumerate}

Following the above procedure, we designed a total of seven experimental setups.
\begin{enumerate}
  \item \textbf{Type 1}: \(m = 512\), \(\sigma_w = 1\), \(X_{ij} \in \mathcal{N}(0, 1)\) and i.i.d., \(\sigma = \text{median}(|Xw|)/16\) \cite{candes2008highly}. Corruptions are obtained from $0.5 \times \mathcal{N}(12\sigma, (4\sigma)^2) + 0.5 \times \mathcal{N}(-12\sigma, (4\sigma)^2)$ \cite{liu2018robust}.  \(n \in \{16, 32, 64, 128\}\).
  \item \textbf{Type 2}:  \(m = 600\), \(\sigma_w = 5\), \(X_{ij} \in U(0, 1)\) and i.i.d., \(\sigma =1\). Corruptions are obtained from $\{-25, 25\}$ \cite{papageorgiou2015robust}.  \(n \in \{50, 100, 170\}\).
  \item \textbf{Type 3}: Identical to Type1, except that corruptions are obtained from $ \mathcal{N}(12\sigma, (4\sigma)^2)$.
  \item \textbf{Type 4}: Identical to Type2, except that corruptions are obtained from $\{25\}$.
  \item \textbf{Type 5}: \(m = 512\), \(\sigma_w = 1\), \(X_{ij} \in \mathcal{N}(0, 1)\) and i.i.d., \(\sigma = \text{median}(|Xw|)/16\). Corruptions are obtained from $\mathcal{N}(0, (\kappa\sigma)^2)$.  \(n=64\).  \(\kappa \in \{8, 12, 16\}\) \cite{liu2018robust}.
\end{enumerate}
We primarily focus on Type 1, Type 2, and Type 5. Type 3 and Type 4 correspond to cases with unimodal outliers. As these cases differ only slightly from Type 1 and Type 2, we present their results in the appendix for brevity.
For the estimated parameter $\hat{w}$, we measure its Relative \(\ell_2\)-Error \cite{liu2018robust, elad2010sparse}
$
\frac{\|\hat{w} - w\|_2}{\|w\|_2},
$
where $w$ is the true parameter. We repeat the experiments 200 times and report the mean Relative \(\ell_2\)-Error to reflect the robustness.

\subsection{Robustness Performance Comparision}

We first evaluate and compare the performance of SARM against the robust models discussed in Section \ref{Comparative Models}, under the experimental setting of Type 1. The comparison results are shown in Fig. \ref{Type1}.
As shown, the breakdown point of SARM is generally higher than that of the other competing models.
When the sample dimensionality is relatively low, the robustness advantage of SARM becomes more pronounced. As the dimensionality increases, the robustness performance of SARM gradually approaches that of AROSI. It is also worth noting that the robustness of all models tends to degrade significantly as the dimensionality increases.

Next, we analyze the possible reasons behind the performance differences. Our analysis mainly considers TLRM, AROSI, and SARM. Since TLRM uses the $\ell_2$ norm loss function in its initial estimation stage, it is particularly vulnerable to the swamping effect, given that the $\ell_2$ loss is considerably more sensitive to outliers than the $\ell_1$ loss. Thanks to the enhanced robustness of the $\ell_1$ loss, AROSI's initial estimation does not deviate excessively from the true value, thereby mitigating the risk of incorrectly discarding inliers. Nevertheless, when the corruption rate becomes sufficiently high, the robustness of the $\ell_1$ loss is no longer adequate for AROSI to resist the swamping effect. Unlike AROSI, SARM does not rely heavily on the robustness of the initial estimate. Its initial parameter estimate is essentially the same as that of TLRM. However, through the self-adjusting mechanism embedded in the iterative procedure of Algorithm \ref{SARMAlgorithm}, SARM demonstrates significantly improved robustness compared to both TLRM and AROSI. In contrast to truncated $\ell_0$-regularized methods such as TLRM and AROSI, SARM considers all data points globally, which allows it to avoid the drawbacks associated with the swamping effect.

\begin{figure}[htbp]
  \centering

  \begin{subfigure}[b]{0.49\textwidth}
    \centering
    \includegraphics[width=\linewidth]{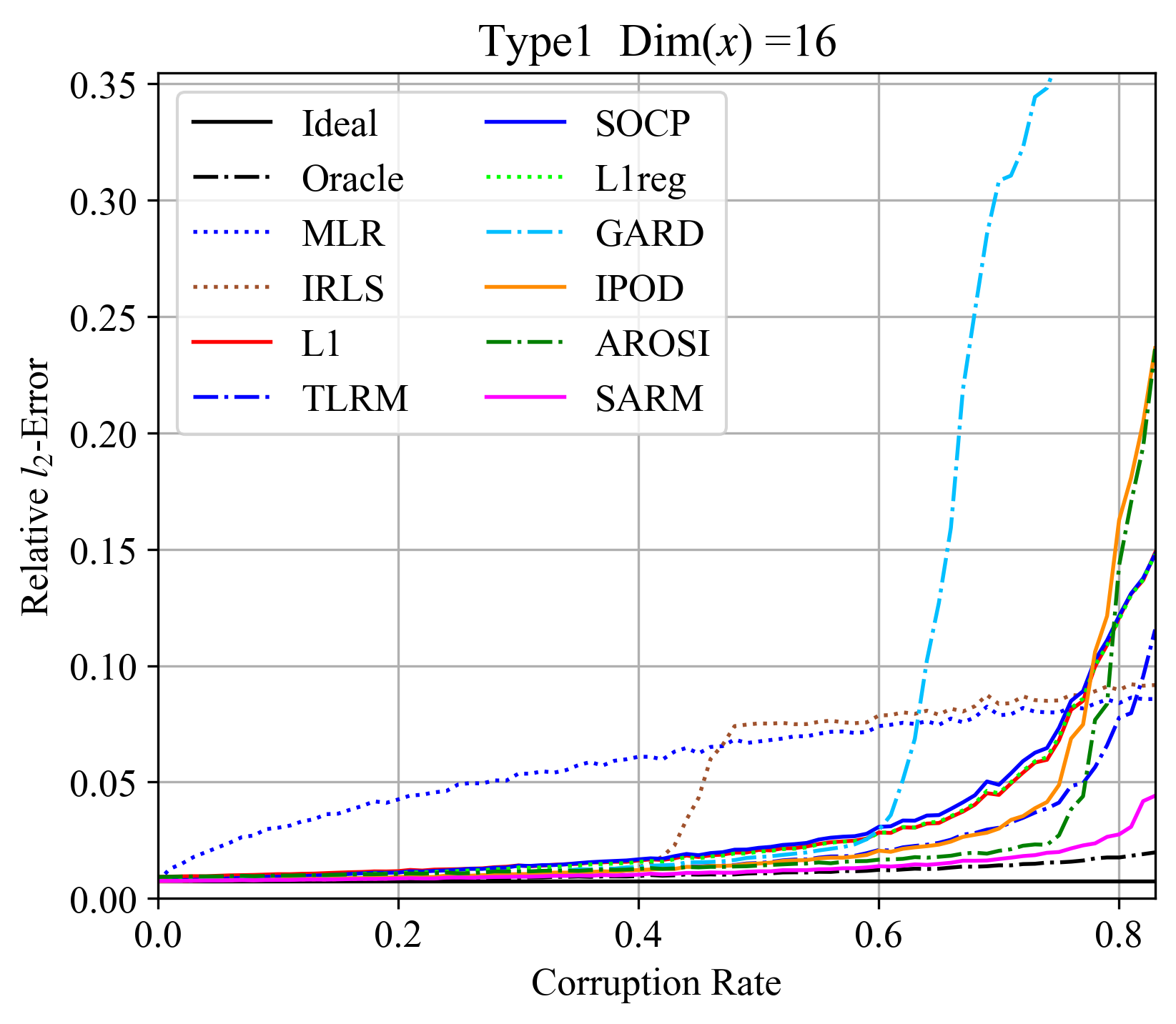}
  \end{subfigure}
  \hfill
  \begin{subfigure}[b]{0.49\textwidth}
    \centering
    \includegraphics[width=\linewidth]{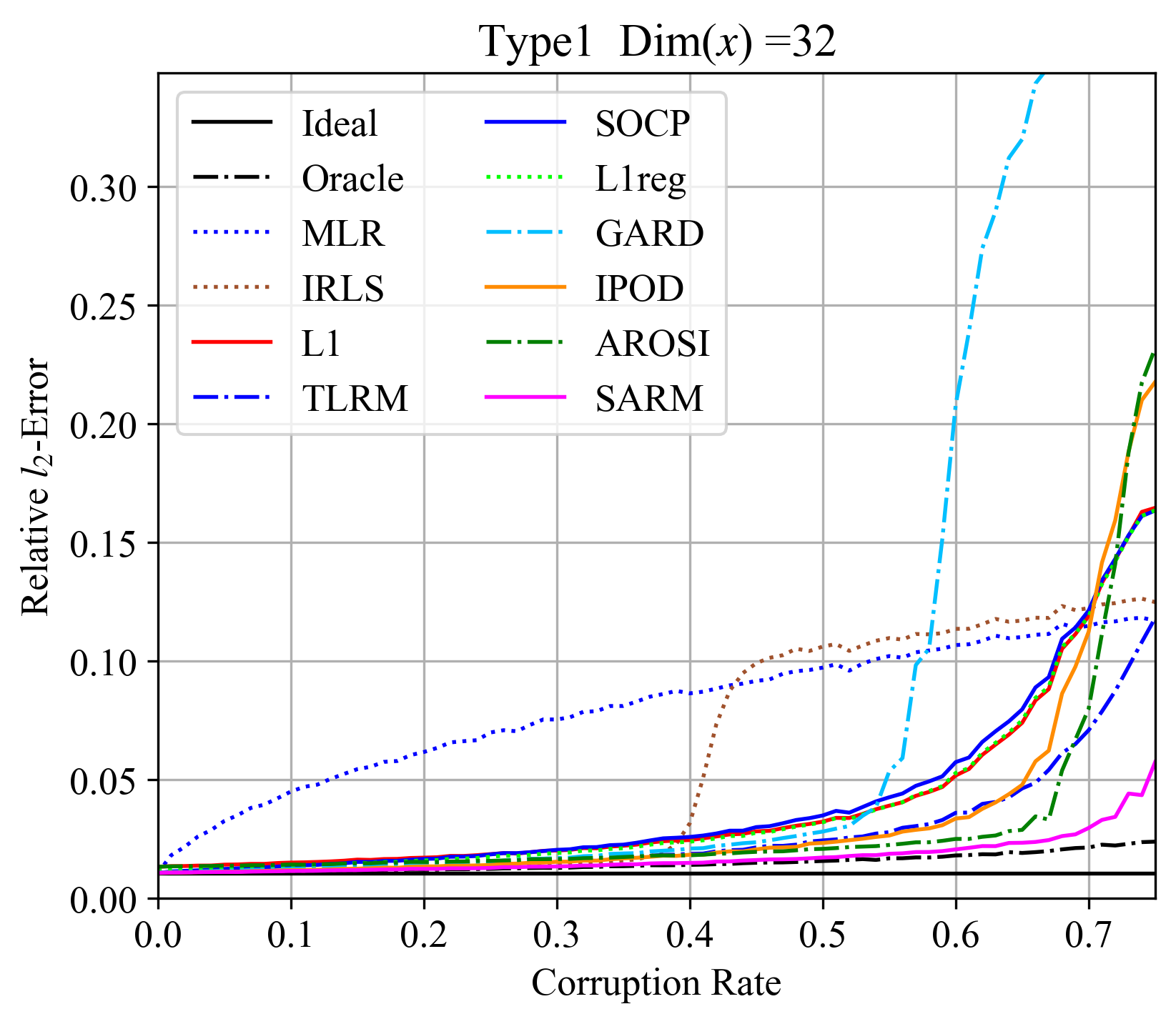}
  \end{subfigure}

  \vspace{1em}

  \begin{subfigure}[b]{0.49\textwidth}
    \centering
    \includegraphics[width=\linewidth]{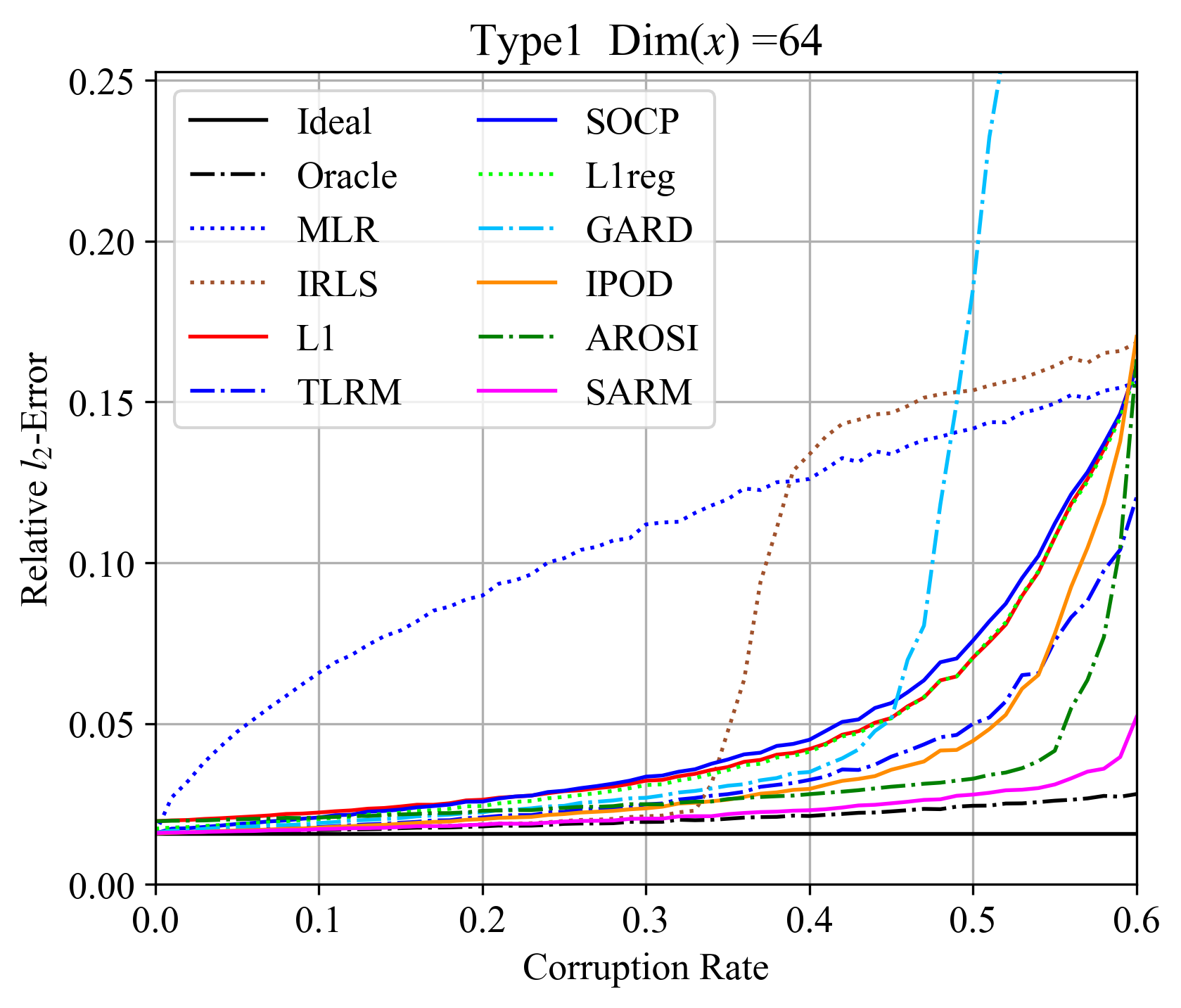}
  \end{subfigure}
  \hfill
  \begin{subfigure}[b]{0.48\textwidth}
    \centering
    \includegraphics[width=\linewidth]{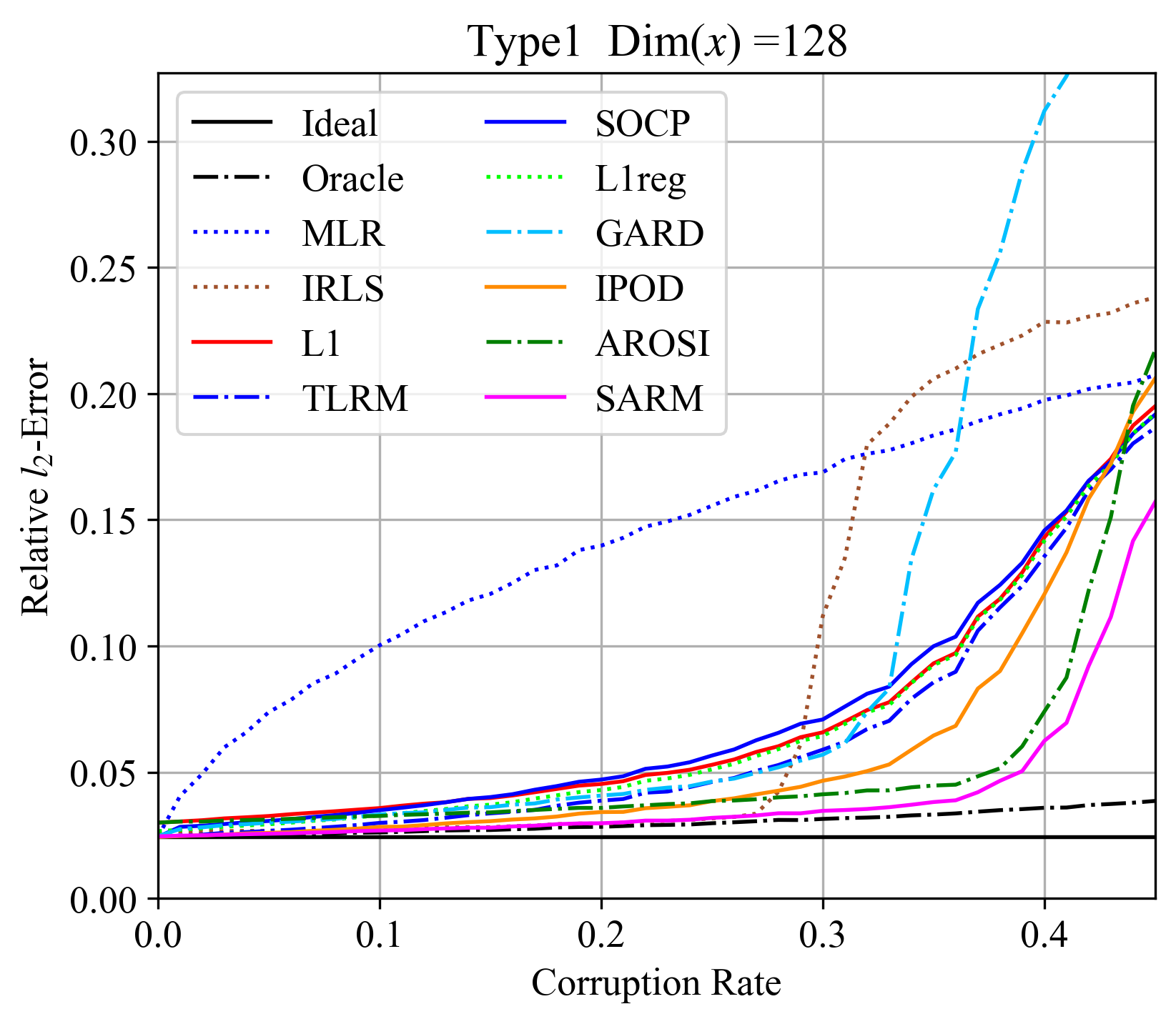}
  \end{subfigure}

  \caption{Empirical Evaluation of Robustness: Type1}
  \label{Type1}
\end{figure}

\subsection{Convergence Analysis}

In the following, we investigate the convergence of Algorithm \ref{SARMAlgorithm}, which serves as the solver for SARM. While its convergence has been rigorously proven in Theorem \ref{Apply KL Property} of Section \ref{Theorem}, we provide additional empirical evidence through numerical experiments to corroborate our theoretical analysis.

We continue to use Type 1 as a representative case, set Dim($x$)=$64$, and gradually increase the corruption rate $p$ from $0.0$ to $0.6$. The Log Scaled experimental results are shown in Fig. \ref{Convergence Analysis}.
We first observe that $H^{k} - H^{k+1} > 0$, which confirms the monotonic decrease of the sequence ${H^k}$ in accordance with (\ref{monotonicity}). Furthermore, the trajectories of $\log(H^{k} - H^{k+1})$, $\log\| \mathbf{w}^{k+1} - \mathbf{w}^{k} \|_2$, and $\log\| \mathbf{z}^{k+1} - \mathbf{z}^{k} \|_2$ are nearly indistinguishable, indicating that these quantities exhibit similar decay behavior and are of comparable magnitude.
This serves as mutual validation of Theorem \ref{Subgradient Lower Bound}.
After sufficient iterations, Algorithm \ref{SARMAlgorithm} demonstrates Q-linear convergence behavior, suggesting a rapid overall convergence. Additionally, it is observed that the iteration count of Algorithm \ref{SARMAlgorithm} is affected by the outlier fraction $p$, with larger proportions of outliers necessitating more iterations.

\begin{figure}[htbp]
  \centering
  \begin{subfigure}[b]{0.34\textwidth}
    \centering
    \includegraphics[width=\linewidth]{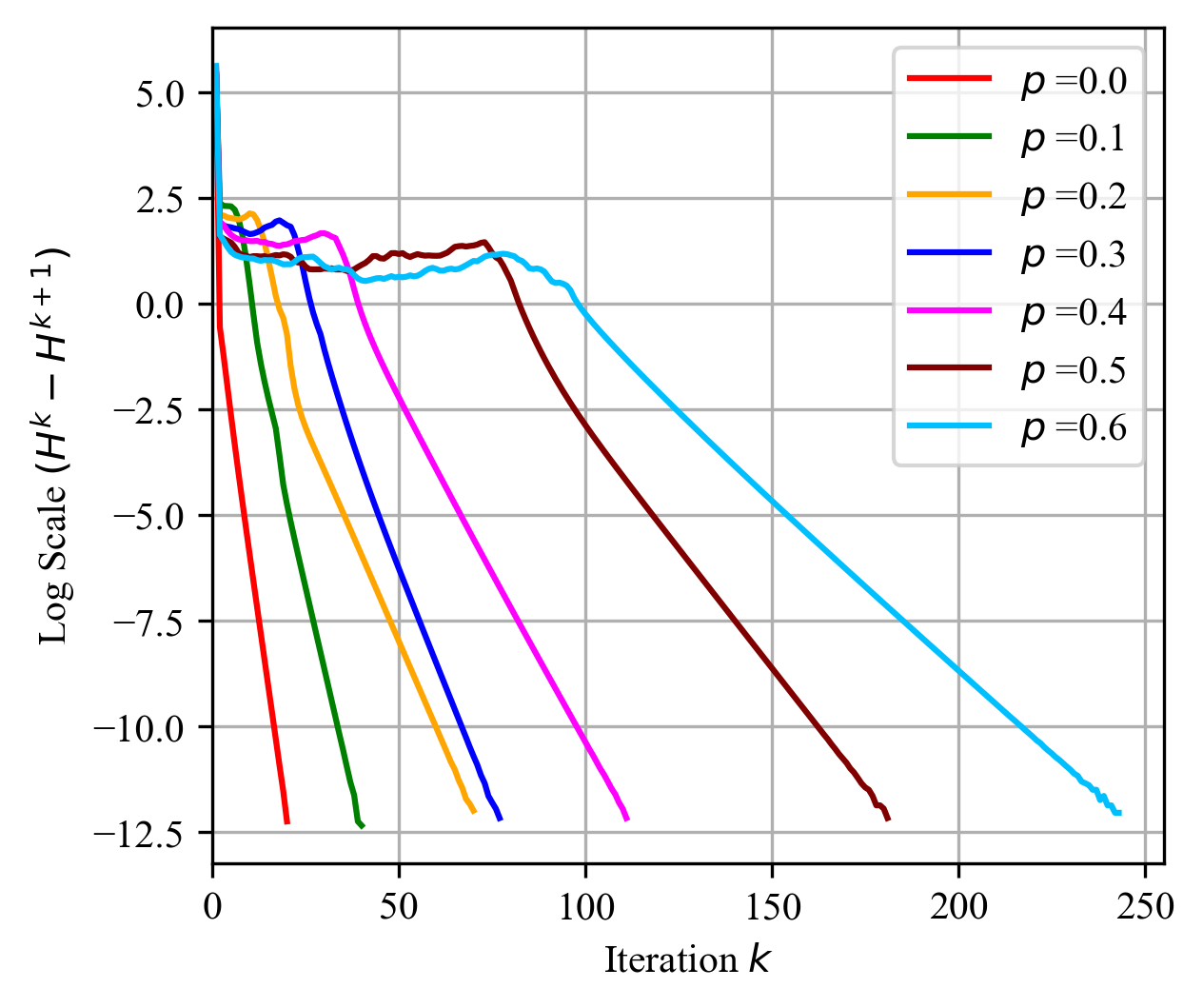}
    \caption{Log Scale $H^{k} - H^{k+1}$}
    \label{Convergence Analysis1}
  \end{subfigure}
  \hfill
  \begin{subfigure}[b]{0.32\textwidth}
    \centering
    \includegraphics[width=\linewidth]{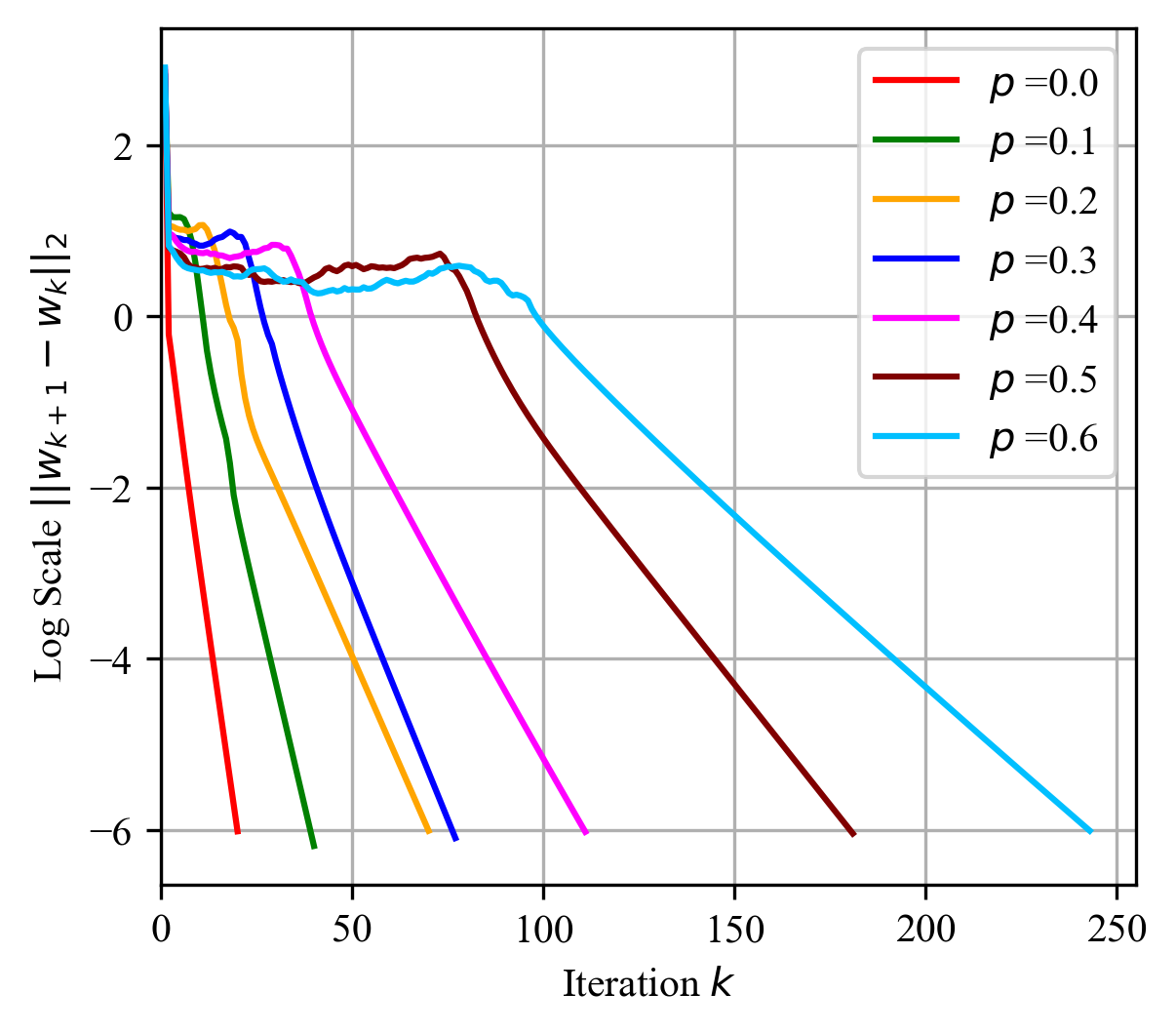}
    \caption{Log Scale $\| w^{k+1} - w^{k} \|_2$}
    \label{Convergence Analysis2}
  \end{subfigure}
  \hfill
  \begin{subfigure}[b]{0.32\textwidth}
    \centering
    \includegraphics[width=\linewidth]{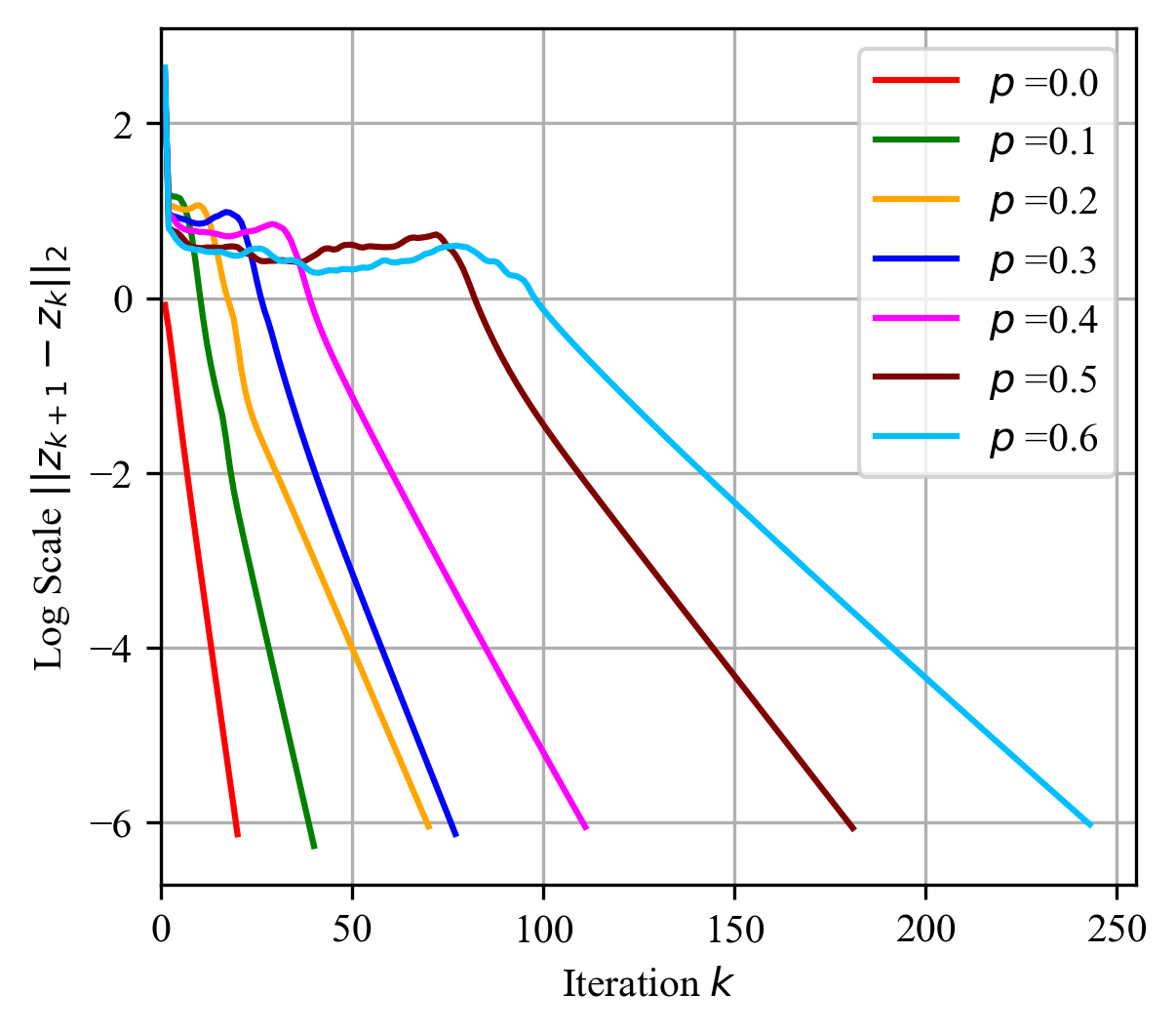}
    \caption{Log Scale $\| z^{k+1} - z^{k} \|_2$}
    \label{Convergence Analysis3}
  \end{subfigure}
  \caption{Convergence Analysis (Type1, Dim($x$)=$64$)}
  \label{Convergence Analysis}
\end{figure}

\subsection{Computational Efficiency Comparision}

We proceed to compare the computational efficiency. While Table \ref{tab:timecomplexity} provides the theoretical computational complexities of various models as an initial reference, practical performance is influenced by algorithmic properties and hardware-level factors, making theoretical complexity insufficient to fully capture real-world efficiency. Hence, we present empirical comparisons based on numerical experiments.

We first fix the sample size and adopt the Type 1 experimental setup with Dim($x$)=$128$. The experimental results are shown in Fig. \ref{Computational Efficiency}.
It can be seen that GARD achieves optimal performance at low outlier ratios (the number of outliers $K \ll m$), but its computational time escalates sharply with increasing outlier fractions.
In contrast, TLRM and SARM consistently demonstrate superior computational efficiency across varying corruption rates. 
However, as illustrated in Fig. \ref{Type1}, SARM demonstrates substantially stronger robustness compared to TLRM under similar computational efficiency. It is also noteworthy that as the corruption rate increases, both SARM and TLRM exhibit a rise in computational time. This phenomenon is primarily attributed to the increased number of iterations, which is consistent with the trend observed in Fig. \ref{Convergence Analysis}.
Among other competing methods, AROSI, whose robustness is the closest to that of SARM, exhibits significantly lower computational efficiency. This is a common issue among models involving the $\ell_1$-loss function and linear programming. Although certain commercial solvers or specially designed optimization algorithms may improve the computational efficiency of such methods, the results presented in Fig. \ref{Computational Efficiency} already provide a representative comparison in terms of efficiency. 

Additionally, another noteworthy advantage of SARM deserves attention.
From the algorithmic structure of SARM, its primary computational bottleneck lies in two matrix-vector multiplications. This implies that SARM is particularly amenable to acceleration via parallel matrix operations, especially when tackling large-scale problems. To evaluate its computational efficiency in such settings, we conduct experiments based on the Type 1 setup while maintaining a fixed sample-to-feature ratio \( \frac{m}{n} =10\), a corruption rate \( p = 0.25 \), and input dimension Dim($x$)=$64$. We gradually increase the sample size to assess the scalability. The results are reported in Fig. \ref{Computational Efficiency}.
We only retain a subset of models that are most representative in terms of computational efficiency.
SARM(GPU) refers to the algorithm executed with GPU.
SARM (CPU) refers to the algorithm executed with CPU (consistent with the other experiments). As shown, even when using CPU, SARM demonstrates high computational efficiency when handling large-scale data. Upon switching to GPU computation, the efficiency of SARM improves significantly. However, when the data scale is relatively small, the CPU implementation becomes less efficient compared to the GPU. This is primarily due to the communication overhead between the CPU and GPU dominating the total runtime.

\begin{figure}[htbp]
  \centering
  \begin{subfigure}[b]{0.535\textwidth}
    \centering
    \includegraphics[width=\linewidth]{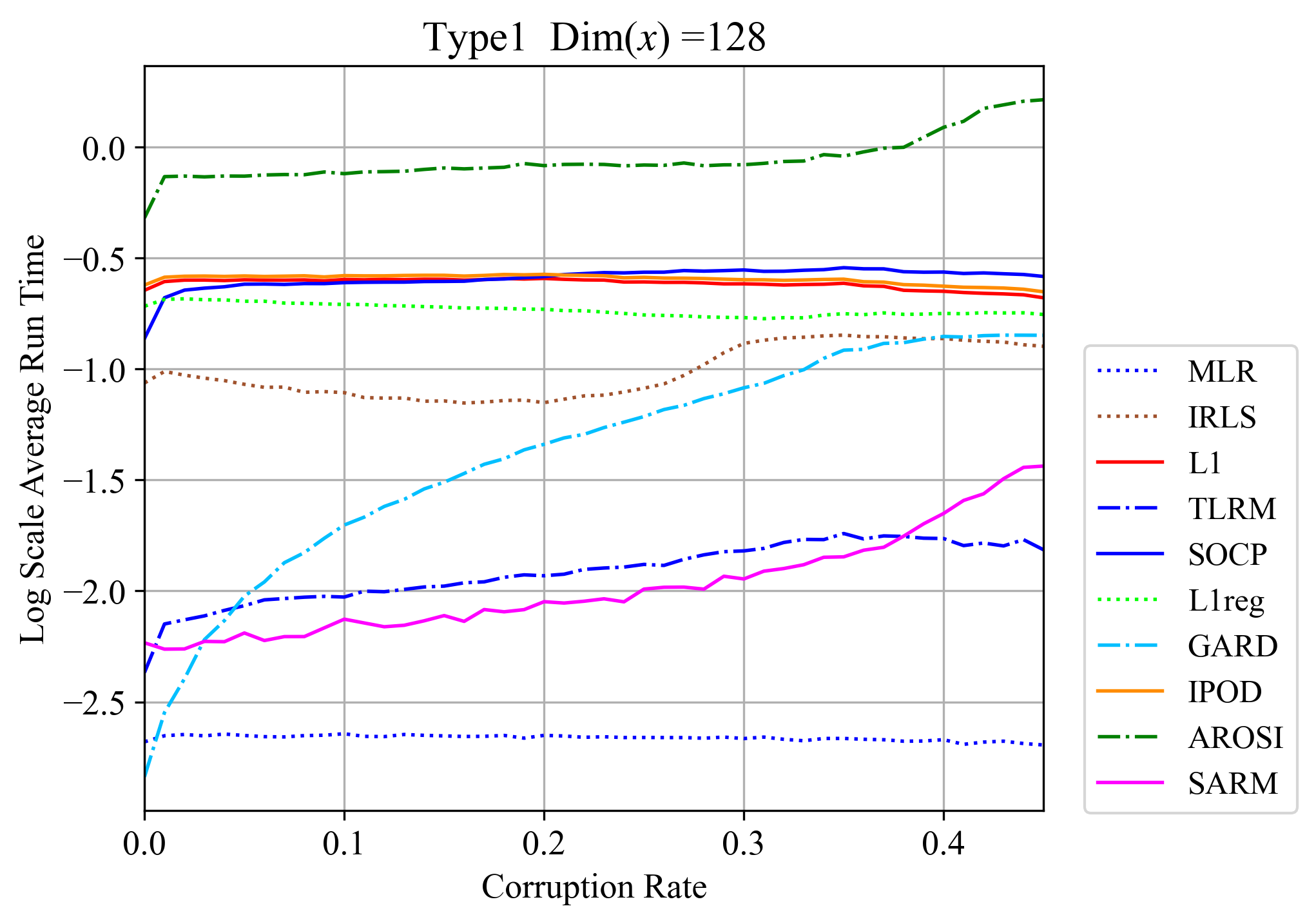}
    \caption{Log scale average run time vs. Corruption rate.}
  \end{subfigure}
  \hfill
  \begin{subfigure}[b]{0.45\textwidth}
    \centering
    \includegraphics[width=\linewidth]{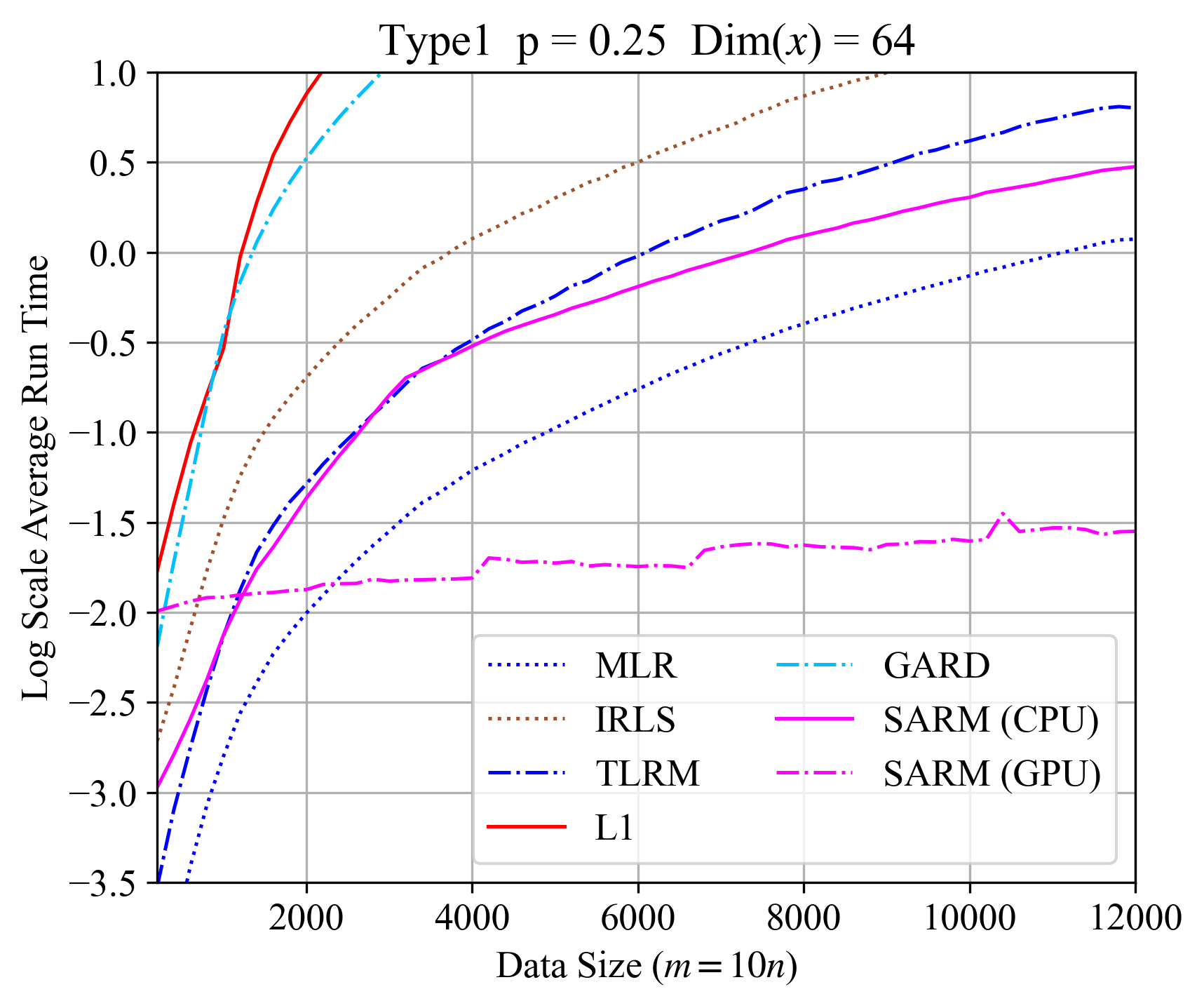}
    \caption{Log scale average run time vs. data scale.}
  \end{subfigure}
  \caption{Empirical Evaluation of Computational Efficiency}
  \label{Computational Efficiency}
\end{figure}

\subsection{Parameter Sensitivity Analysis}

We next investigate the sensitivity of Algorithm \ref{SARMAlgorithm} with respect to the parameter $\delta$. When \(\delta\) is small, the overall accuracy of parameter estimation tends to degrade. As \(\delta\) increases, the estimation performance gradually improves. This behavior is consistent with the theoretical requirement in Theorem \ref{upper bound}, which specifies that \(\delta\) should exceed a certain threshold to ensure accurate estimation. 
Intuitively, a sufficiently large \(\delta\) helps reduce the risk of incorrectly identifying non-outliers as outliers. 

Empirically, we observe that when \(\delta\) is within the range of $3\sigma^2-6\sigma^2$, SARM maintains relatively stable performance. However, when \(\delta\) becomes too large (e.g., \(\delta = 10\sigma^2\)), the robustness of SARM to outliers begins to deteriorate. This observation is in line with Theorem \ref{upper bound} and the bound specified by the inequality \eqref{lambdanomo2}. 
An excessively large $\delta$ may impair the discriminative ability of SARM, resulting in diminished precision in outlier detection.

\begin{figure}[htbp]
  \centering
  \begin{subfigure}[b]{0.485\textwidth}
    \centering
    \includegraphics[width=\linewidth]{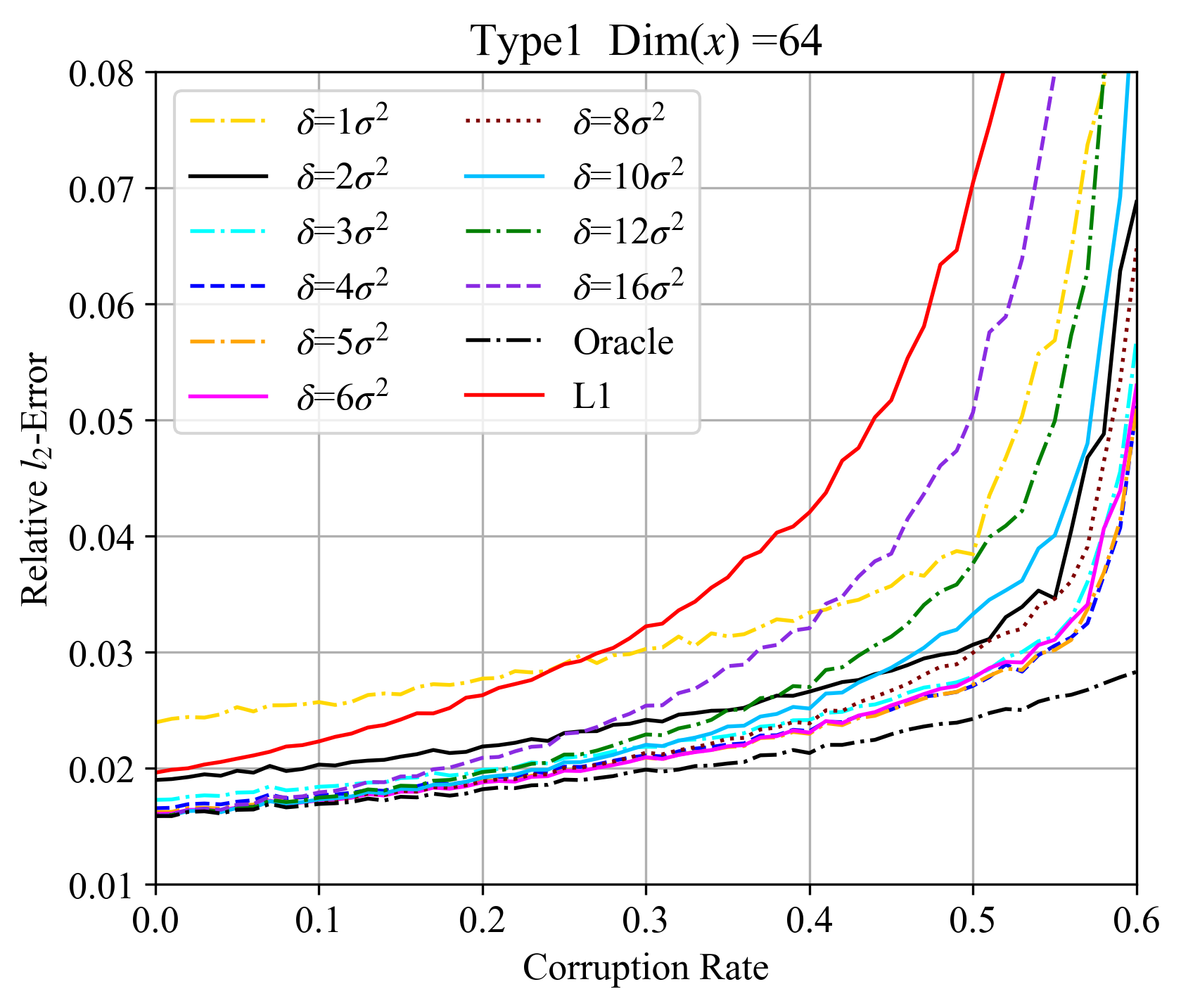}
  \end{subfigure}
  \hfill
  \begin{subfigure}[b]{0.495\textwidth}
    \centering
    \includegraphics[width=\linewidth]{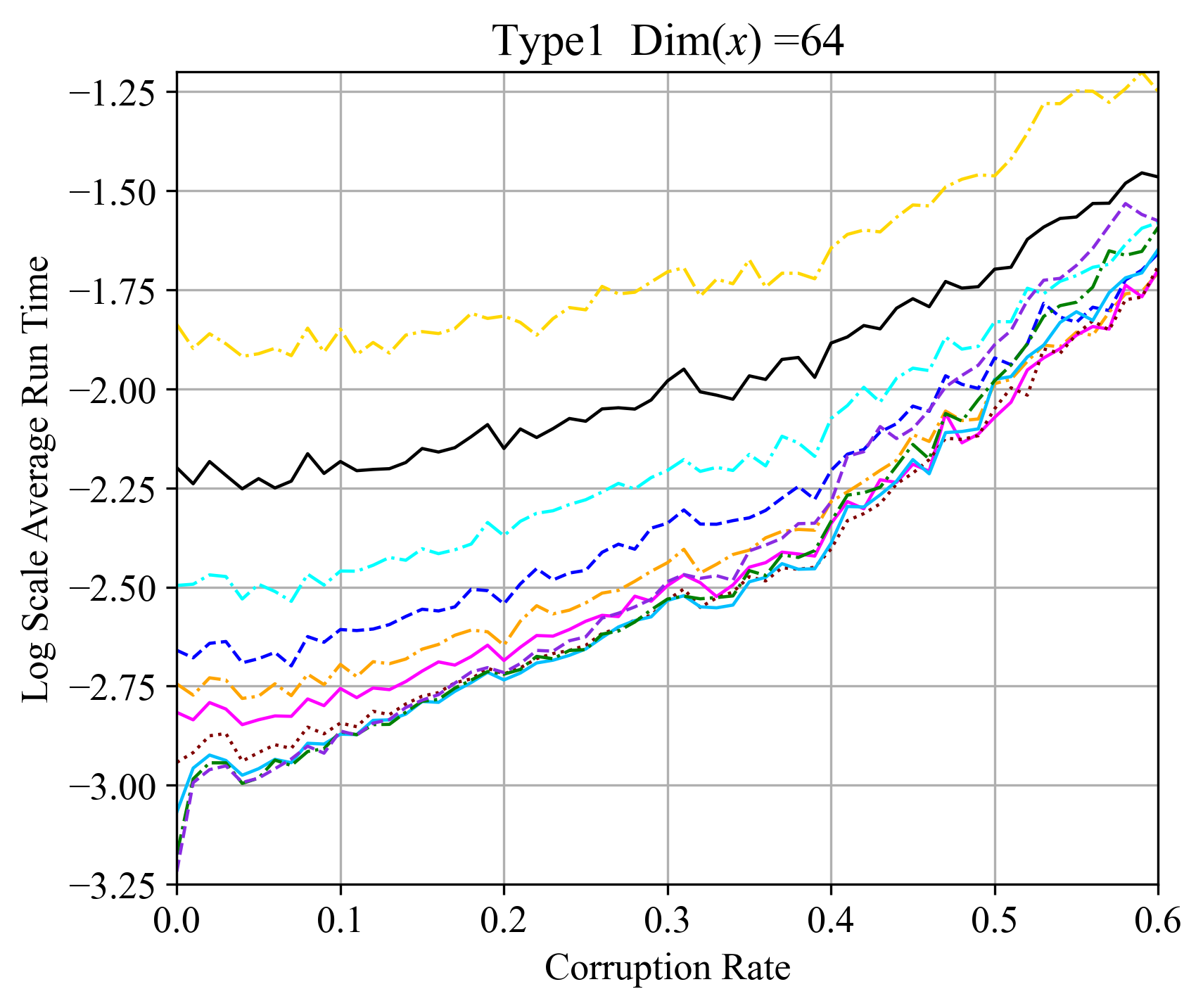}
  \end{subfigure}
  \caption{Sensitivity Analysis for $\delta$. }
  \label{Sensitivity Analysis}
\end{figure}

\subsection{Extended Comparative Experiments}

While considering a single outlier setting such as Type 1 provides initial insights, it is insufficient to fully demonstrate the generalizability of SARM. To further assess its robustness and adaptability across diverse contamination scenarios, we extend our investigation to include additional synthetic settings: Type 2 and Type 5.

Type 2 simulates a contamination mechanism that involves point-mass contamination of large magnitude in response.
The experimental results are presented in Fig. \ref{Type2}.
We observe that the performances of L1 and those models relying on L1 for initial estimation or iteration (e.g., AROSI and IPOD) almost coincide.
However, the performance of TRLM is unexpectedly superior. This suggests that the robustness of L1 against the swamping effect may vary depending on the distribution of outliers, and that the $\ell_2$ loss function could potentially exhibit better performance under point-mass contamination scenarios.
Despite variations in robustness exhibited by other models under different contamination scenarios, SARM consistently demonstrates superior robust performance in this challenging scenario.

Type 5 describes a scenario where a subset of sample points exhibits an abnormally increased variance. The experimental results are presented in Fig. \ref{Type5}. In this setting, the performance of convex $\ell_1$-relaxation-based models tends to degrade as the variance ($(\kappa \sigma)^2$) of outliers increases. This again confirms that $\ell_1$-relaxation-based approaches are sensitive to the magnitude of outlier deviation. In contrast, models based on $\ell_0$ penalties, such as GARD and AROSI, as well as truncation-based methods like TRLM and IPOD, exhibit stable or even improved performance with increasing outlier variance. Overall, SARM continues to demonstrate nearly optimal robustness in this scenario.

The above experimental results further demonstrate the broad applicability of SARM across different types of data contamination.

\begin{figure}[htbp]
  \centering
  \begin{subfigure}[b]{0.325\textwidth}
    \centering
    \includegraphics[width=\linewidth]{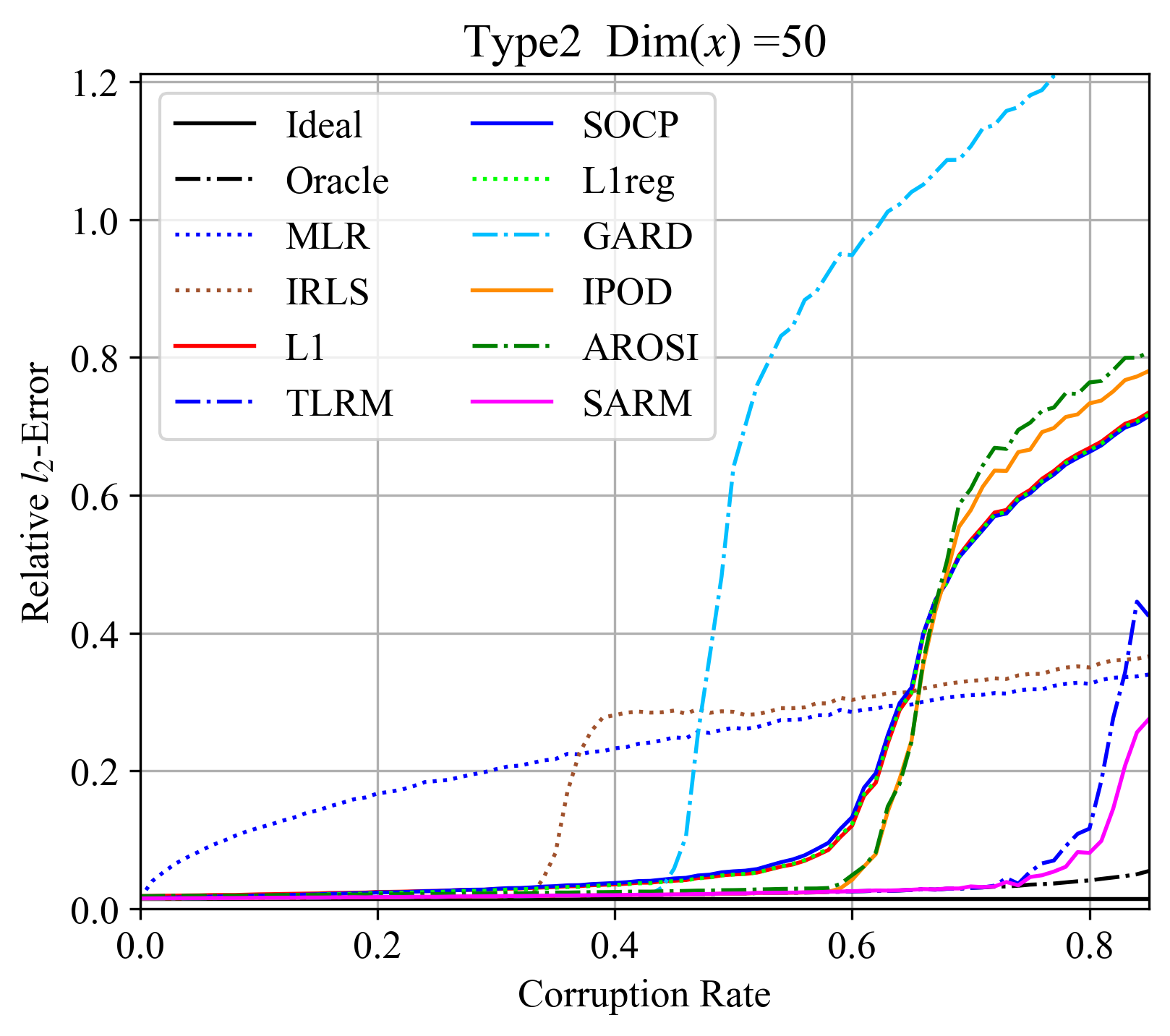}
  \end{subfigure}
  \hfill
  \begin{subfigure}[b]{0.33\textwidth}
    \centering
    \includegraphics[width=\linewidth]{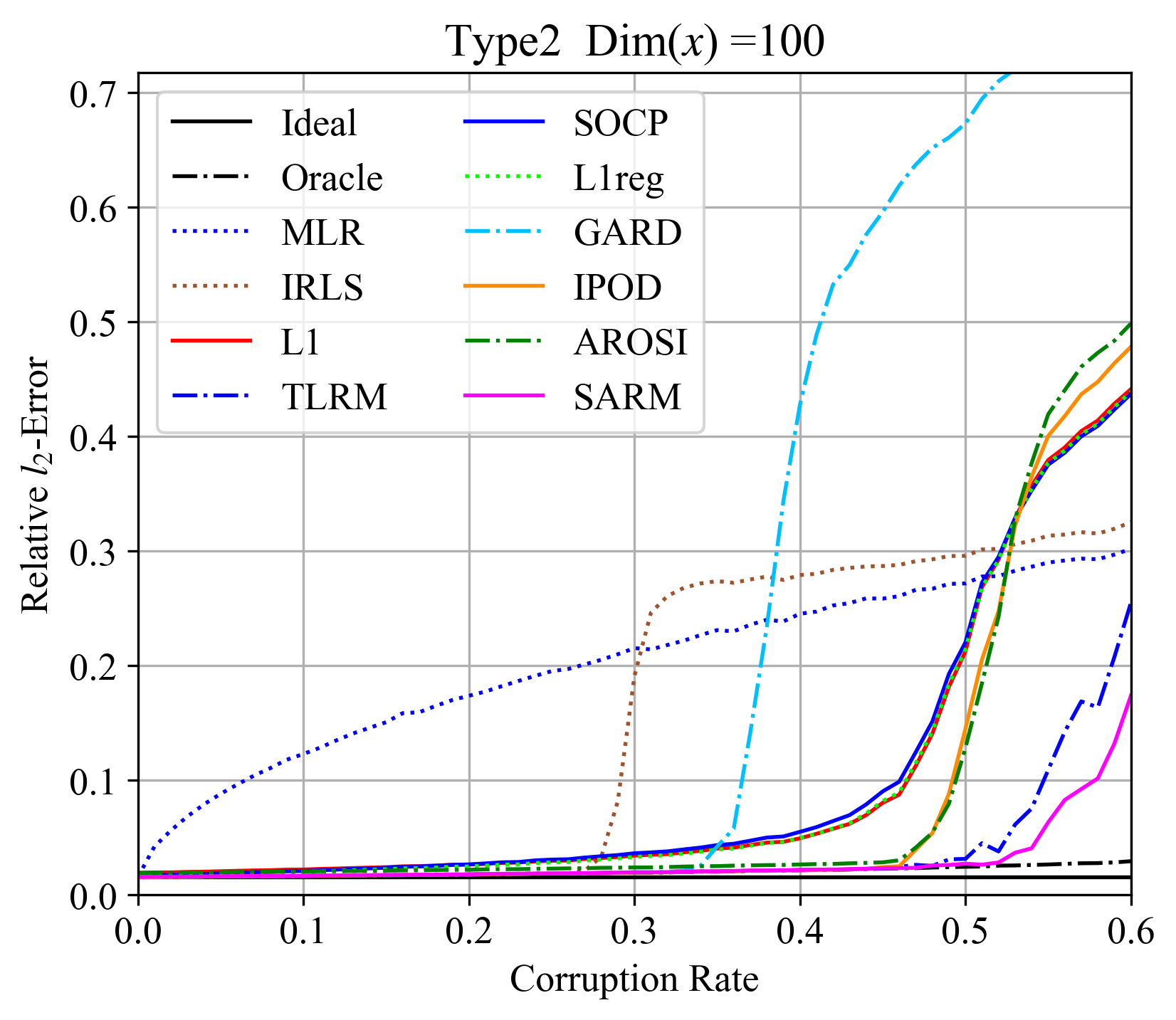}
  \end{subfigure}
  \hfill
  \begin{subfigure}[b]{0.325\textwidth}
    \centering
    \includegraphics[width=\linewidth]{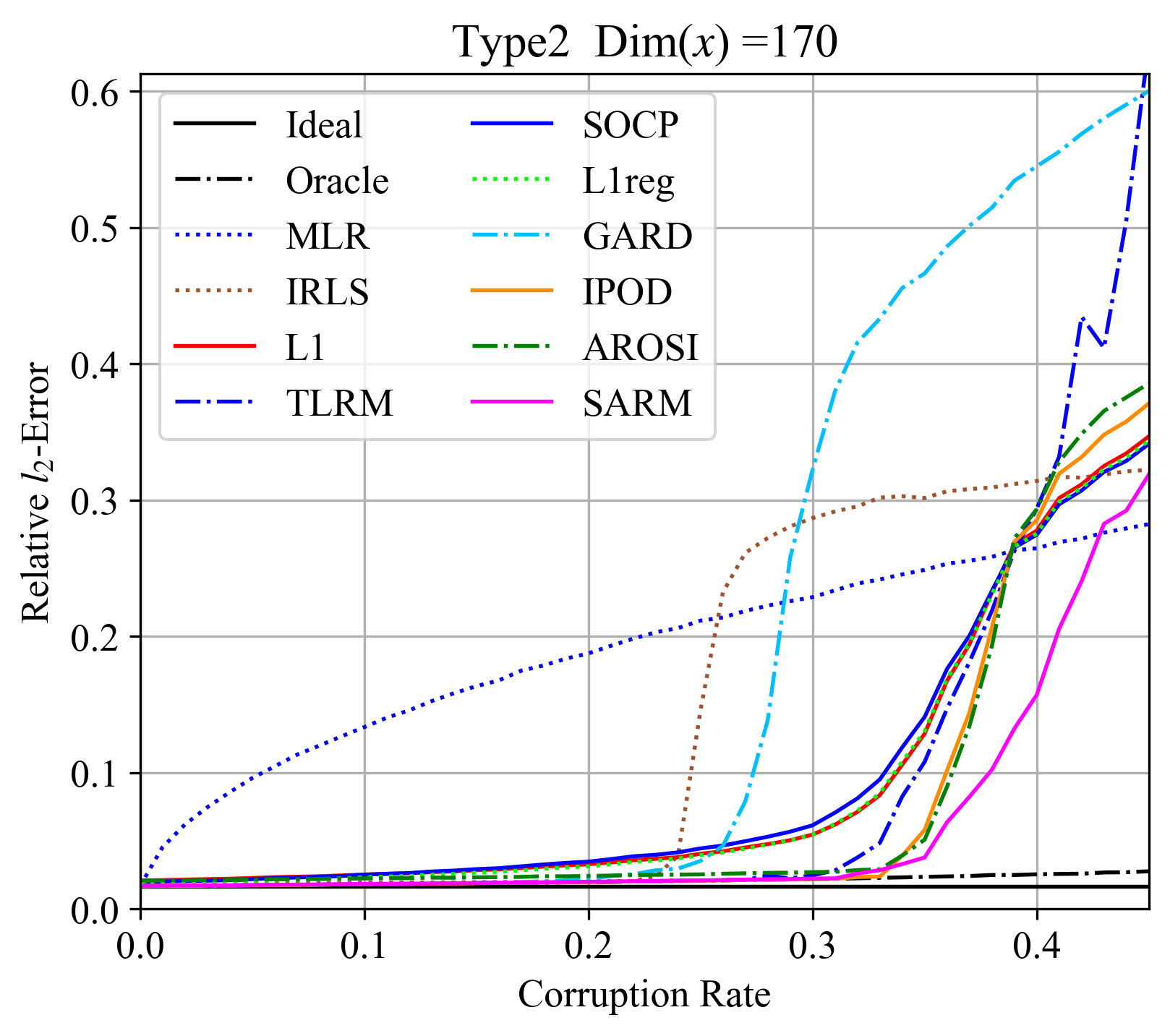}
  \end{subfigure}
  \caption{Empirical Evaluation of Robustness: Type2}
  \label{Type2}
\end{figure}

\begin{figure}[htbp]
  \centering
  \begin{subfigure}[b]{0.33\textwidth}
    \centering
    \includegraphics[width=\linewidth]{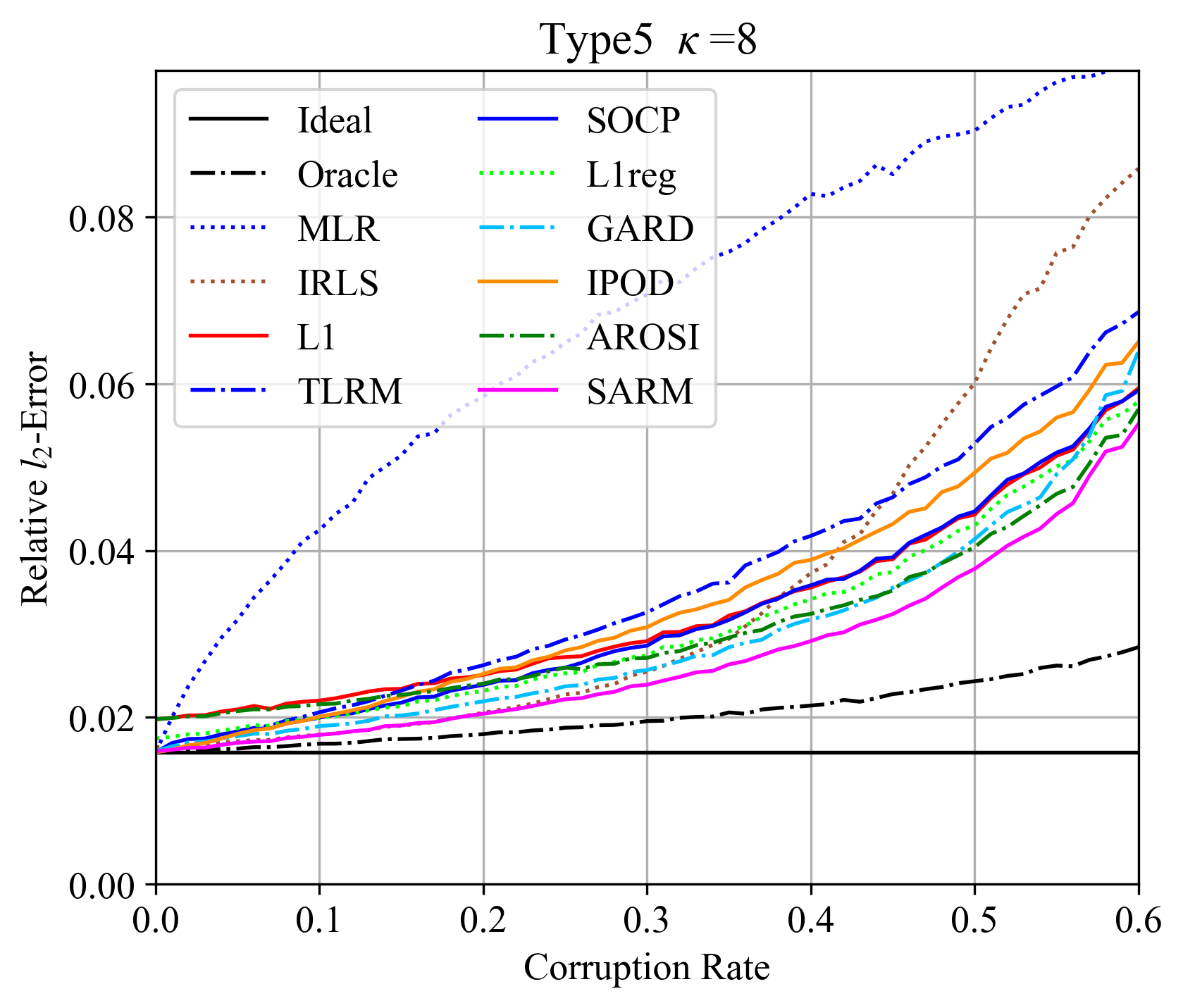}
  \end{subfigure}
  \hfill
  \begin{subfigure}[b]{0.33\textwidth}
    \centering
    \includegraphics[width=\linewidth]{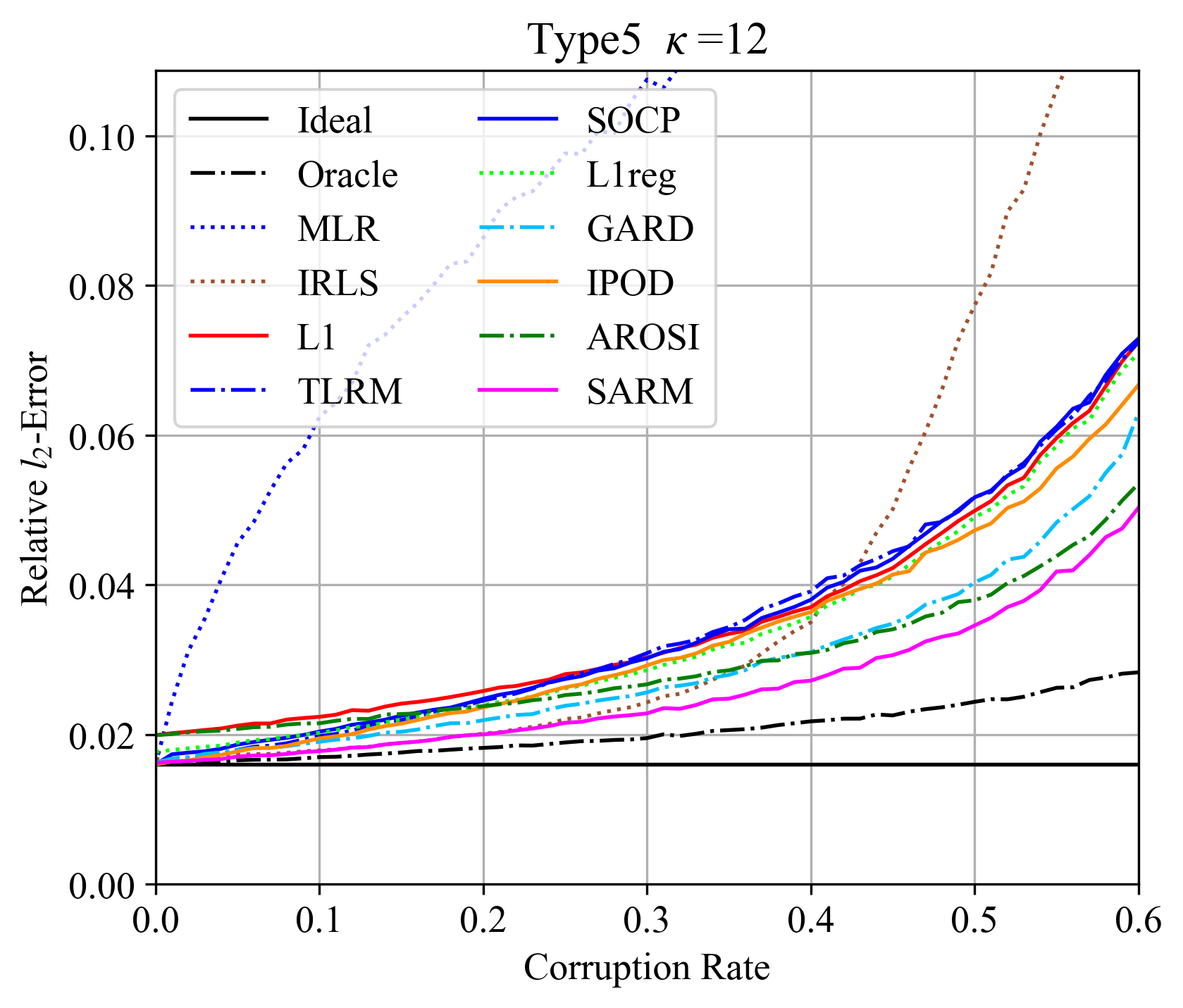}
  \end{subfigure}
  \hfill
  \begin{subfigure}[b]{0.33\textwidth}
    \centering
    \includegraphics[width=\linewidth]{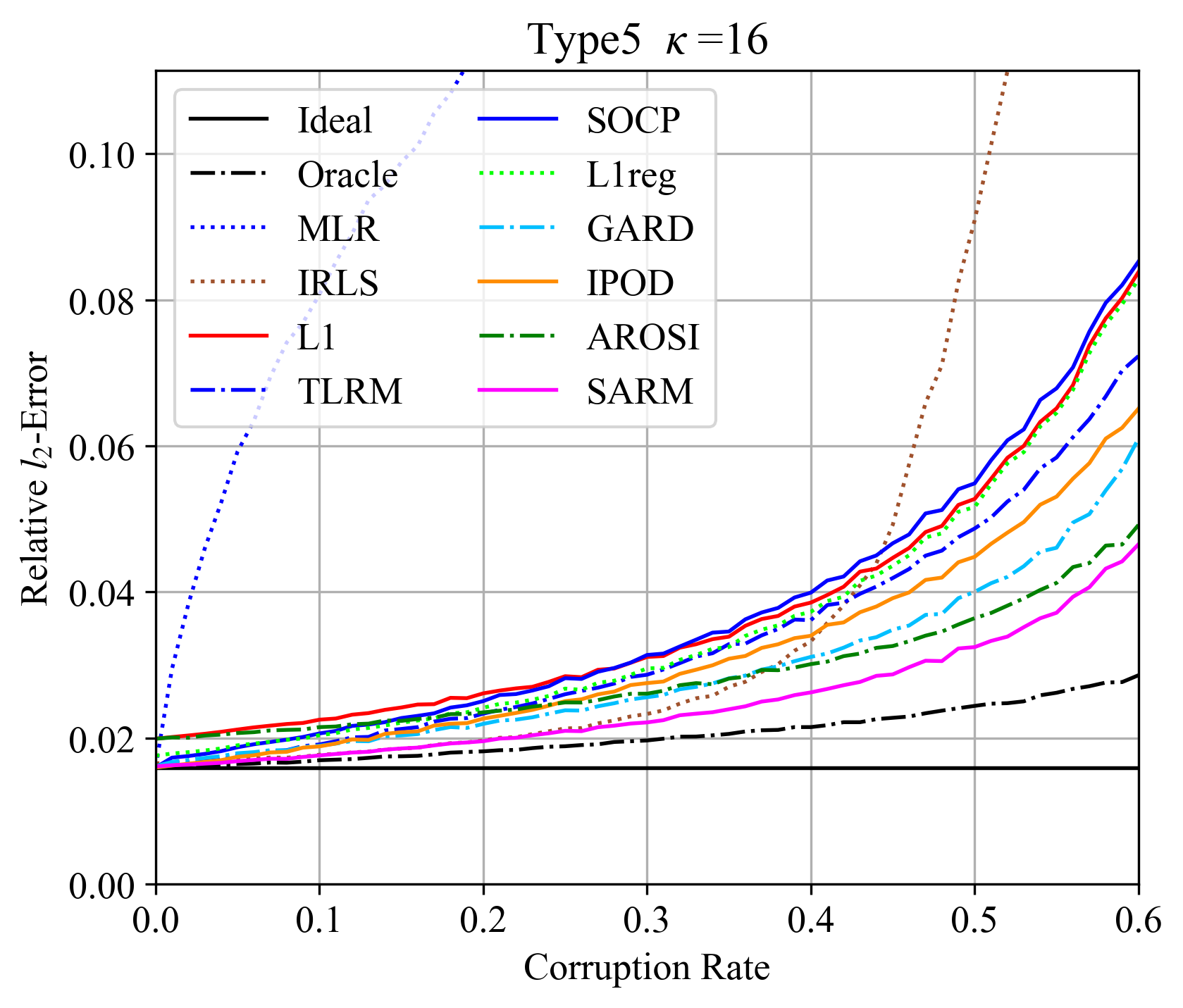}
  \end{subfigure}
  \caption{Empirical Evaluation of Robustness: Type5}
  \label{Type5}
\end{figure}

\section{The Two-Stage SARM Framework} \label{TSSARM}

Let us now revisit Fig. \ref{Type1} and Fig. \ref{Type2}. An evident pattern emerges from the results: as the dimensionality of the feature space increases, the breakdown point of the estimators tends to decline accordingly.
A similar conclusion can also be reflected through the upper bound of the estimation error (Theorem \ref{upper bound}).
Based on this observation, a straightforward idea is to enhance model robustness by applying PCA or other dimensionality reduction methods. However, such a strategy inevitably sacrifices potential information. In light of the third concern raised in Section \ref{motivation}, we consider that this property of robust models may instead be leveraged by SARM to provide better initialization under certain conditions.

Assume that \( X \) is the normalized input matrix (Min-max, Z-score, etc.).
We apply an orthogonal transformation matrix \( Q \in \mathbb{R}^{n\times n} \) to the original design matrix \( X \):
\begin{align}
y &= XQ w  + e + z.
\end{align}
Subsequently, we select \( q \) columns of $Q$ to perform feature compression on \( X \):
\begin{align}
y &= X\left[Q_q, Q_{n-q} \right]\begin{bmatrix}
w_{q}  \\
 w_{n-q}
\end{bmatrix}  + e + z = XQ_q w_q + XQ_{n-q} w_{n-q}  + e + z,
\end{align}
When \( X Q_{n-q} w \) is sufficiently small, it can be temporarily regarded as part of random noise $e$.
Considering the following inequality:
\begin{equation}
\|XQ_{n-q}w_{n-q}\|_2 \leq \|XQ_{n-q}\|_2 \| w_{n-q}\|_2 ,
\end{equation}
we aim to minimize \( \|XQ_{n-q}\|_2 \) to satisfy the assumption that \( XQ_{n-q}w \) is sufficiently small.
\( \|XQ_{n-q}\|_2^2 \) is equivalent to:
\begin{align}
 & \max_{u \in \mathbb{R}^{n-q}, \|u \|_2 = 1} u^T Q_{n-q}^T X^T XQ_{n-q} u. \label{min2norm}
\end{align}
Given the singular value decomposition (SVD) $X = U\Sigma V^T$ and suppsoing $\tilde{Q}_{n-q} = V^T Q_{n-q}$, (\ref{min2norm}) is also equivalent to:
\begin{equation}
 \max_{u \in \mathbb{R}^{n-q}, \|u \|_2 = 1} u^T \tilde{Q}^T_{n-q}\Sigma^2 \tilde{Q}_{n-q} u,
\end{equation}
where $\Sigma^2 = \mathrm{diag}(s_1^2, s_2^2, \ldots, s_n^2)$ and $s_i$ is the $i$-th largest singular value of $X$.
Since the diagonal elements of $\Sigma$ are arranged in descending order, choosing 
$\tilde{Q}_{n-q} = \begin{bmatrix}
\mathbf{0}  \\
I_{n-q}
\end{bmatrix}$ results in the minimal \( \|XQ_{n-q}\|_2 = s_{q+1} \), where $s_{q+1}$ denotes ${(q+1)}$-th singular value of $X$. Therefore, we take $Q_{n-q} =( V^T)^{-1}\begin{bmatrix}
\mathbf{0}  \\
I_{n-q}
\end{bmatrix} = V_{n-q}$, which denote the last $n-q$ columns of $V$. Then we take $Q = V$ and $XV_{n-q} w=e_c$:
\begin{align}
y &= XV_q w + (e_c  + e) + z.
\end{align}

We still need to apply preconditioning to \(X V^T_q\) and the new input matrix is \(X V_q \Sigma_q^{-1}\), where $\Sigma_q^{-1} = \mathrm{diag}(\frac{1}{s_1}, \frac{1}{s_2}, \ldots, \frac{1}{s_q})$. Note that the essence of preconditioning is to normalize the singular values of the input matrix $X$ to be all equal to 1 and \(X V_q \Sigma_q^{-1} = U \Sigma V^T V_q \Sigma_q^{-1}=U_q\) satisfies the condition.
From a computational efficiency perspective, it is unnecessary to explicitly compute $U$. Instead, we can obtain $V$ and $\Sigma^2$ by conduct SVD on $X^TX$: $X^TX = V\Sigma^2V^T$. This helps reduce the computational complexity, especially when $n \ll m$.

When the singular values of \(X\) exhibit significant imbalance, we can select an appropriate \(q\) such that \(e_c\) is sufficiently small.
This allows us to obtain a more robust initial estimate using fewer sample features.
Subsequently, we reconsider all features and re-execute the SARM algorithm conditioned on this initial estimate, ultimately obtaining the estimate of \(w\). The specific procedure is presented in Algorithm \ref{SARMTSAlgorithm}.
We refer to this algorithm as SARM with Two Stages (SARMTS), where the two stages correspond to the initial pre-estimation and correction modules. 

Here, our method for selecting \(q\) is based on $\frac{s_q}{s_1}$; in other words, we set a threshold $\eta \in [0,1]$ and choose the smallest integer $q$ 
 such that \(\frac{\|X U_{n-q} \|_2^2}{\|X\|_2^2} = \frac{s_{q+1}}{s_1} < \eta\). 
 $\eta$ is typically required to be sufficiently small. In our implementation, it is set to $0.005$.
 This approach is similar to PCA dimensionality reduction, but since the process is derived purely from SVD, it focuses solely on matrix approximation. In addition, a larger parameter $\delta_{pre}$ is set during the pre-estimation stage, which corresponds to the fact that $\|e_u + e\|_2 $ is usually larger than $ \|e\|_2$.

\begin{algorithm}
\caption{Algorithm For SARM with Two Stages}
\begin{algorithmic}[1]
\Require $y, X_{origin}, \alpha, \eta \in [0,1],\delta >0, \delta_{pre} >0$
\State Perform Singular Value Decomposition on $X^T X$, and obtain the left singular vectors matrix $U$
\State Select the smallest integer \( q \) such that $s_{q+1} < \eta s_1$
\State Initialization: $k = 0$, $w^{0} = 0 \in \mathbb{R}^q$, $z^{0} = 0 \in \mathbb{R}^m$
\State \(X = X_{origin} U_q \Sigma_q^{-\frac{1}{2}} \in \mathbb{R}^{m\times q}\)
\While{$H(w,z)$ not converged}
    \State $r^{k} = y - Xw^{k}$
    \State (Update $w$):
    $w^{k+1} = w^{k} - \alpha X^T (r^{k} - z^{k}) - \alpha \cdot \delta_{pre} \cdot X^T \frac{z^{k}}{S^2(r^{k})}$ \label{update w}
    \State (Update $z$):
        $z^{k+1}_i = \begin{cases}
        0, & \text{if } | r^{k}_i |  \leq \sqrt{\delta} \\
        r^{k}_i - \frac{\delta}{r^{k}_i}, & \text{if } | r^{k}_i  | > \sqrt{\delta}
    \end{cases}$
    \State $k := k + 1$
\EndWhile
\State Obtain $w_{pre} = w^{k}\in \mathbb{R}^q$ and $z_{pre} = z^{k}\in \mathbb{R}^m$.
\State Initialization: $k = 0$, $w^{0} = [w_{pre}; 0] \in \mathbb{R}^n$, $z^{0} = z_{pre} \in \mathbb{R}^m$
\State \(X = X_{origin} U \Sigma^{-\frac{1}{2}} \in \mathbb{R}^{m\times n}\)
\While{$H(w,z)$ not converged}
    \State $r^{k} = y - Xw^{k}$
    \State (Update $w$):
    $w^{k+1} = w^{k} - \alpha X^T (r^{k} - z^{k}) - \alpha \cdot \delta \cdot X^T \frac{z^{k}}{S^2(r^{k})}$ \label{update w}
    \State (Update $z$):
        $z^{k+1}_i = \begin{cases}
        0, & \text{if } | r^{k}_i |  \leq \sqrt{\delta} \\
        r^{k}_i - \frac{\delta}{r^{k}_i}, & \text{if } | r^{k}_i  | > \sqrt{\delta}
    \end{cases}$
    \State $k := k + 1$
\EndWhile
\State $w^{k} = U \Sigma^{-\frac{1}{2}} w^{k}$
\Ensure Final solution $w^{k}$ and $z^{k}$
\end{algorithmic}
\label{SARMTSAlgorithm}
\end{algorithm}

Under the Type1 and Type2 settings, the condition number of the input matrix \(X\) (i.e., the ratio between its largest and smallest singular values) is typically small, in which case SARMTS degenerates into the standard SARM. To demonstrate the performance of SARMTS under imbalanced singular value distributions, we design a synthetic experiment where the singular values of \(X\) are deliberately made highly unequal. The specific setup is as follows:
\begin{enumerate}
  \item \textbf{Type6}: \(m = 512\), \(\sigma_w = 1\), $X = DB$, \(\sigma = \text{median}(|Xw|)/16\). Corruptions are obtained from $0.5 \times \mathcal{N}(12\sigma, (4\sigma)^2) + 0.5 \times \mathcal{N}(-12\sigma, (4\sigma)^2)$.  \(n \in \{64, 128,192\}\). \(D \in \mathbb{R}^{m \times n} \), \(D_{ij} \in \mathcal{N}(0, 1)\).
\end{enumerate}
where \( B \in \mathbb{R}^{n \times n} \) is constructed as the following way.
First, we generate a diagonal vector $\mathbf{d}$ of length \( n \) exponentially decreasing values by first constructing \( n \) linearly spaced points $d_i$ over the interval \([-0.1, \log(1.1/n)]\) and then applying the exponential function element-wise.
The diagonal entries of matrix \(B\) are set to $\mathbf{d}$, and the elements in the \(j\)-th column are assigned as $B_{ij} =  \frac{1-\mathbf{d}_j}{n-1}, \forall i$. The resulting matrix \( X = DB \) possesses a highly unbalanced singular value distribution. As an example, Fig. \ref{SVDsimu} shows the singular value distribution.
The experimental results are shown in Fig. \ref{type6}.

\begin{figure}[htbp]
  \centering
  \begin{subfigure}[b]{0.33\textwidth}
    \centering
    \includegraphics[width=\linewidth]{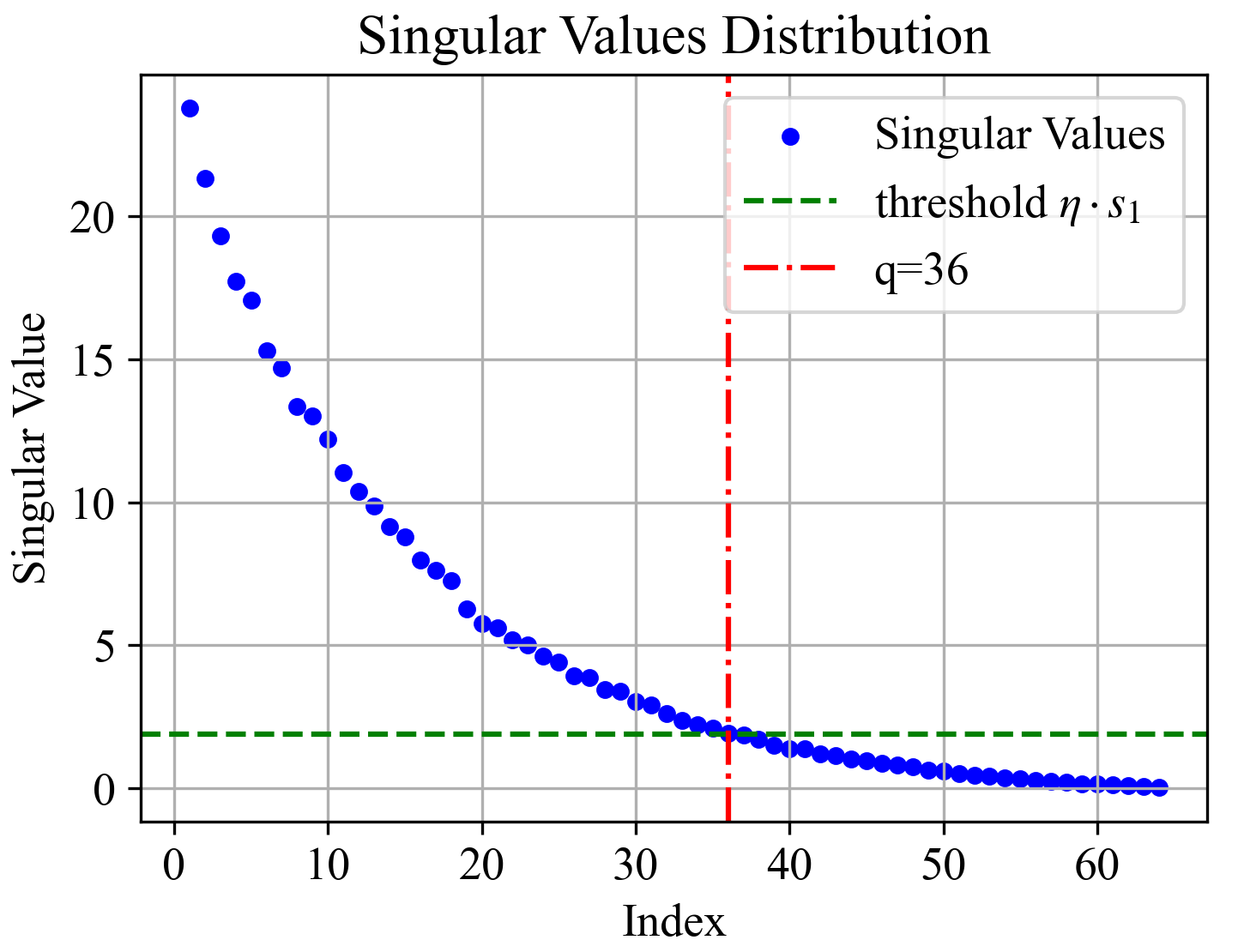}
    \caption{Type6 Dim(x)=64}
  \end{subfigure}
  \hfill
  \begin{subfigure}[b]{0.33\textwidth}
    \centering
    \includegraphics[width=\linewidth]{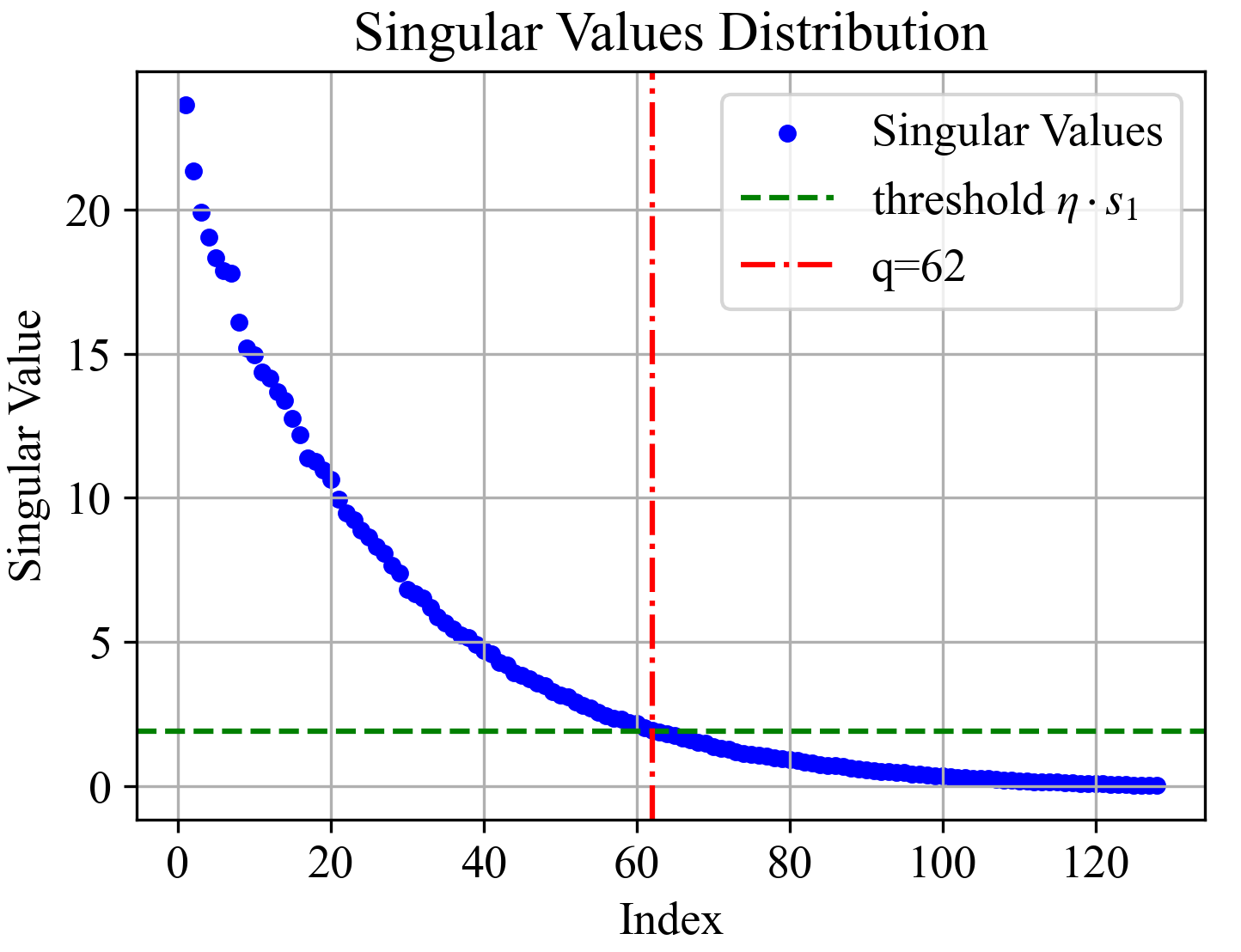}
    \caption{Type6 Dim(x)=128}
  \end{subfigure}
  \hfill
  \begin{subfigure}[b]{0.33\textwidth}
    \centering
    \includegraphics[width=\linewidth]{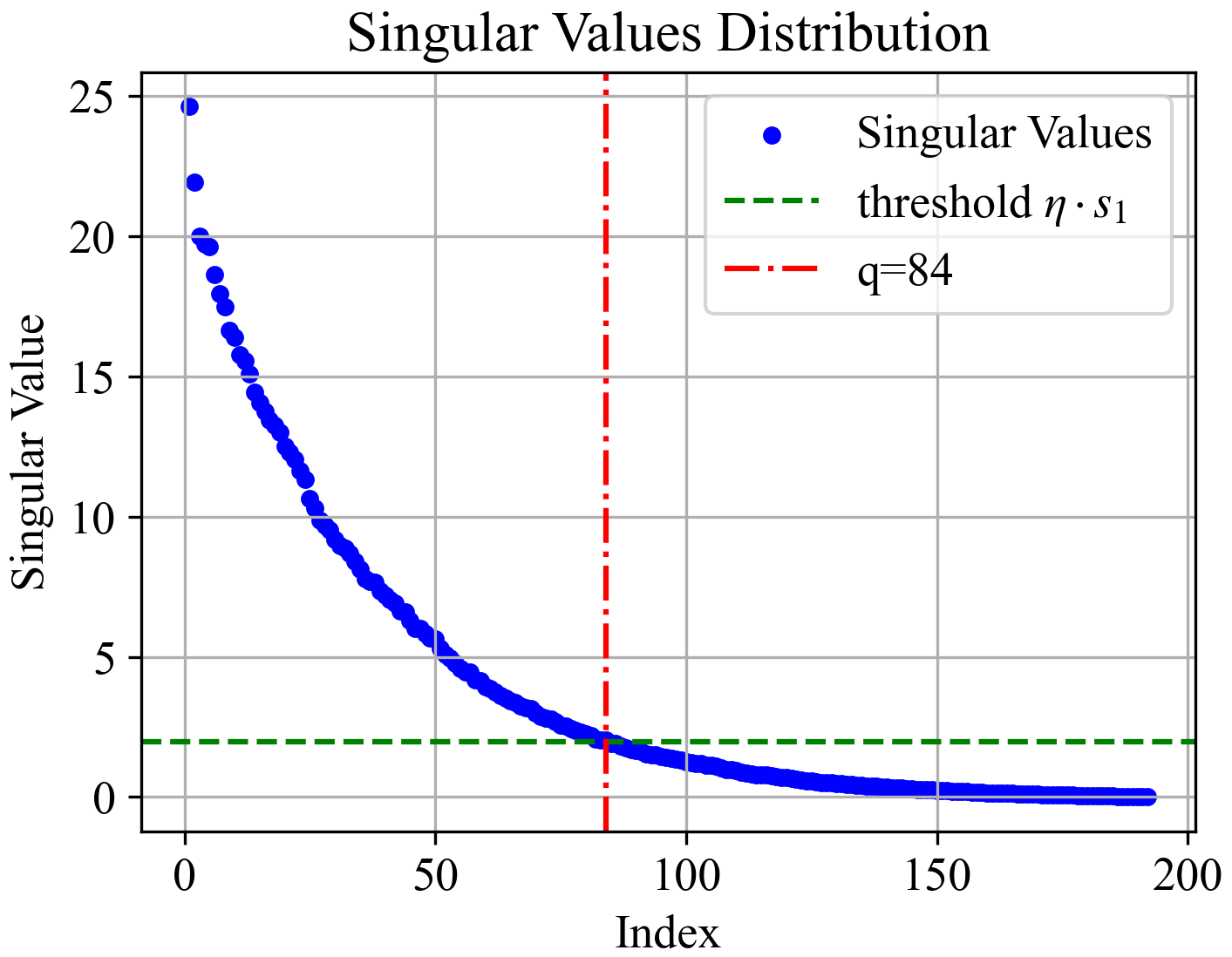}
    \caption{Type6 Dim(x)=192}
  \end{subfigure}
  
  \caption{The singular value distribution for Type6}
  \label{SVDsimu}
\end{figure}

\begin{figure}[htbp]
  \centering
  \begin{subfigure}[b]{0.33\textwidth}
    \centering
    \includegraphics[width=\linewidth]{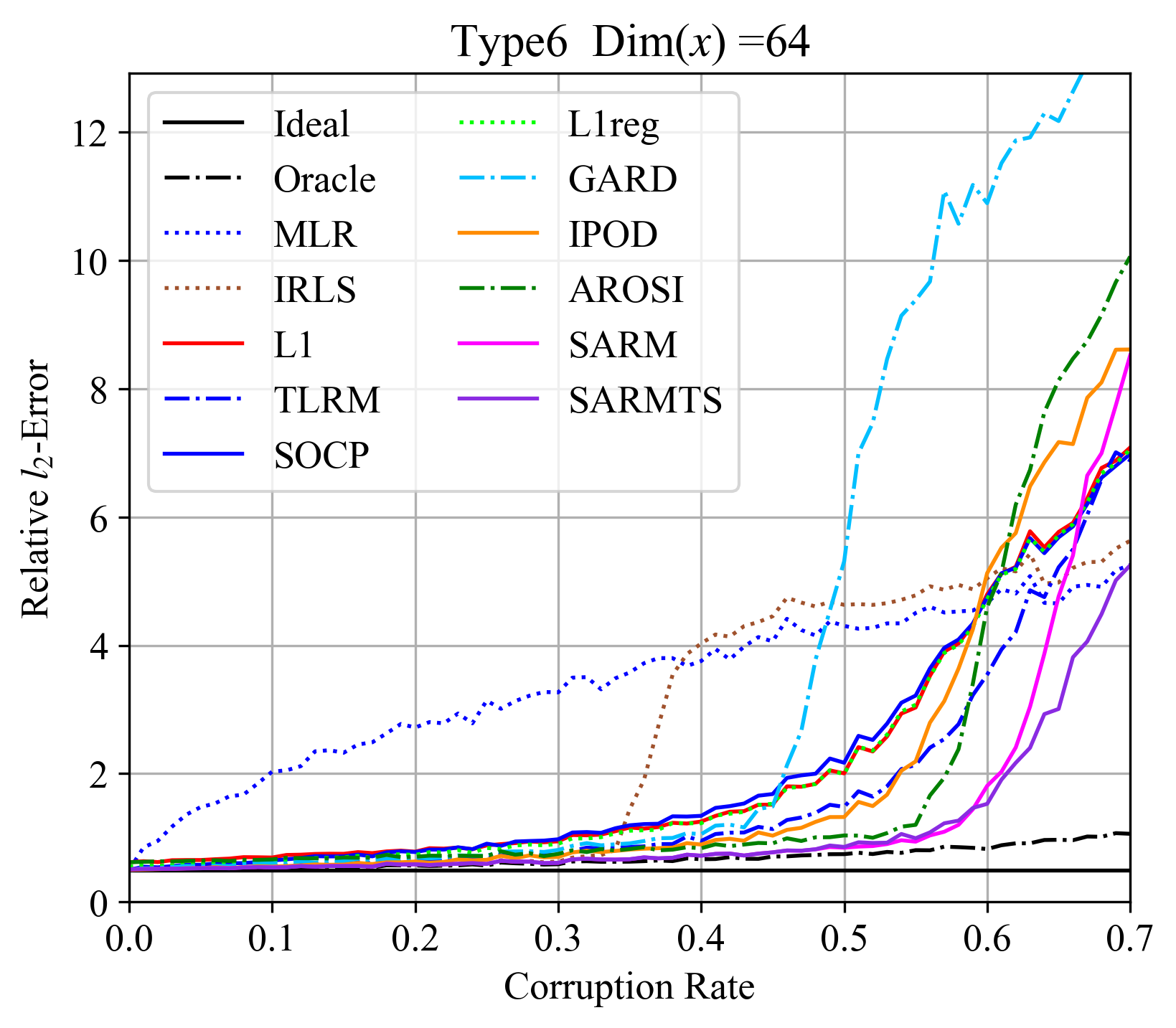}
  \end{subfigure}
  \hfill
  \begin{subfigure}[b]{0.33\textwidth}
    \centering
    \includegraphics[width=\linewidth]{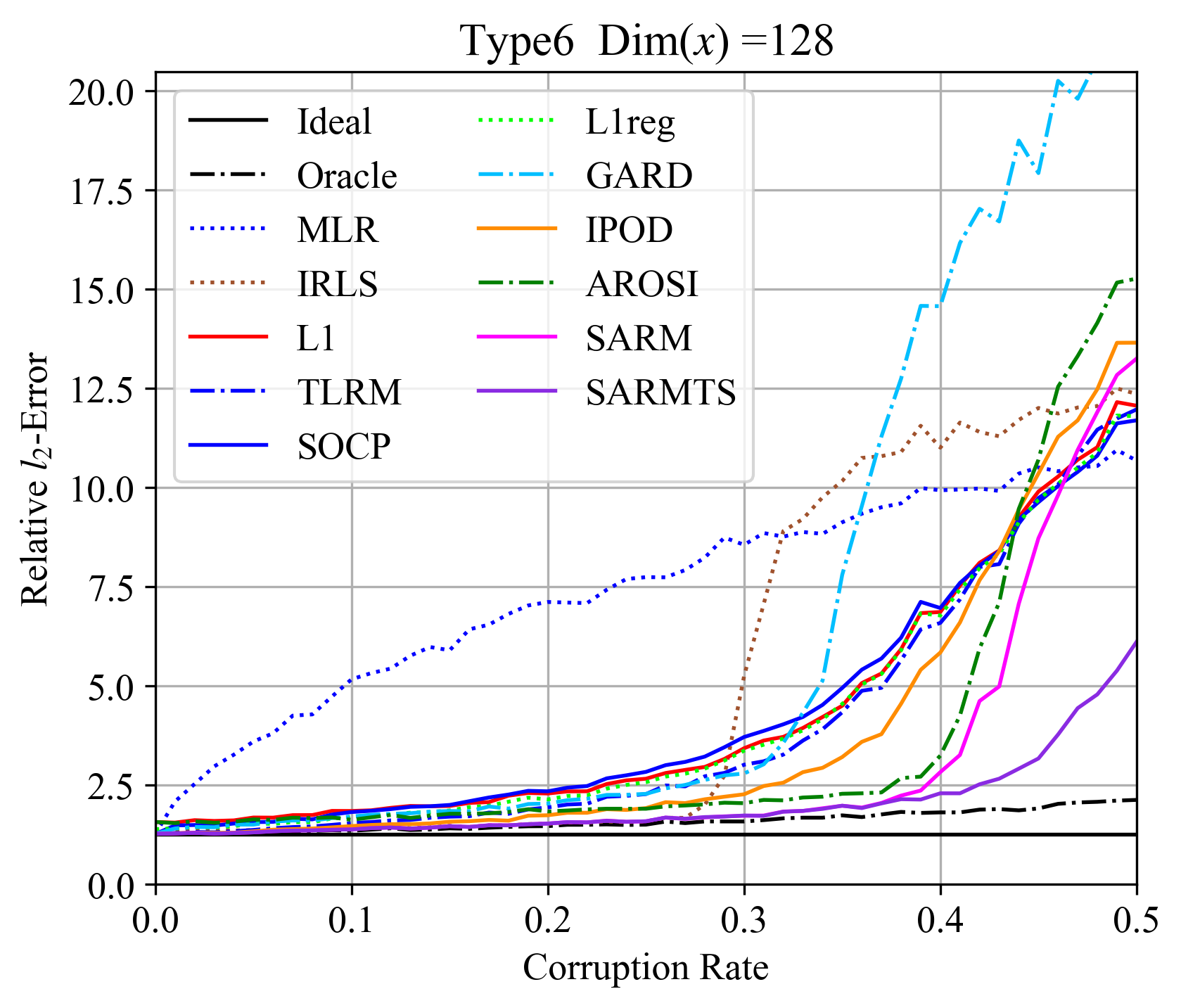}
  \end{subfigure}
  \hfill
  \begin{subfigure}[b]{0.33\textwidth}
    \centering
    \includegraphics[width=\linewidth]{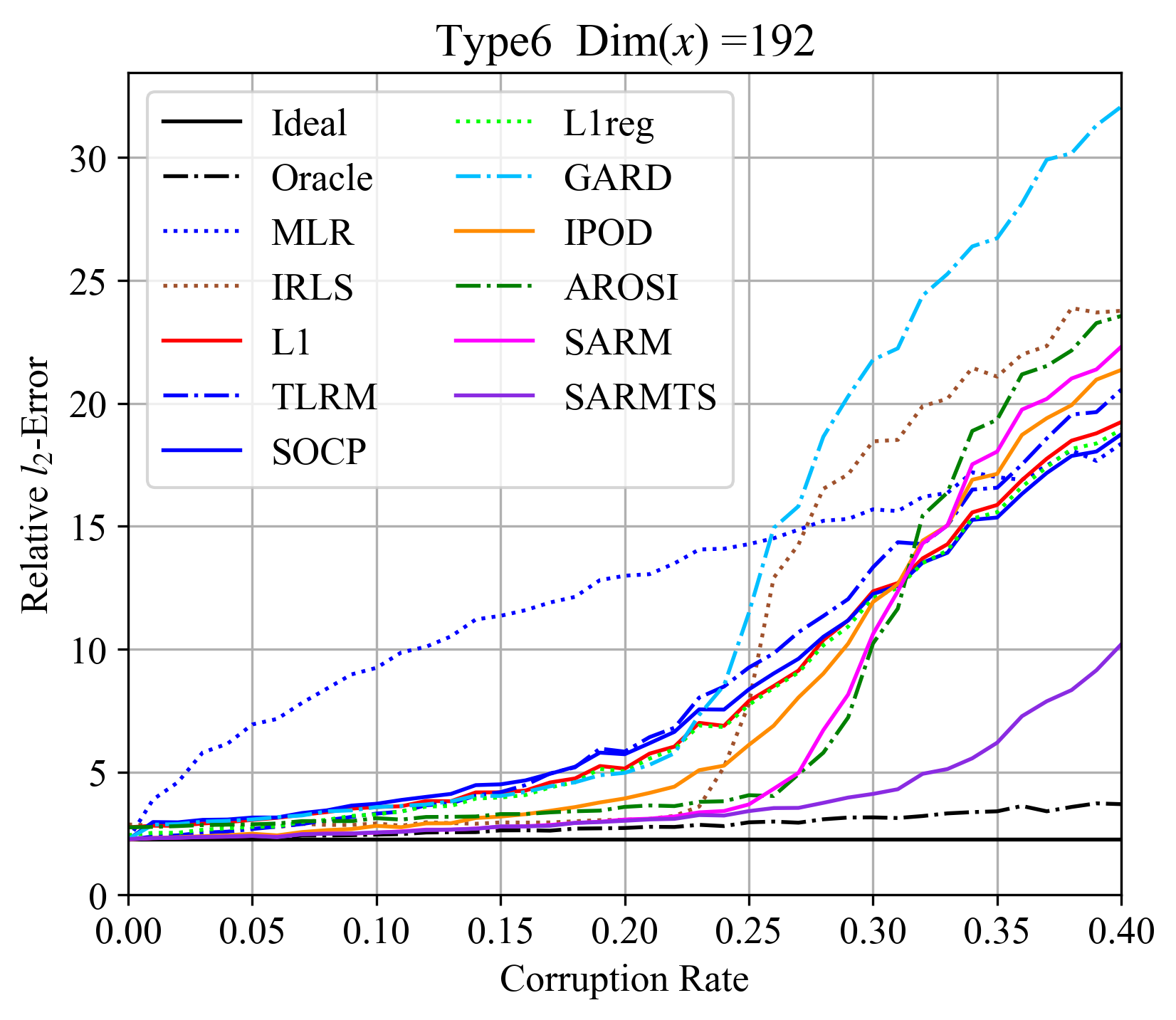}
  \end{subfigure}
  \caption{Empirical Evaluation of Robustness: Type6}
  \label{type6}
\end{figure}

\section{Applied Studies on Load Forecasting} \label{Applied}

\subsection{Data Integrity Attacks and Load Forecasting}

The power grid has become an essential infrastructure in modern society, and ensuring its security is vital for the stable and continuous functioning of social and economic activities.
The coordinated cyberattacks on three Ukrainian power companies in 2015 \cite{case2016analysis}, along with the cyberattack targeting a major U.S. load forecasting service provider in 2018, have brought increasing attention to the potential threats that cyberattacks pose to the reliability of power systems. Among various forms of cyber threats, integrity attacks, such as the manipulation or falsification of input data, are particularly threatening due to their covert nature, making them more dangerous than other types of cyber threats \cite{hong2022data}.

In this section, we apply SARM (SARMTS) to a load forecasting task subject to data integrity attacks in order to evaluate its practical performance. A typical load forecasting model is Tao's vanilla benchmark. It can be expressed as:
\begin{align} 
y_t= & \beta_0 + \beta_1 Trend_t + \beta_2 H_t + \beta_3 W_t + \beta_4 M_t \notag \\ 
& + \beta_5 H_t W_t + f(T_t) 
\label{Tao}
\end{align}
and $f \left( T \right)$ is
\begin{align*} 
f \left( T \right)= & \beta_6 T + \beta_7 \left(T\right)^2 + \beta_8 \left(T\right)^3 + \beta_9 T_t H_t + \beta_{10} \left( T \right)^2
H_t + \\ &  \beta_{11} \left(T\right)^3 + \beta_{12} T M_t + \beta_{13} \left(T\right)^2 M_t + \beta_{14} \left(T\right)^3 M_t, 
\end{align*}
\par where $y_t$ denotes the load at time $t$, 
$W_t$ is categorical variables indicating 7 days of a week.
$M_t$ is categorical variables indicating 12 months. 
$H_t$ is categorical variables indicating 24 hours. 
$W_t, M_t, H_t$ are all one-hot encoding vectors. 
$T_t$ is the temperature at time $t$.
$Trend_t$ is a variable of the increasing integers representing a linear trend. 
After removing linearly dependent components, a total of 285 feature variables were retained.

Tao's vanilla benchmark is a Multiple Linear Regression (MLR) load forecasting model. When historical load data are maliciously tampered with, the load forecasting model fitted on the corrupted data will produce significant forecasting errors, which may lead to disruptions in energy dispatch, resulting in substantial economic losses or widespread blackouts \cite{luo2018robust,luo2023robust}.

Next, we briefly introduce two types of data integrity attacks discussed in this section.

When data attackers inject positive attacks on the load datasets, forecasting models trained on these datasets affected by malicious attacks often lead to an overestimation of power load, which may lead to severe economic losses \cite{luo2023robust}.
Similar to the setups in \cite{luo2018robust,luo2023robust,yue2019descriptive,sobhani2020temperature,zheng2020load}, 
the data integrity attacks targeting economic losses is mathematically simulated by Positively-biased Uniform and Gaussian Attacks, respectively.
On the other hand, decreasing historical load values may result in under-estimating, which subsequently elevates the risk of brownouts. Similar to \cite{luo2023robust}, 
we simulate the data integrity attacks aimed at inducing system blackouts by Negatively-biased Uniform Attacks. Specifically, we simulate data integrity attacks by three ways:
\begin{enumerate}
  \item \textbf{Positively-biased Uniform Attacks}: Randomly pick $k\%$ of load values being deliberately increased by $p\%$. Vary $k$ from $0$ to $80$. $p\%$ is generated by the uniform distribution $U(a , b)$;
  \item \textbf{Positively-biased Gaussian Attacks}: Randomly pick $k\%$ of load values being deliberately increased by $p\%$. Vary $k$ from $0$ to $80$. $p\%$ is generated by the normal distribution $N(\mu , \sigma^2)$;
  \item \textbf{Negatively-biased Uniform Attacks}: Randomly pick $k\%$ of load values being deliberately reduced by $p\%$. Vary k from $10$ to $80$. $p\%$ is generated by the uniform distribution $U(a , b)$;
\end{enumerate}

\subsection{Evaluation Metrics and Data Description}

In accordance with the conventional evaluation metrics employed in load forecasting studies \cite{luo2023robust,xie2017variable}, the Mean Absolute Percentage Error (MAPE) is utilized as the principal criterion for evaluating forecasting accuracy. 
\begin{align*}
MAPE = \frac{1}{n_t} \sum_{i=1}^{n_t} \left\vert \frac{y_i - \widehat{y_i} }{y_i} \right\vert \times 100\%.
\end{align*}  

The datasets \footnote{https://www.iso-ne.com/isoexpress/web/reports/pricing/-/tree/zone-info} utilized for our applied studies are from ISO New England (ISONE), which serves as the power system operator for the New England region of the United States. 
These datasets contain hourly load and temperature data for the ISO New England, covering eight load zones: Maine (ME), New Hampshire (NH), Vermont (VT), Connecticut (CT), Rhode Island (RI), Southeastern Massachusetts (SEMASS), Western Central Massachusetts (WCMASS), and Northeastern Massachusetts (NEMASS). In addition, the total electrical load and average temperature data for three regions in Massachusetts (MASS) and all regions mentioned above (TOTAL) are also included.
Similar to \cite{luo2023robust}, we pick two years of hourly load (2013, 2014) for training and one year (2015) for testing, and apply Min-Max normalization for the input variables in (\ref{Tao}).

\subsection{Evaluation Metrics and Data Description}
 
We first consider singular value distributions. As a representative example, we present and analyze the TOTAL dataset. The singular value distribution of the input matrix $X$ is shown in Fig. \ref{SVDload}.
It can be observed that the singular values of $X$ exhibit a significant imbalance and form a distinct staircase-like pattern.
Similar imbalanced singular value distributions are observed across the datasets from other load zones as well.
Therefore, we adopt SARMTS instead of SARM to better exploit the imbalance.

\begin{figure}[htbp]
  \centering
    \centering
    \includegraphics[width=0.8\textwidth]{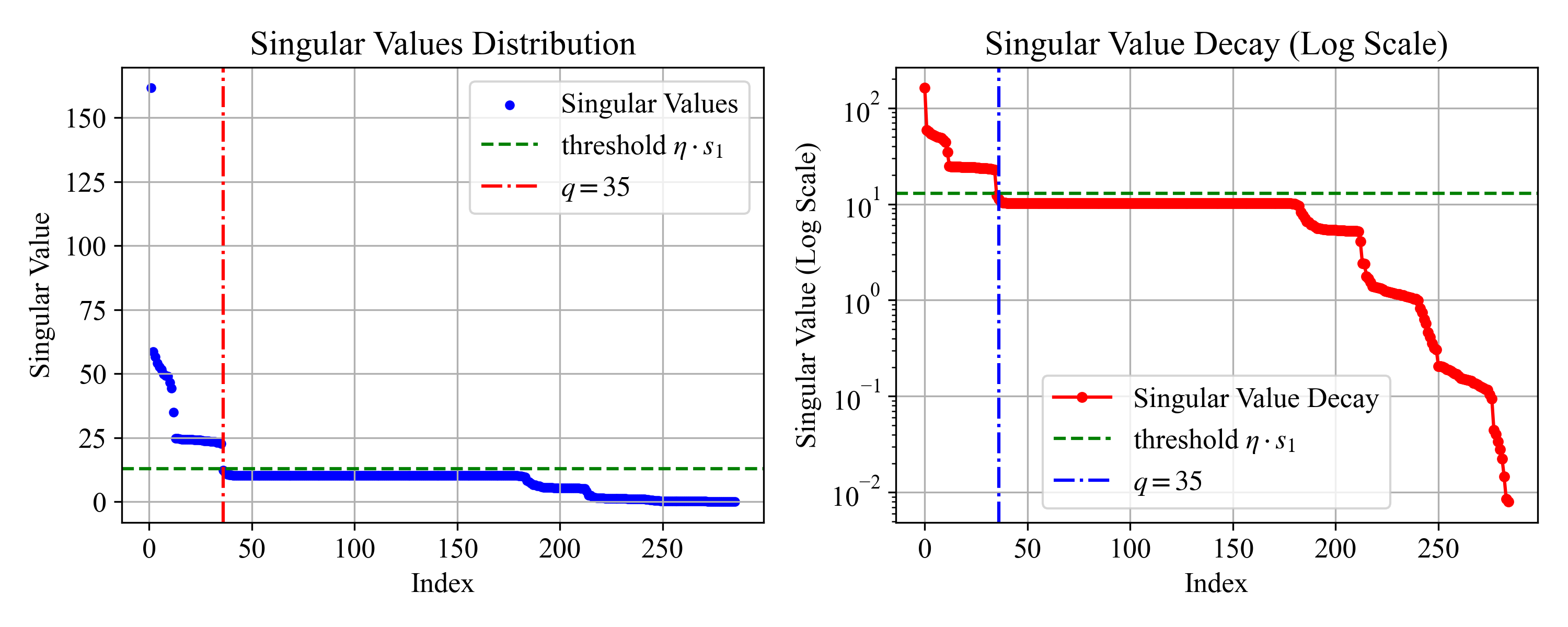}
  \caption{The singular value distribution for the ISONE load dataset.}
  \label{SVDload}
\end{figure}

An ideal robust forecasting model should maintain satisfactory performance even in the absence of data integrity attacks. To this end, we compare the forecasting performance of SARMTS against Tao's vanilla benchmark (MLR) model across ten datasets, as shown in \ref{noattack}.
The difference between the two models is negligible, and their average performance is virtually equivalent under the no-attack scenario.

\begin{table}[htbp] 
\centering
\caption{The GenCo coefficients for the 10-unit System.}
\begin{tabular}{m{0.08\textwidth}<{\centering} |
  m{0.05\textwidth}<{\centering} 
  m{0.05\textwidth}<{\centering} 
  m{0.05\textwidth}<{\centering} 
  m{0.07\textwidth}<{\centering} 
  m{0.05\textwidth}<{\centering} 
  m{0.05\textwidth}<{\centering} 
  m{0.07\textwidth}<{\centering} 
  m{0.05\textwidth}<{\centering} 
  m{0.05\textwidth}<{\centering} 
  m{0.07\textwidth}<{\centering} 
  m{0.05\textwidth}<{\centering} 
  }
\toprule
\textbf{Model} & \textbf{CT} & \textbf{MASS} & \textbf{ME} & \textbf{NEMASS} & \textbf{NH} & \textbf{RI} & \textbf{SEMASS}& \textbf{TOTAL}& \textbf{VT} & \textbf{WCMASS} & \textbf{AVE}\\ 
\midrule
MLR     &4.29 	&3.77 	&4.73 	&3.97 	&3.59 	&3.85 	&4.66 	&3.50 	&4.44 	&4.23     &4.10 \\ 
SARMTS  &4.22 	&3.80 	&4.87 	&3.92 	&3.59 	&3.80 	&4.55 	&3.52 	&4.36 	&4.34 	  &4.10 \\ 
\bottomrule
\end{tabular}
\label{noattack}
\end{table}

\begin{figure}[htbp]
  \centering
  \begin{subfigure}[b]{0.33\textwidth}
    \centering
    \includegraphics[width=\linewidth]{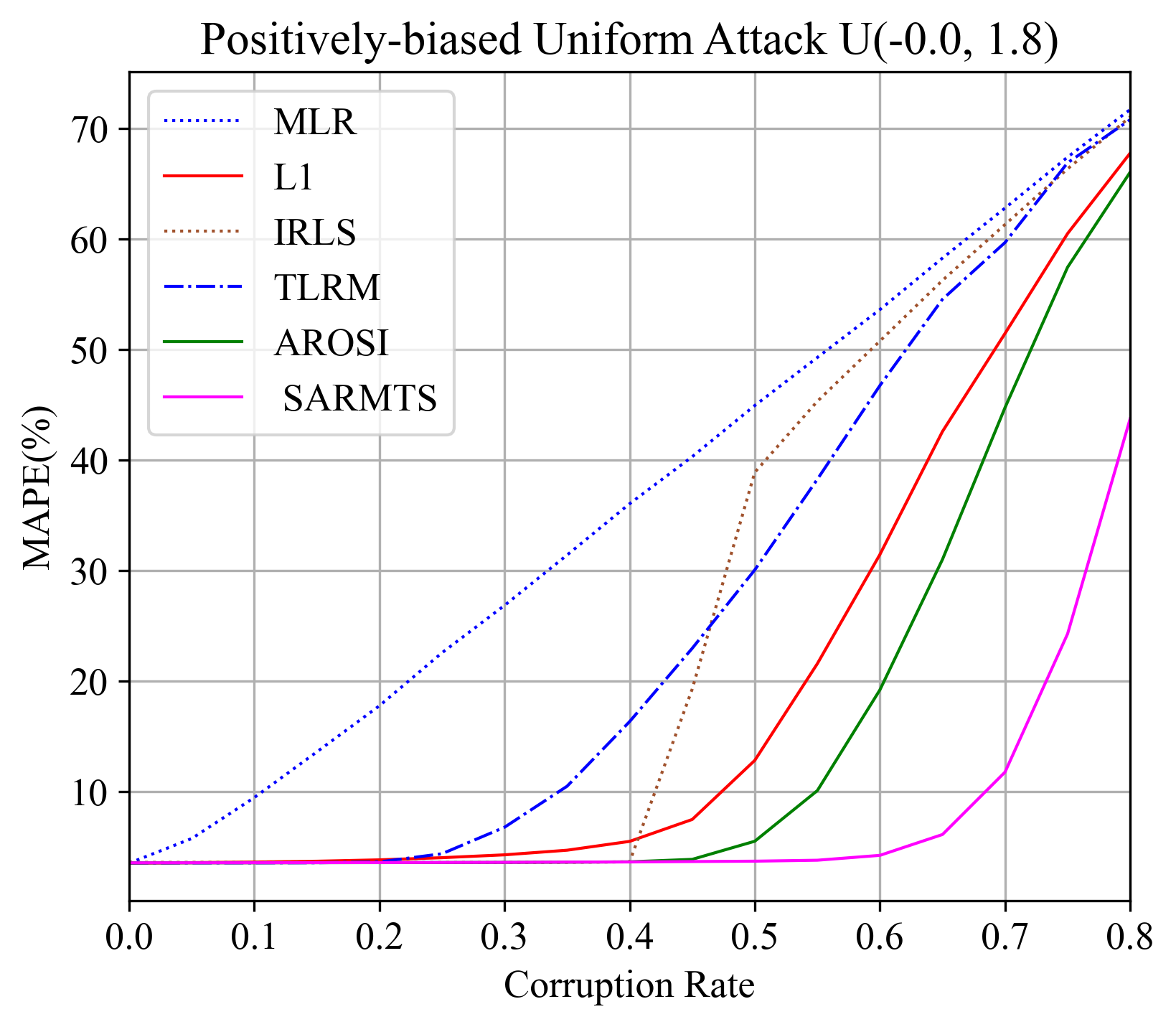}
  \end{subfigure}
  \hfill
  \begin{subfigure}[b]{0.33\textwidth}
    \centering
    \includegraphics[width=\linewidth]{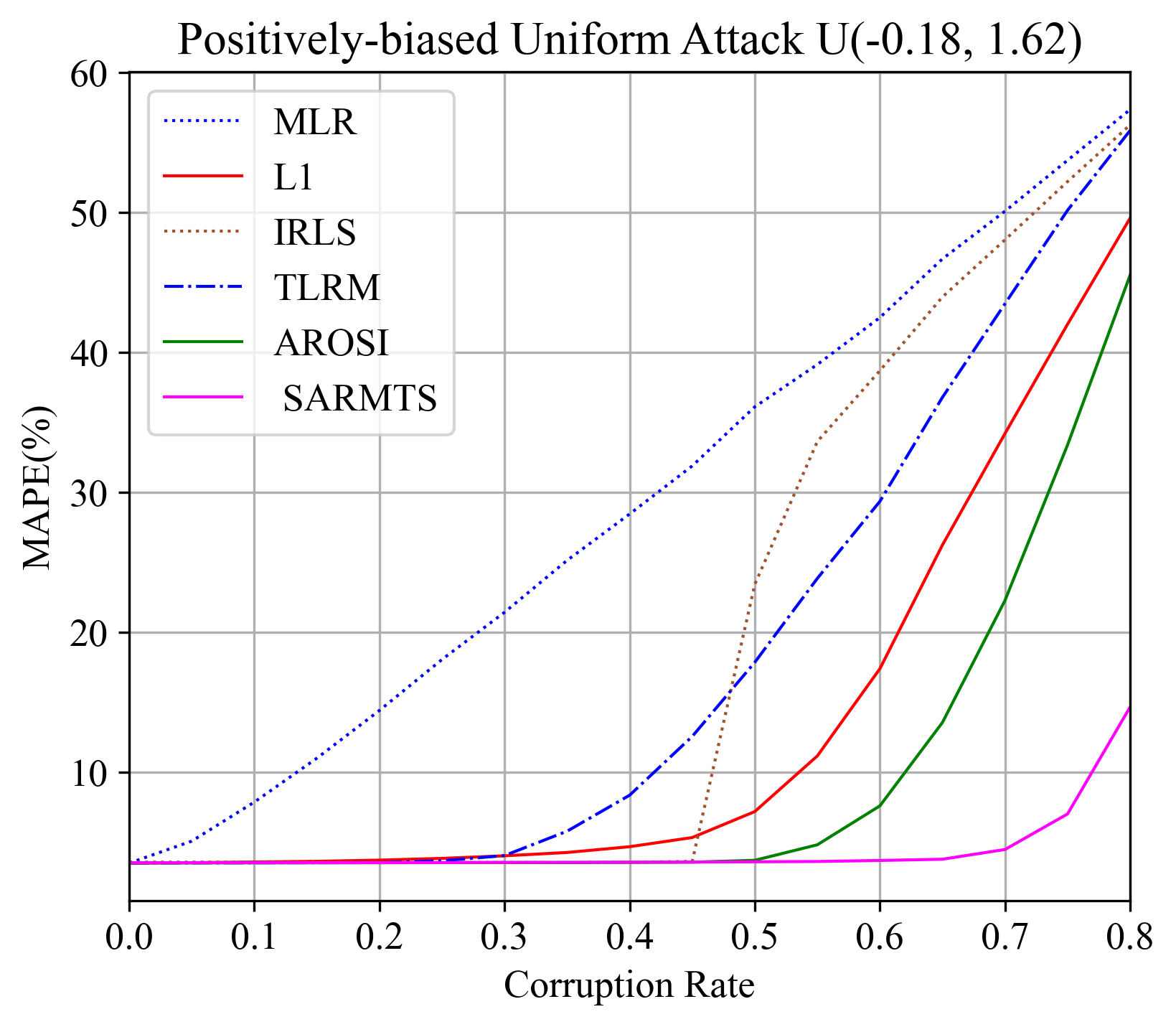}
  \end{subfigure}
  \hfill
  \begin{subfigure}[b]{0.33\textwidth}
    \centering
    \includegraphics[width=\linewidth]{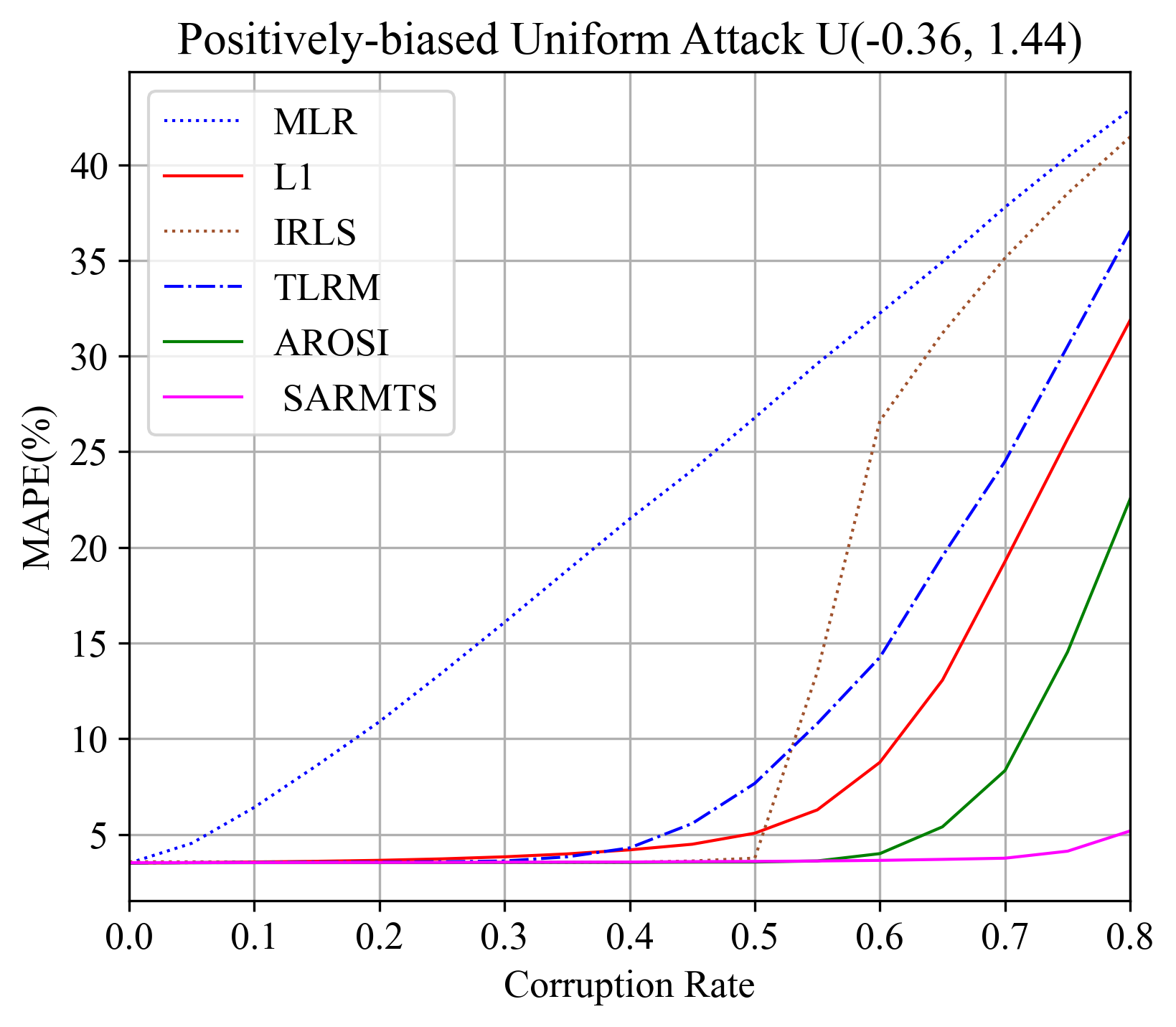}
  \end{subfigure}
  \caption{Positively-biased Uniform Attacks}
\end{figure}

\begin{figure}[htbp]
  \centering
  \begin{subfigure}[b]{0.33\textwidth}
    \centering
    \includegraphics[width=\linewidth]{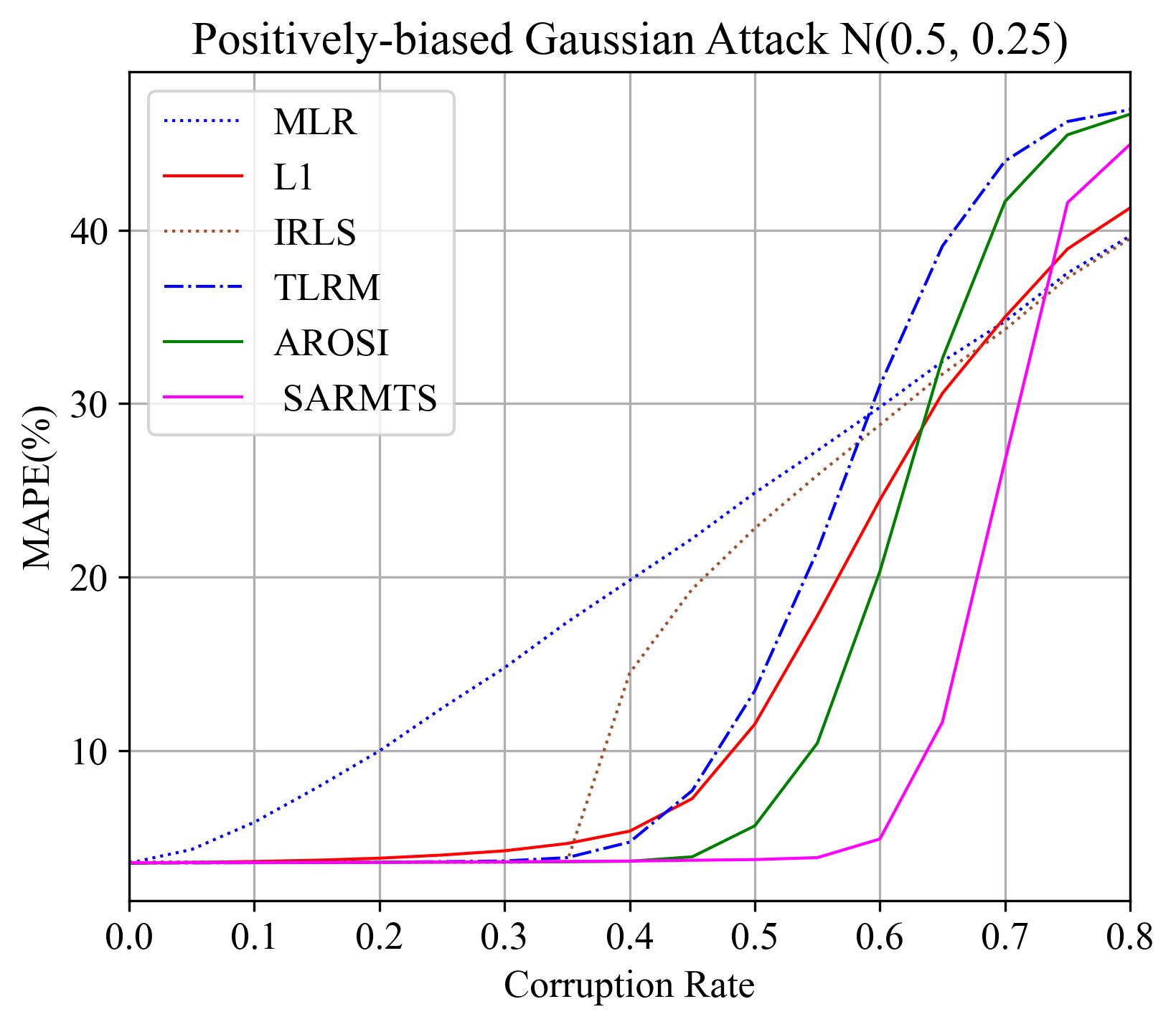}
  \end{subfigure}
  \hfill
  \begin{subfigure}[b]{0.33\textwidth}
    \centering
    \includegraphics[width=\linewidth]{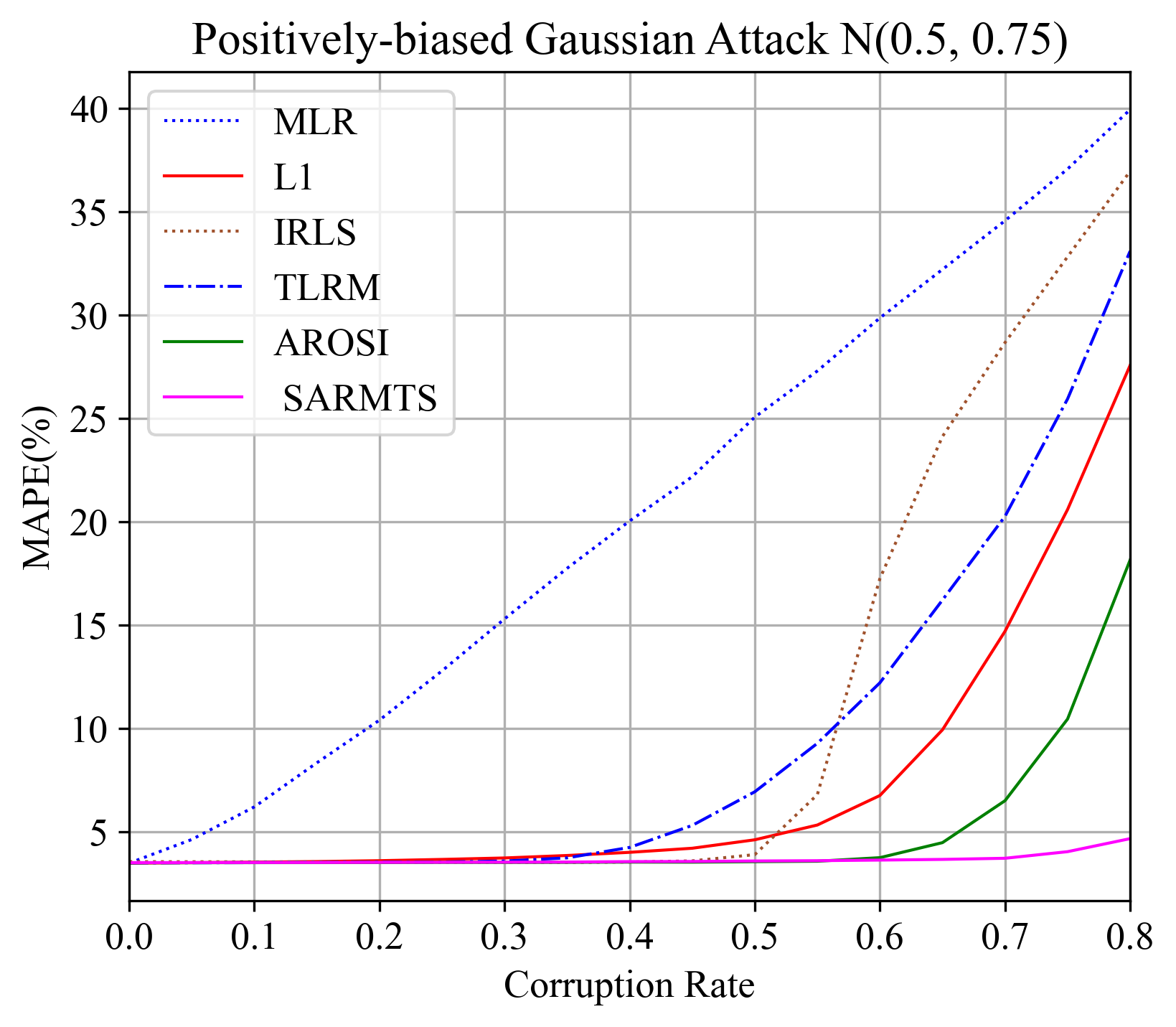}
  \end{subfigure}
  \hfill
  \begin{subfigure}[b]{0.33\textwidth}
    \centering
    \includegraphics[width=\linewidth]{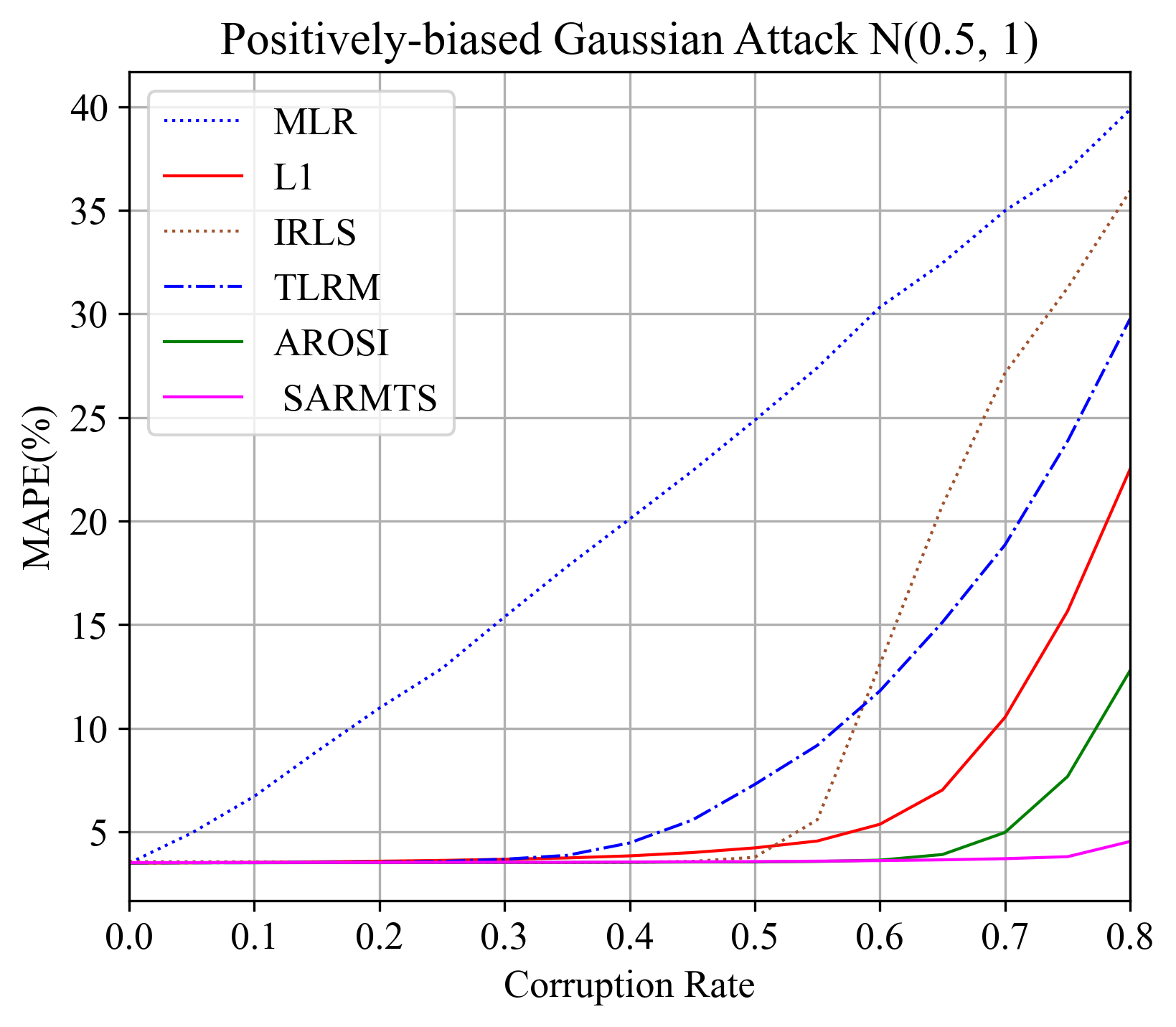}
  \end{subfigure}

  \vspace{1em}

  \begin{subfigure}[b]{0.33\textwidth}
    \centering
    \includegraphics[width=\linewidth]{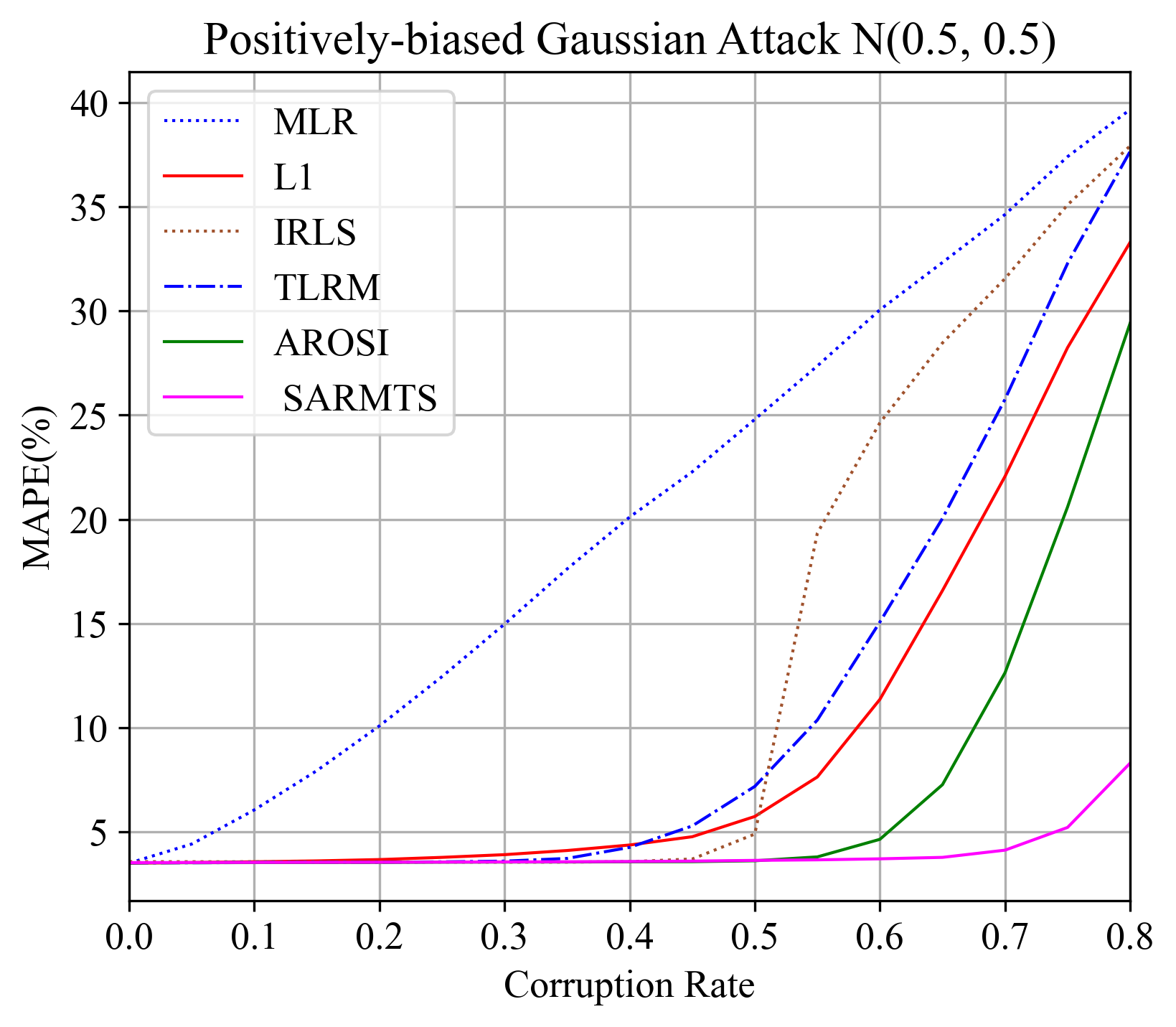}
  \end{subfigure}
  \hfill
  \begin{subfigure}[b]{0.33\textwidth}
    \centering
    \includegraphics[width=\linewidth]{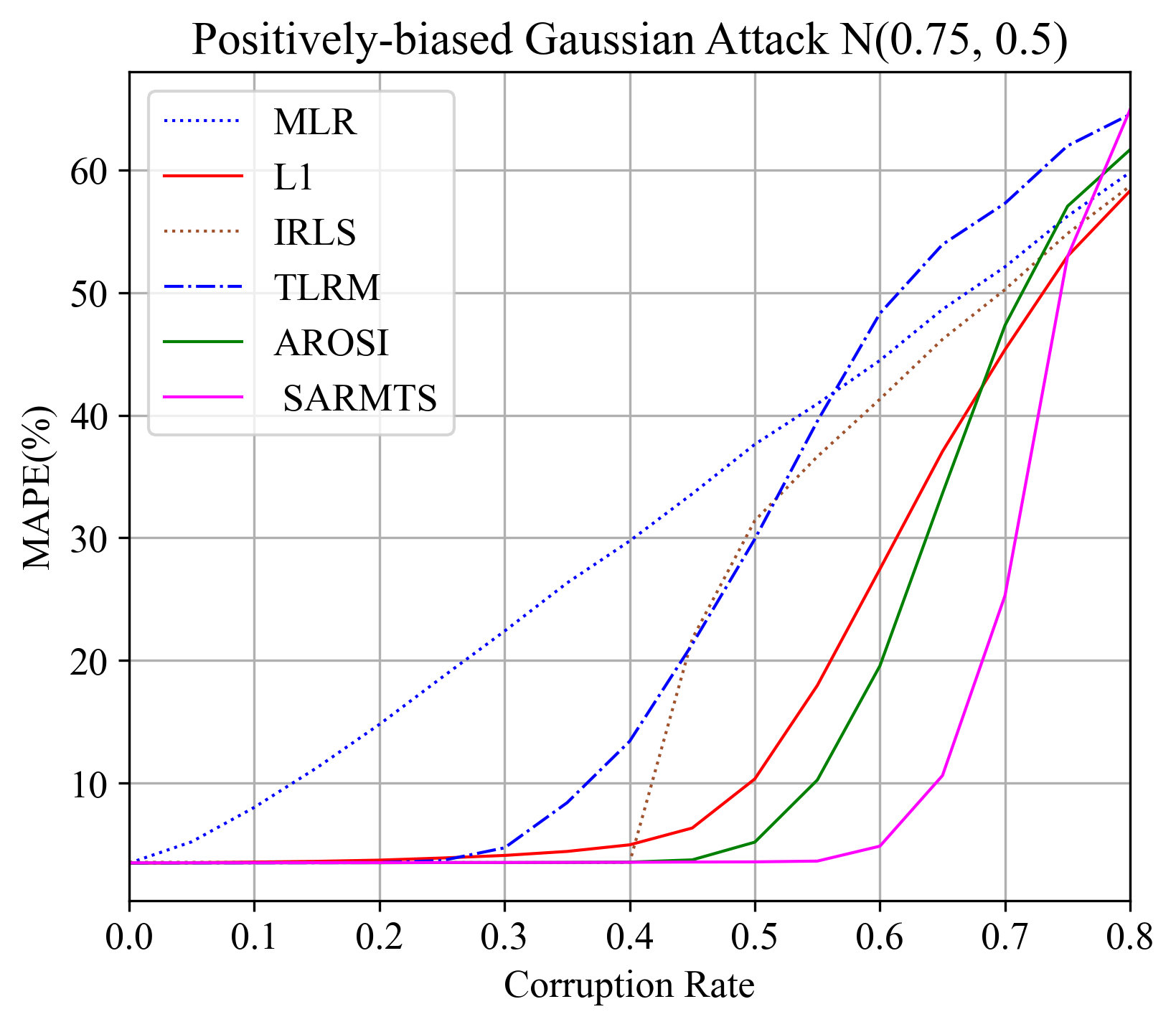}
  \end{subfigure}
  \hfill
  \begin{subfigure}[b]{0.33\textwidth}
    \centering
    \includegraphics[width=\linewidth]{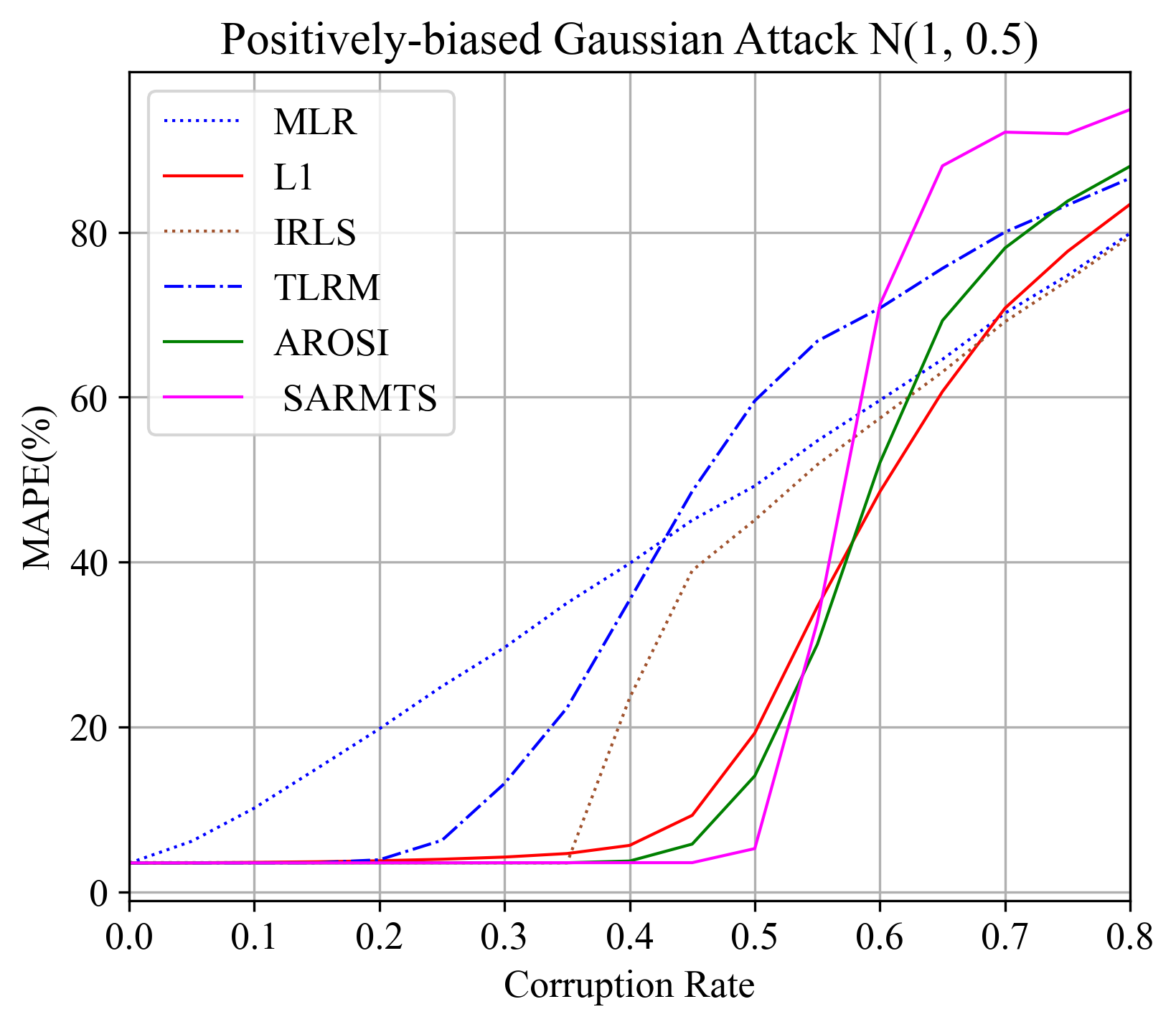}
  \end{subfigure}
  \caption{Positively-biased Gaussian Attacks}
\end{figure}

\begin{figure}[htbp]
  \centering
  \begin{subfigure}[b]{0.33\textwidth}
    \centering
    \includegraphics[width=\linewidth]{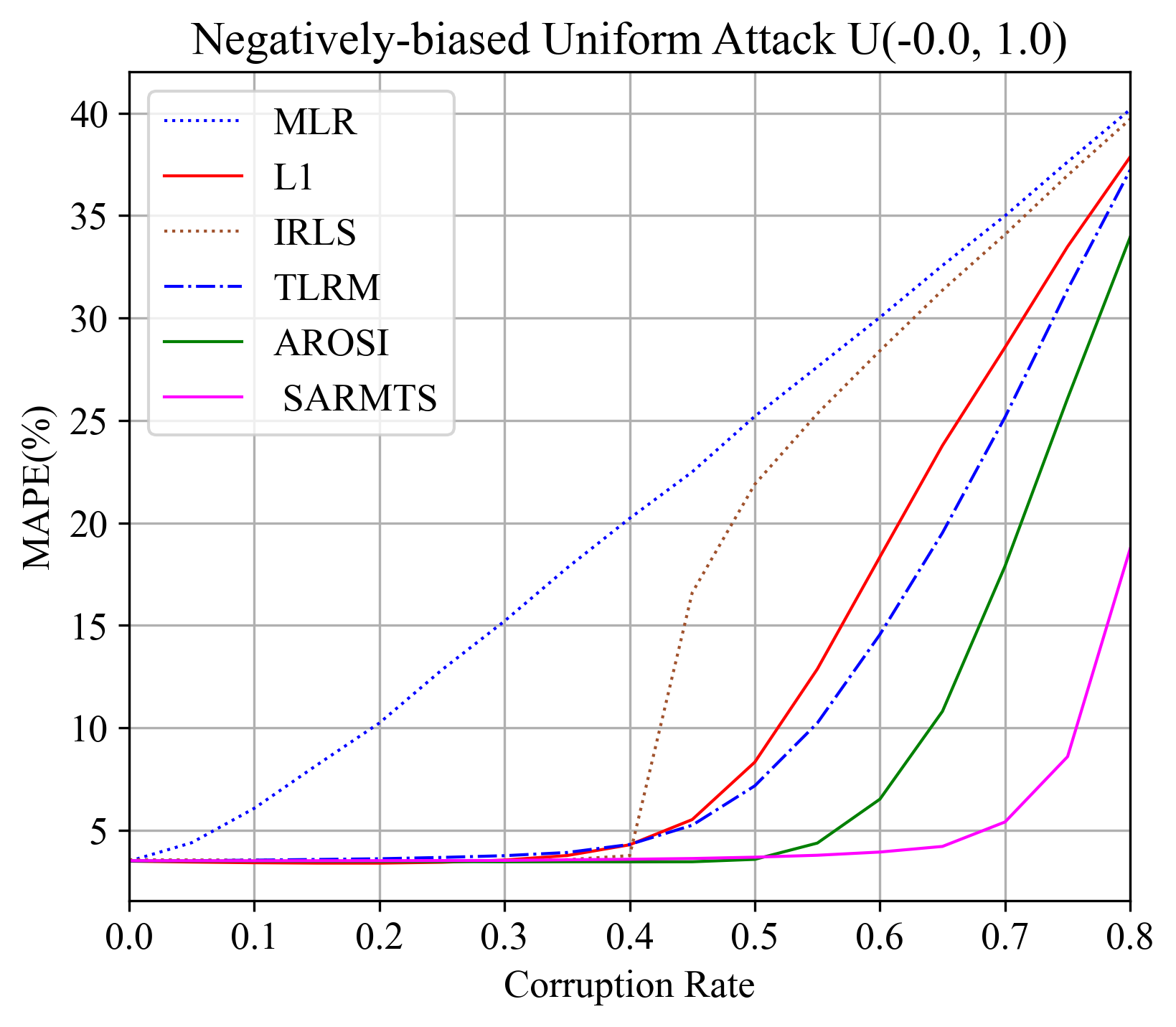}
  \end{subfigure}
  \hfill
  \begin{subfigure}[b]{0.33\textwidth}
    \centering
    \includegraphics[width=\linewidth]{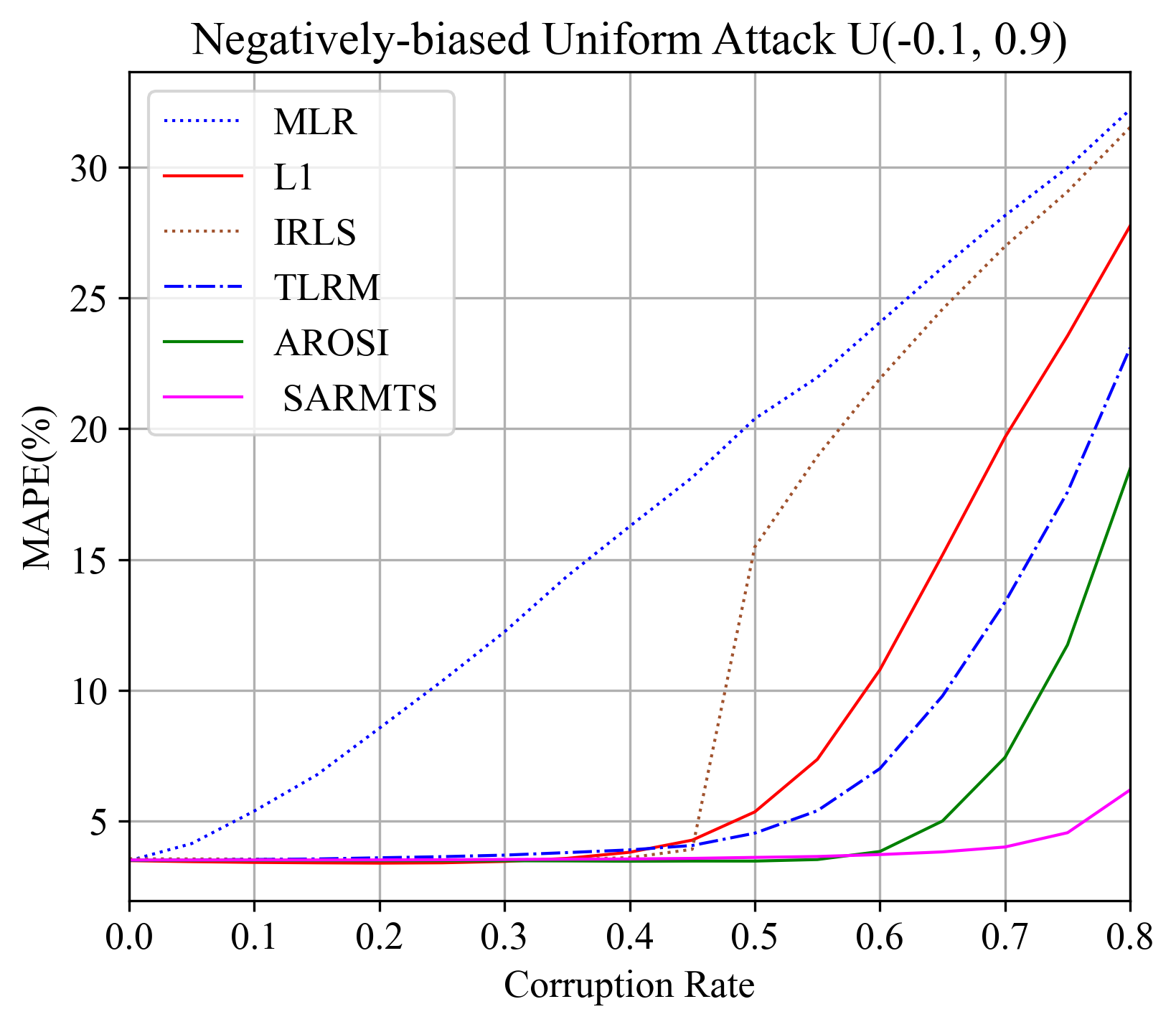}
  \end{subfigure}
  \hfill
  \begin{subfigure}[b]{0.33\textwidth}
    \centering
    \includegraphics[width=\linewidth]{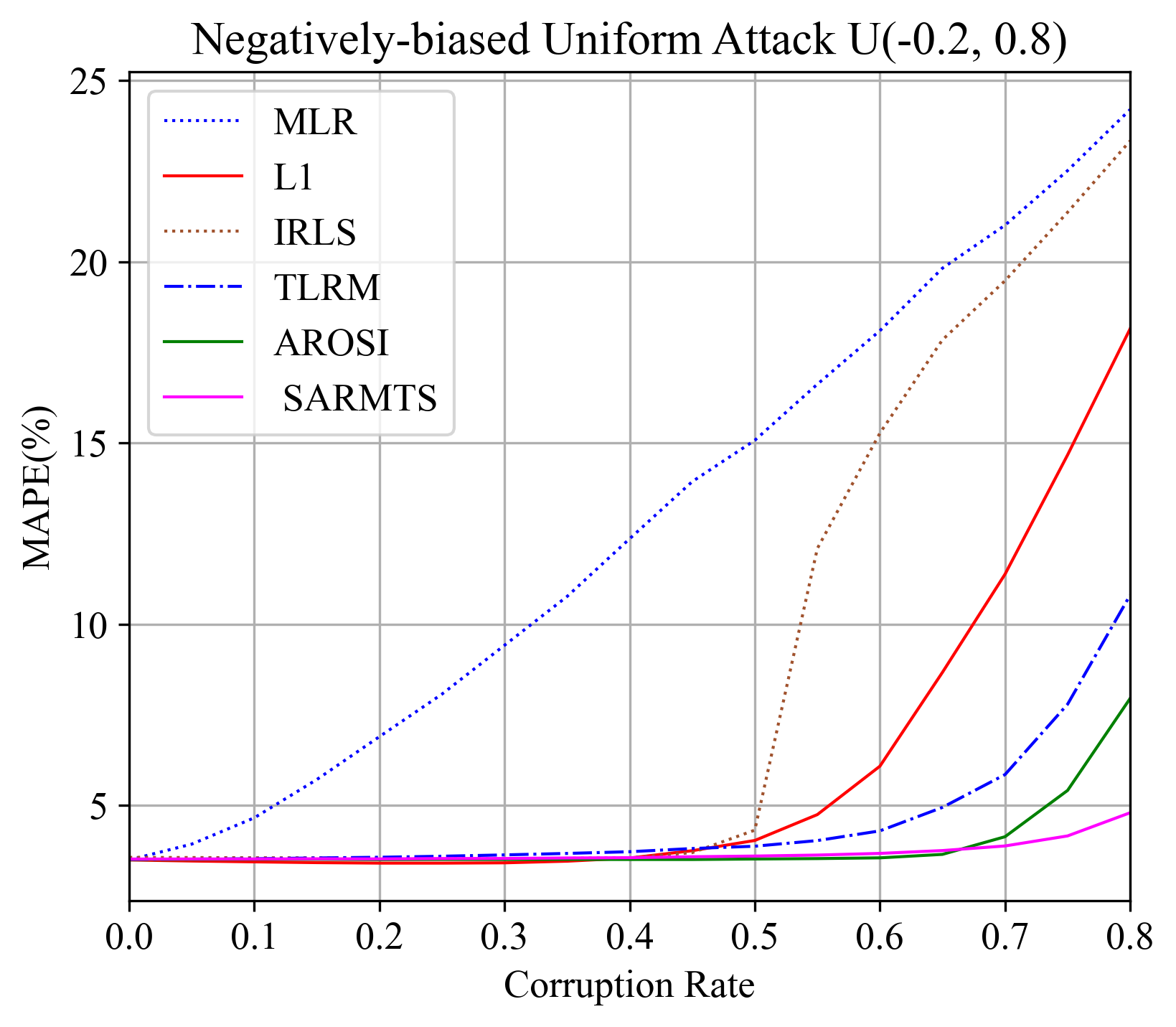}
  \end{subfigure}
  \caption{Negatively-biased Uniform Attacks}
\end{figure}

\section{Conclusion}
Your conclusion here

\section*{Acknowledgments}
This was supported in part by......

\section*{Appendix A: Proof of Convergence of SARM}

\subsection{The Proof of Theorem \ref{Sufficient decrease}}

The proof framework for general non-convex optimization problems typically requires that \( H(w, z) \) is differentiable with respect to all variables $w$ and $z$ and is a L-smooth function with respect to $w$ and $z$ respectively in order to ensure the existence of a quadratic upper bound. However, \( H(w, z) \) clearly does not satisfy this condition. Therefore, we need to prove an important lemma.
\begin{lemma} \textbf{(Quadratic Upper Bound)}
For any $w^{k+1}, w^k, z^{k+1}, z^k$, there exists a constant \( L > 0 \) such that the following inequality holds:
\begin{equation}
    H(w^{k+1}, z^k) \leq H(w^k, z^k) + \nabla_w H(w^k, z^k)^T (w^{k+1} - w^k) + \frac{L}{2} \| w^{k+1} - w^k \|_2^2,
\end{equation}
\label{Quadratic Upper Bound}
\end{lemma}

\begin{proof}
Let $g(t) = H(w^k + t(w^{k+1} - w^k), z^k), t \in [0,1]$. For simplicity, we denote $ H(w, z^k) $ as $ H(w) $. $S(x)$ (\ref{SMOOTHF}) is differentiable $\Rightarrow$ $H(w)$ is differentiable with respect to $ w $ $\Rightarrow$  $g(t)$ is differentiable with respect to $ t $. Then it can be obtained that:
\[
g'(t) = \nabla_w H(w^k + t(w^{k+1} - w^k))^T (w^{k+1} - w^k).
\]
Because $g(1) - g(0) = \int_0^1 g'(t) \, dt$, the following equation holds:
\begin{align*}
H(w^{k+1}) - H(w^k) &= \int_0^1 \nabla_w H(w^k + t(w^{k+1} - w^k))^T (w^{k+1} - w^k) \, dt. \\
\end{align*}
Then:
\begin{align}
&H(w^{k+1}) - H(w^k) - \nabla_w H(w^k)^T (w^{k+1} - w^k)  \notag\\
= & \int_0^1 g'(t) - g'(0) \, dt \notag\\
= &\int_0^1 \left[\nabla_w H(w^k + t(w^{k+1} - w^k)) - \nabla_w H(w^k)\right]^T (w^{k+1} - w^k) \, dt \notag\\
\leq & \int_0^1 \left\| \nabla_w H(w^k + t(w^{k+1} - w^k)) - \nabla_w H(w^k) \right\|_2 \| w^{k+1} - w^k \|_2 \, dt  \label{quadratic1}\\
\leq & L \int_0^1 \| w^{k+1} - w^k \|_2^2 \, dt  \label{quadratic2}\\
= & \frac{L}{2}\| w^{k+1} - w^k \|_2^2 \notag
\end{align}
The above procedure constitutes a classical proof of the quadratic upper bound. To complete the argument, it suffices to show the existence of a constant \( L > 0 \) such that  
\begin{align} \label{lsmooth}
\left\| \nabla_w H(w^k + t(w^{k+1} - w^k)) - \nabla_w H(w^k) \right\|_2 \leq L \| w^{k+1} - w^k \|_2,
\end{align}
which implies that (\ref{quadratic1}) \(\implies\) (\ref{quadratic2}). We now proceed to prove this assertion.

For simplicity, we denote \(\tilde{w}^k = w^k + t(w^{k+1} - w^k) \). According to the definition of \( H(w) \),
\begin{align}
\nabla_w H(w) = X^T X w - X^Ty - X^Tz^k + \delta X^T \left[ \frac{|z^k_i| \cdot S'(y_i - x_i^Tw)}{S^2(y_i - x_i^Tw)} \right]. \label{gredient}
\end{align}
As before, we simplify the notation by letting $r^k_i = y_i - x_i^Tw^k$ and $\tilde{r}^k_i = y_i - x_i^T\tilde{w}^k$.
Then according to the triangle inequality for norms:
\begin{align}
&\| \nabla_w H(\tilde{w}^k) - \nabla_w H(w^k) \|_2  \notag  \\
=& \left\| X^T X (\tilde{w}^k - w^k) + \delta X^T \left( \left[ |z^k_i| \cdot \left(  \frac{ S'(y_i - x_i^T\tilde{w}^k)}{S^2(y_i - x_i^T\tilde{w}^k)} -\frac{S'(y_i - x_i^Tw^k)}{S^2(y_i - x_i^Tw^k)} \right)\right]  \right) \right\|_2 \label{start1} \\
=& \left\| X^T X (\tilde{w}^k - w^k) + \delta X^T \left( \left[ |z^k_i| \cdot \left(  \frac{ S'(\tilde{r}^k_i)}{S^2(\tilde{r}^k_i)} -\frac{S'(r^k_i)}{S^2(r^k_i)} \right)\right]  \right) \right\|_2  \notag  \\
\leq & \| X^T X \|_2 \|\tilde{w}^k - w^k \|_2 + \delta  \| X^T \|_2 \left\| \left[ \frac{|z^k_i|}{S(r^k_i)} \cdot S(r^k_i) \cdot \left(  \frac{S'(\tilde{r}^k_i)}{S^2(\tilde{r}^k_i)} -\frac{ S'(r^k_i)}{S^2(r^k_i)} \right) \right] \right\|_2   \notag  \\
= &  \|\tilde{w}^k - w^k \|_2 + \delta  \left\|  \left[ \left| \frac{|z^k_i|}{S(r^k_i)} \right|\cdot \left( \frac{S(r^k_i)}{S(\tilde{r}^k_i)} \cdot \frac{S'(\tilde{r}^k_i)}{S(\tilde{r}^k_i)} -\frac{ S'(r^k_i)}{S(r^k_i)} \right) \right] \right\|_2 \label{zscale1} \\
\leq &  \|\tilde{w}^k - w^k \|_2 + \delta  \left\| \left[ \frac{S(r^k_i)}{S(\tilde{r}^k_i)} \cdot \frac{S'(\tilde{r}^k_i)}{S(\tilde{r}^k_i)} -\frac{ S'(r^k_i)}{S(r^k_i)} \right] \right\|_2 \label{zscale2} \\
\leq & \|\tilde{w}^k - w^k \|_2 + \delta \left\| \left[ \frac{S(r^k_i)-S(\tilde{r}^k_i)}{S(\tilde{r}^k_i)} \cdot \frac{S'(\tilde{r}^k_i)}{S(\tilde{r}^k_i)} + \frac{S'(\tilde{r}^k_i)}{S(\tilde{r}^k_i)} -\frac{ S'(r^k_i)}{S(r^k_i)} \right] \right\|_2  \notag  \\
\leq & \|\tilde{w}^k - w^k \|_2 + \delta \left\| \left[ \left| S(r^k_i)-S(\tilde{r}^k_i) \right| \cdot \left| \frac{S'(\tilde{r}^k_i)}{S^2(\tilde{r}^k_i)} \right|\right] \right\|_2 + \delta \left\| \left[ \frac{S'(\tilde{r}^k_i)}{S(\tilde{r}^k_i)} -\frac{ S'(r^k_i)}{S(r^k_i)} \right] \right\|_2.  \notag 
\end{align}

Note that $\| X^T X \|_2 = \|X \|_2 = 1$ because $X$ has been preconditioned as (\ref{pre}). \eqref{zscale1} \(\implies\) \eqref{zscale2} follows from the fact that $\left| \frac{T(x) }{S(x)} \right| \leq 1$, where:
\begin{align}
T(x) =
    \begin{cases}
  0 , & \text{if } x < \sqrt{\delta} \\
  x-\frac{1}{x}, & \text{if } x \geq \sqrt{\delta}.
\end{cases} 
\label{TX}
\end{align}
By setting \( x = y_i - x_i^T w \), we have $\frac{|z^k_i|}{S(y_i - x_i^T w)} \leq 1 $ and \eqref{zscale1} \(\implies\) \eqref{zscale2} holds.
The functions \( T(x) \) and \( S(x) \) are illustrated in Fig. \ref{function1} for better understanding.

Next, we examine the following functions:
\begin{equation}
S(x) =
\begin{cases}
  \frac{1}{ 2 \sqrt{\delta} } x^2 + \frac{ \sqrt{\delta} }{2}, & \text{if } |x| < \sqrt{\delta} \\
  | x |, & \text{if } |x| \geq \sqrt{\delta},
\end{cases}
\end{equation}
\begin{align}
    \gamma(x) = \frac{S'(x)}{S^2(x)} = \begin{cases}
  \frac{4\sqrt{\delta}x}{ (x^2 + \delta)^2 } , & \text{if } |x| < \sqrt{\delta} \\
  \frac{sign(x)}{x^2}, & \text{if } |x| \geq \sqrt{\delta},
\end{cases}
\end{align}
\begin{align}
    \kappa(x) = \frac{S'(x)}{S(x)} = \begin{cases}
  \frac{2x}{ x^2 + \delta } , & \text{if } |x| < \sqrt{\delta} \\
  \frac{sign(x)}{|x|} = \frac{1}{x}, & \text{if } |x| \geq \sqrt{\delta},
\end{cases}
\end{align}
It is easy to verify that for any $x_1, x_2, x$, there exist constants \( L_1, L_2, L_3 > 0 \) such that:
\begin{align} \label{ineq1}
    |S(x_1) - S(x_2)| \leq L_1 |x_1 - x_2| = L_1 |x_2 - x_1|,
\end{align}
\begin{align} \label{ineq2}
    |\gamma(x)| \leq L_2,
\end{align}
\begin{align} \label{ineq3}
    |\kappa(x_1) - \kappa(x_2)| \leq L_3 |x_1 - x_2| = L_3 |x_2 - x_1|,
\end{align}
holds for any $x_1, x_2 \in \mathbb{R}$.
Although the proof is straightforward, it is somewhat tedious. We therefore omit the detailed steps, but provide the corresponding function plots as Fig. \ref{function} for reference.
\begin{figure}[htbp]
  \centering
  \begin{subfigure}[b]{0.335\textwidth}
    \centering
    \includegraphics[width=\linewidth]{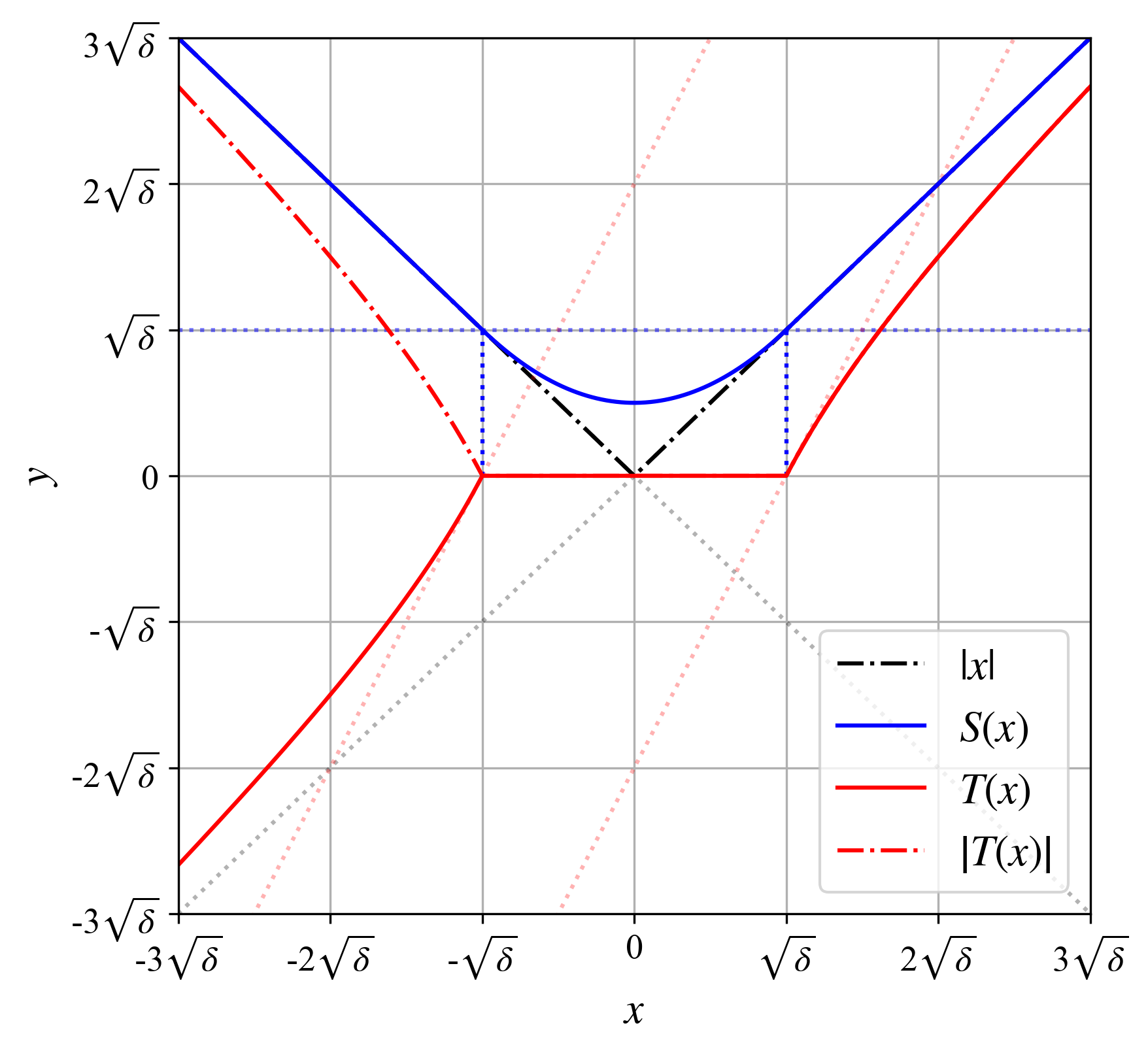}
    \caption{$S(x) $ and $ T(x)$}
    \label{function1}
  \end{subfigure}
  \hfill
  \begin{subfigure}[b]{0.335\textwidth}
    \centering
    \includegraphics[width=\linewidth]{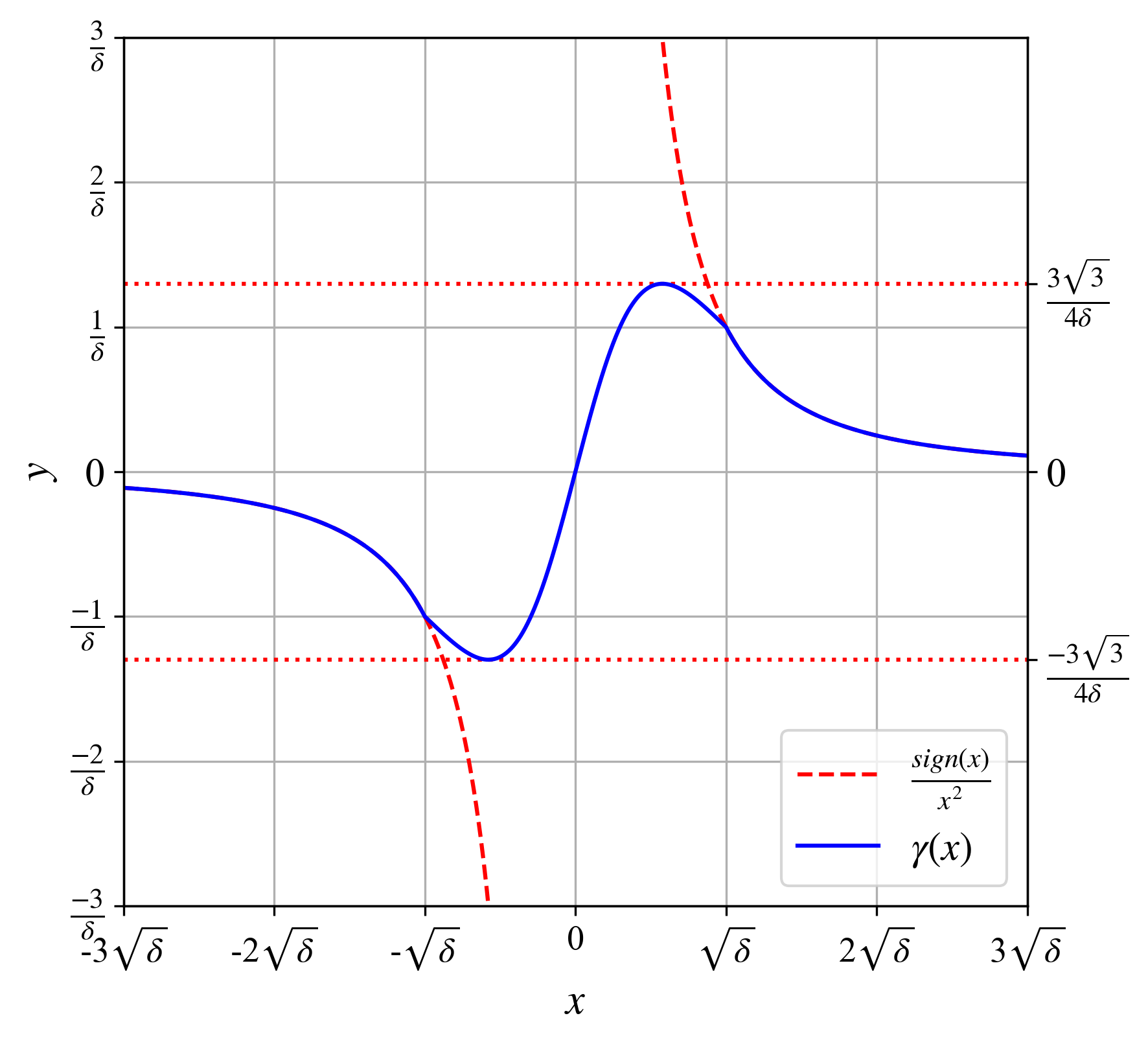}
    \caption{$\gamma(x)$}
  \end{subfigure}
  \hfill
  \begin{subfigure}[b]{0.32\textwidth}
    \centering
    \includegraphics[width=\linewidth]{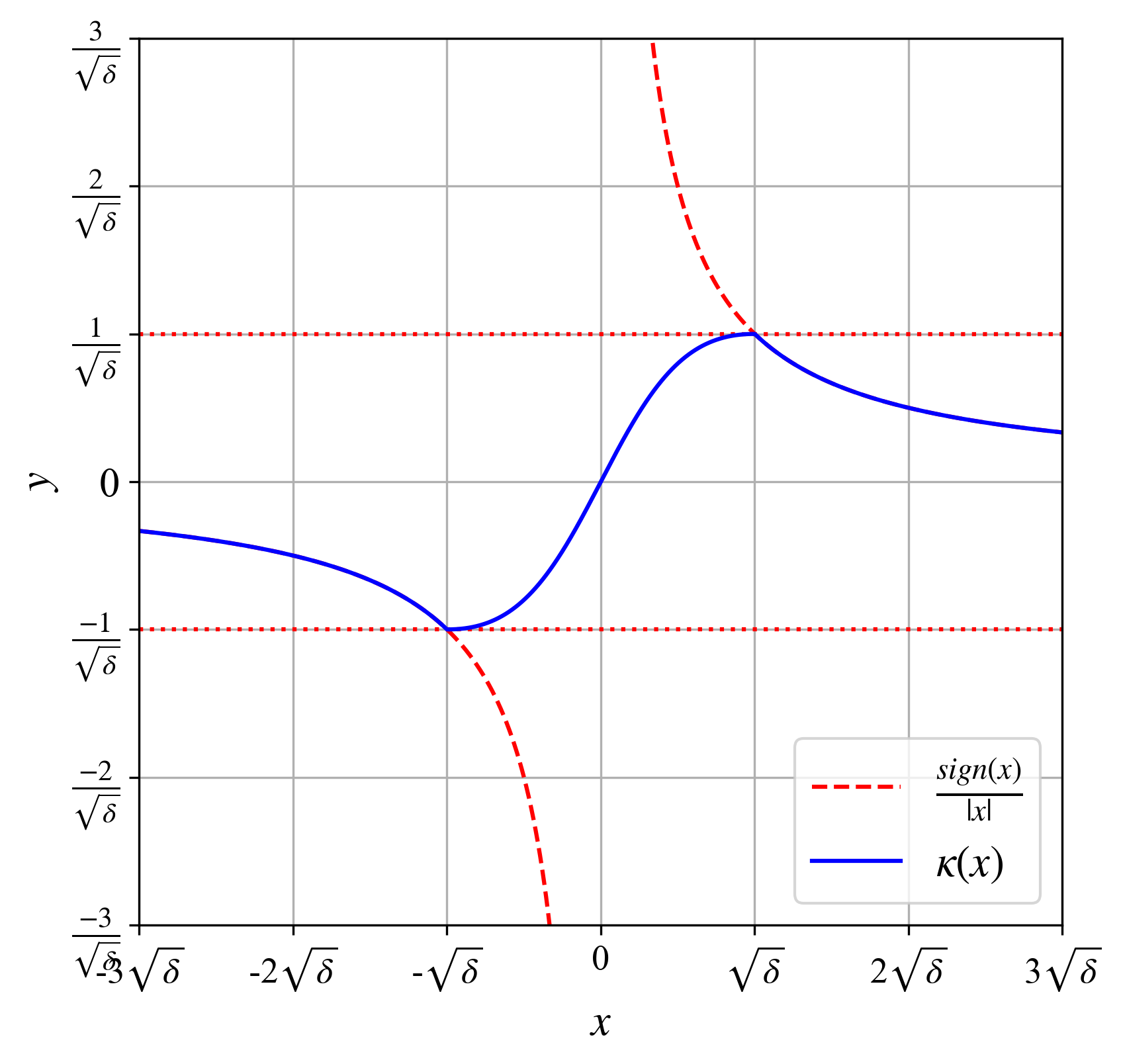}
    \caption{$\kappa(x)$}
  \end{subfigure}
  \caption{The graphs of $S(x) $, $ T(x)$, $\gamma(x)$ and $\kappa(x)$.}
  \label{function}
\end{figure}

 According to (\ref{ineq1}), (\ref{ineq2}), and (\ref{ineq3}), 
 \begin{align}
&\| \nabla_w H(\tilde{w}^k) - \nabla_w H(w^k) \|_2 \notag \\
\leq & \|\tilde{w}^k - w^k \|_2 + \delta \left\| \left[  L_1L_2  \left|\tilde{r}^k_i - r^k_i \right|\right] \right\|_2 + \delta \left\| \left[  L_3  \left(\tilde{r}^k_i - r^k_i \right)\right] \right\|_2  \notag \\
\leq & \|\tilde{w}^k - w^k \|_2 + \delta L_1L_2\left\| \left[ \tilde{r}^k_i - r^k_i \right] \right\|_2 + \delta L_3 \left\| \left[  \tilde{r}^k_i - r^k_i \right] \right\|_2  \label{rtow1} \\
\leq & \|\tilde{w}^k - w^k \|_2 + \delta (L_1L_2 + L_3)\left\| \left[ x_i^T (\tilde{w}^k - w^k) \right] \right\|_2  \label{rtow2} \\
\leq & \|\tilde{w}^k - w^k \|_2 + \delta (L_1L_2 + L_3)\left\| X (\tilde{w}^k - w^k) \right\|_2  \notag \\
\leq & \|\tilde{w}^k - w^k \|_2 + \delta (L_1L_2 + L_3)\left\| X \right\|_2 \left\|(\tilde{w}^k - w^k) \right\|_2  \notag \\
= & (1 + \delta (L_1L_2 + L_3)) \left\|(\tilde{w}^k - w^k) \right\|_2 \label{start2},
\end{align}
where (\ref{rtow1}) $\implies$ (\ref{rtow2}) due to the fact that: $\left| r^k_i -\tilde{r}^k_i \right| = \left|(y_i - x_i^Tw^k) - (y_i - x_i^T\tilde{w}^k)\right| = \left|-x_i^Tw^k + x_i^T\tilde{w}^k\right|= \left|x_i^Tw^k - x_i^T\tilde{w}^k\right|$.
Letting \( L = (1 + \delta (L_1L_2 + L_3))  \), the inequality in (\ref{lsmooth}) holds.
\end{proof}

\begin{remark}
    It is worth noting that the constant \( L \) provided above is merely an upper bound and not necessarily the smallest possible $L$ satisfying  Inequality (\ref{lsmooth}). Lemma \ref{Quadratic Upper Bound} guarantees the existence of such an \( L \), and thus there must exist some \( L_0 \) such that \( L_0 = \inf \{L \in \mathbb{R}|L \text{ satisfies (\ref{lsmooth}) } \} \). 
    We construct a sequence \( \{L^{(k)}\} \) that satisfies (\ref{lsmooth}) and converges to \( L_0 \). By substituting \( L^{(k)} \) into (\ref{lsmooth}) and taking the limit on both sides of the inequality as \( k \to \infty \), we conclude that the limiting value \( L_0 \) also satisfies  Inequality (\ref{lsmooth}), hence \( L_0 \) can be regarded as a tighter choice of the constant \( L \).
\end{remark}

We next prove Theorem \ref{Sufficient decrease} based on Lemma \ref{Quadratic Upper Bound}.

\begin{proof}\textbf{(Proof of Theorem \ref{Sufficient decrease})}
By the optimality of \( z_{k+1} \), we obtain: $H(w_{k+1}, z_{k}) - H(w_{k+1}, z_{k+1}) \geq 0$. Using this fact together with Lemma \ref{Quadratic Upper Bound}, it follows that:
\begin{align}
H(w^k, z^k) - H(w^{k+1}, z^{k+1}) 
&= H(w^k, ^k) - H(w^{k+1}, z^{k}) + H(w^{k+1}, z^{k}) - H(w^{k+1}, z^{k+1})  \notag\\
&\geq H(w^k, z^k) - H(w^{k+1}, z^{k})  \notag\\
&\geq -\nabla_w H(w^k, z^k)^T (w^{k+1} - w^k) - \frac{L}{2} \| w^{k+1} - w^k \|_2^2 \label{nabla1} \\
&\geq \frac{1}{\alpha}\| w^{k+1} - w^k \|_2^2- \frac{L}{2} \| w^{k+1} - w^k \|_2^2  \label{nabla2} \\
&\geq \left( \frac{1}{\alpha} - \frac{L}{2} \right) \| w^{k+1} - w^k \|_2^2   \notag
\end{align}
where (\ref{nabla1}) \(\implies\) (\ref{nabla2}) due to the fact that $w^{k+1} = w^{k} - \alpha \nabla_w H(w^k,z^k)$.
We can also obtained that:
\begin{align}
  0 \leq \alpha \leq \frac{2}{L} \Rightarrow  \frac{1}{\alpha} - \frac{L}{2} \geq 0 \Rightarrow H(w^k, z^k) - H(w^{k+1}, z^{k+1}) \geq 0.
  \label{monotonicity}
\end{align}
Meanwhile, it is easy to verify that $\left| T(x_2) - T(x_1) \right| \leq 2\left| x_2 - x_1 \right|$ (see Fig. \ref{function1}), where $T(x)$ is defined as (\ref{TX}).
Therefore:
\begin{align}
  \| z^{k+1} - z^{k} \|_2 &= \left\|T(y - Xw^{k+1}) - T(y - Xw^k) \right\|_2 \notag \\
&\leq 2 \left\|y - Xw^{k+1} + Xw^k - y \right\|_2 \notag \\
  &= 2 \left\|Xw^{k+1} - Xw^k \right\|_2 \notag \\
  &\leq 2 \left\|X\|_2 \|w^{k+1} - w^k \right\|_2 \notag \\
    &\leq 2 \left\|w^{k+1} - w^k \right\|_2. \notag 
\end{align}
Finally, we let $\rho_1 = \frac{1}{5}\left( \frac{1}{\alpha} - \frac{L}{2} \right)$ and obtain that for any $k$:
\begin{align}
    \rho_1 \left(\left\|[w^{k+1}; z^{k+1}] - [w^{k}; z^{k}]\right\|^2_2 \right) &= 
    \frac{1}{5}\left( \frac{1}{\alpha} - \frac{L}{2} \right) (\|w^{k+1} - w^k\|_2^2 + \| z^{k+1} - z^{k} \|_2^2 ) \notag \\
    &\leq\frac{1}{5}\left( \frac{1}{\alpha} - \frac{L}{2} \right) (\|w^{k+1} - w^k\|_2^2 + 4\| w^{k+1} - w^{k} \|_2^2 ) \notag\\
    &=\left( \frac{1}{\alpha} - \frac{L}{2} \right)\|w^{k+1} - w^k\|_2^2 \notag\\
    &\leq H \left( w^{k+1}, z^{k+1} \right) - H\left(w^{k}, z^{k}\right). \notag
\end{align}
\end{proof}

\begin{remark}
Up to this point, we can establish another important result. 
First, according to (\ref{monotonicity}), \( H(w^k, z^k) \) is monotonically decreasing with respect to $k$. 
Moreover, it is evident that $ H(w, z) \geq 0 $, and thus \( H(w^k, z^k) \) converges to a nonnegative real value \( H^* \). 
Furthermore, by Theorem \ref{Sufficient decrease}, we obtain:
\begin{align}
  \sum_{k=0}^{N-1} \left(\left\|[w^{k+1}; z^{k+1}] - [w^{k}; z^{k}]\right\|^2_2 \right) &\leq \frac{1}{\rho_1} \left( H(w^0,z^0) - H(w^N,z^N)\right) \\
  &\leq \frac{1}{\rho_1} \left( H(w^0,z^0) - H^*\right), 
\end{align}
where $N $ is an arbitrary positive integer. We let $N \to \infty $ and then obtain that 
$\sum_{k=0}^{\infty} \left(\left\|[w^{k+1}; z^{k+1}] - [w^{k}; z^{k}]\right\|^2_2 \right) \leq \infty$ , which also means that:
\begin{align}
\left\|[w^{k+1}; z^{k+1}] - [w^{k}; z^{k}]\right\|_2 \to 0,
\end{align}
and a simple corollary is:
\begin{align}
  \|w^{k+1} - w^k\|_2 \to 0, \| z^{k+1} - z^{k} \|_2 \to 0
  \label{extrem}
\end{align}.
\end{remark}

\subsection{Proof of Theorem \ref{Subgradient Lower Bound}}

Next, we present the proof of Theorem \ref{Subgradient Lower Bound}. For convenience, we still denote $y_i - x_iw^k = r^k_i$.
\begin{proof}\textbf{(Proof of Theorem \ref{Subgradient Lower Bound})}
We first consider \( A^{k+1} = \nabla_w H(w^{k+1}, z^{k+1}) \). 
Note that, since \( H(w, z) \) is differentiable with respect to \( w \), the notation used here refers to the gradient $\nabla$, not the subgradient  $\partial$.
According to (\ref{gredient}):
\begin{align}
&\nabla_w H(w^{k+1},z^{k+1}) \notag  \\
=& X^T X w^{k+1} - X^Ty - X^Tz^{k+1} + \delta X^T \left[ \frac{|z^{k+1}_i| \cdot S'(r^{k+1}_i)}{S^2(r^{k+1}_i)} \right]  \notag \\
=& \frac{1}{\alpha}w^{k+1} - \frac{1}{\alpha} \left( I - \alpha X^T X \right)w^{k+1} - X^Ty - X^Tz^{k+1} + \delta X^T \left[ \frac{|z^{k+1}_i| \cdot S'(r^{k+1}_i)}{S^2(r^{k+1}_i)} \right]  \notag \\
=& \frac{1}{\alpha} \left(w^{k} - \alpha X^T X w^{k} - \alpha X^Ty - \alpha X^Tz^{k} + \alpha \delta X^T \left[ \frac{|z^{k}_i| \cdot S'(r^{k}_i)}{S^2(r^{k}_i)} \right] \right)  \notag  \notag \\
& - \frac{1}{\alpha} \left( I - \alpha X^T X \right)w^{k+1} - X^Ty - X^Tz^{k+1} + \delta X^T \left[ \frac{|z^{k+1}_i| \cdot S'(r^{k+1}_i)}{S^2(r^{k+1}_i)} \right]  \notag \\
=& \frac{1}{\alpha} \left( I - \alpha X^T X \right) \left( w^{k} - w^{k+1} \right) + \alpha X^T(z^{k} - z^{k+1}) + \delta X^T \left( \left[ \frac{|z^{k}_i| \cdot S'(r^{k}_i)}{S^2(r^{k}_i)} \right] - \left[ \frac{|z^{k+1}_i| \cdot S'(r^{k+1}_i)}{S^2(r^{k+1}_i)} \right] \right).  \notag
\end{align}
Therefore:
\begin{align}
 &\left\| \nabla_w H(w^{k+1},z^{k+1}) \right\|_2   \notag \\
\leq&\frac{1}{\alpha} \left\| I - \alpha X^T X \right\|_2 \left\| w^{k} - w^{k+1} \right\|_2 + \alpha \left\|X^T \right\|_2 \left\|z^{k} - z^{k+1}\right\|_2 + \delta \left\|X^T\right\|_2 \left\| \left[ \frac{|z^{k}_i| \cdot S'(r^{k}_i)}{S^2(r^{k}_i)} \right] - \left[ \frac{|z^{k+1}_i| \cdot S'(r^{k+1}_i)}{S^2(r^{k+1}_i)} \right] \right\|_2  \notag \\
\leq&\frac{1}{\alpha} (1-\alpha) \left\| w^{k} - w^{k+1} \right\|_2 + \alpha \left\|z^{k} - z^{k+1}\right\|_2 + \delta \left\| \left[ \frac{|z^{k}_i| \cdot S'(r^{k}_i)}{S^2(r^{k}_i)} \right] - \left[ \frac{|z^{k+1}_i| \cdot S'(r^{k+1}_i)}{S^2(r^{k+1}_i)} \right] \right\|_2.  \notag
\end{align}

The first two terms clearly tend to zero; we now analyze the last term.
\begin{align}
&\left\| \left[ \frac{|z^{k}_i| \cdot S'(r^{k}_i)}{S^2(r^{k}_i)} - \frac{|z^{k+1}_i| \cdot S'(r^{k+1}_i)}{S^2(r^{k+1}_i)} \right] \right\|_2  \notag \\
=&\left\| \left[ \frac{(|z^{k}_i| - |z^{k+1}_i|) \cdot S'(r^{k}_i)}{S^2(r^{k}_i)}\right] + \left[\frac{|z^{k+1}_i| \cdot S'(r^{k}_i)}{S^2(r^{k}_i)} - \frac{|z^{k+1}_i| \cdot S'(r^{k+1}_i)}{S^2(r^{k+1}_i)} \right] \right\|_2  \notag \\
\leq& \left\| \left[(|z^{k}_i| - |z^{k+1}_i|)  \cdot \frac{ S'(r^{k}_i)}{S^2(r^{k}_i)} \right] \right\|_2 + \left\| \left[|z^{k+1}_i| \cdot \left( \frac{ S'(r^{k}_i)}{S^2(r^{k}_i)} - \frac{ S'(r^{k+1}_i)}{S^2(r^{k+1}_i)} \right) \right] \right\|_2  \notag \\
\leq& L_2 \left\| \left[(|z^{k}_i| - |z^{k+1}_i|)\right] \right\|_2 + \left\| \left[|z^{k+1}_i| \cdot \left( \frac{ S'(r^{k}_i)}{S^2(r^{k}_i)} - \frac{ S'(r^{k+1}_i)}{S^2(r^{k+1}_i)} \right) \right] \right\|_2 \label{scalesimilar1}  \\
\leq& L_2 \left\| \left[(|z^{k}_i| - |z^{k+1}_i|)\right] \right\|_2 + (L_1L_2 + L_3)\left\| w^{k}_i - w^{k+1}_i \right\|_2 \label{scalesimilar2}\\
\leq& L_2 \left\| \left[z^{k}_i - z^{k+1}_i\right] \right\|_2 + (L_1L_2 + L_3)\left\| w^{k}_i - w^{k+1}_i \right\|_2, \label{scalesimilar3}
\end{align}
where (\ref{scalesimilar1}) \(\implies\) (\ref{scalesimilar2}) is similar to the proof procedure from (\ref{start1}) to (\ref{start2}).
(\ref{scalesimilar2}) \(\implies\) (\ref{scalesimilar3}) holds due to the fact that 
$\left|a - b \right| \geq \left||a| - |b| \right|, \forall a,b \in \mathbb{R}$.

Therefore we can bound \( \left\| \nabla_w H(w^{k+1},z^{k+1}) \right\|_2 \) like:
\begin{align}
 & \left\| \nabla_w H(w^{k+1},z^{k+1}) \right\|_2   \notag \\
\leq&\frac{1}{\alpha} (1-\alpha) \left\| w^{k} - w^{k+1} \right\|_2 + \alpha \left\|z^{k} - z^{k+1}\right\|_2 + \delta \left\| \left[ \frac{|z^{k}_i| \cdot S'(r^{k}_i)}{S^2(r^{k}_i)} \right] - \left[ \frac{|z^{k+1}_i| \cdot S'(r^{k+1}_i)}{S^2(r^{k+1}_i)} \right] \right\|_2  \notag \\
\leq& \left(\frac{1-\alpha}{\alpha} + \delta \left(L_1L_2 + L_3 \right) \right) \left\| w^{k} - w^{k+1} \right\|_2 + (\alpha + \delta L_2) \left\|z^{k} - z^{k+1}\right\|_2. \label{decentbound}
\end{align}

Next, we consider \(B^{k+1} \). $x_i^T w^{k+1} - y_i$ remains abbreviated as $ r_i^{k+1}$.
The subdifferential of \( H(w,z)\) with respect to \( z \) can be expressed as:
\begin{align}
\partial_z H(w^{k+1}, z^{k+1}) =& z^{k+1} - X w^{k+1} + y + \delta \left[ \frac{\partial |z^{k+1}_i|}{S(x_i^Tw- y_i)}\right]  \notag \\
=&     \begin{cases}
  r_i^{k+1} - \frac{\delta}{r_i^{k+1}} - r_i^{k+1} + \frac{\delta \cdot sign (z^{k+1})}{| r_i^{k+1}|}, & \text{if } | r_i^{k+1}| \geq \sqrt{\delta}\\
   - r_i^{k+1} + \frac{\delta \cdot \partial |z^{k+1}| }{S( r_i^{k+1})}, & \text{if } | r_i^{k+1}| < \sqrt{\delta}
\end{cases}  \notag \\
=&     \begin{cases}
  0, & \text{if } | r_i^{k+1}| \geq \sqrt{\delta}  \\
  \frac{\delta \cdot \partial |z^{k+1}| }{S( r_i^{k+1})} - r_i^{k+1}, & \text{if } | r_i^{k+1}| < \sqrt{\delta},
\end{cases}
\end{align}
where $+$ denotes sum of sets. 
Considering only the \( i \)-th component, when $| r_i^{k+1}| < \sqrt{\delta} $, we have $ z^{k+1} = 0 $ and $\partial |z^{k+1}| = [-1, 1]$.
We let $c = \frac{S(r_i^{k+1})}{\delta} \cdot r_i^{k+1}$, where $|c| \in [-1, 1]$ so $c \in \partial |z^{k+1}|$. 
Moreover, since $\frac{\delta \cdot c }{S( r_i^{k+1})} - r_i^{k+1} = 0$, we can obtain that if $| r_i^{k+1}| < \sqrt{\delta}$, $ 0 \in \partial_{z_i} H(w^{k+1}, z^{k+1})$.

Furthermore, combining with the situation where \( | r_i^{k+1} | \geq \sqrt{\delta} \) and (\ref{decentbound}), the following conclusion holds: 
$ \mathbf{0} \in \partial_{z} H(w^{k+1}, z^{k+1})$ holds for any $w^{k+1}, z^{k+1}$. 
By setting \( B^{k+1} = \mathbf{0} \), we obtain the conclusion:
\begin{align}
\left\|[A^{k+1}; B^{k+1}]\right\|_2 &= \left\|A^{k+1}\right\|_2   \notag \\
&= \left\| \nabla_w H(w^{k+1},z^{k+1}) \right\|_2   \notag \\
&\leq \left(\frac{1-\alpha}{\alpha} + \delta \left(L_1L_2 + L_3 \right) \right) \left\| w^{k} - w^{k+1} \right\|_2 + (\alpha + \delta L_2) \left\|z^{k} - z^{k+1}\right\|_2   \notag \\
&\leq \rho_2 \left\|[w^{k+1}; z^{k+1}] - [w^{k}; z^{k}] \right\|_2,  \notag
\end{align}
where $\rho_2$ is $ max \left\{\frac{1-\alpha}{\alpha} + \delta \left(L_1L_2 + L_3 \right), \alpha + \delta L_2 \right\}$.
\end{proof}

\begin{remark}
When $k \to \infty, \|w^{k+1} - w^k\|_2 \to 0$ and $ \| z^{k+1} - z^{k} \|_2 \to 0$ (see (\ref{extrem})), then we obtain \( \left\| \nabla_w H(w^{k+1},z^{k+1}) \right\|_2 \to 0\).
\end{remark}

\subsection{Properties of the Critical Points}

We now proceed to discuss the properties of the critical points. 
For the sake of simplicity, let \( u^k = [w^{k}; z^{k}]\).
Similar to \cite{bolte2014proximal, liu2020optimization}, we define \( \omega(u^0) \) as the set of limit points of the sequence $\{ u^k\}$ generated from Algorithm \ref{SARMAlgorithm} with \( u^0 \) as the initial point.

\begin{lemma} \label{Properties of the Critical Points}
\textbf{(Properties of the Critical Points)}: Supposing that $u^k$ obtained from Algorithm \ref{SARMAlgorithm} are bounded and 
$\alpha$ satisfied $0 \leq \alpha \leq \frac{2}{L}$. Then the following properties hold:
\begin{enumerate}
  \item $\omega(u^0) \subset \operatorname{crit} H$ and $\omega(u^0) \neq \emptyset$.
  \item $\lim_{k \to \infty} \operatorname{dist}(u^k, \omega(u^0)) = 0$. \label{Properties of the Critical Points 1}
  \item $\omega(u^0)$ is a nonempty, compact, and connected set. \label{Properties of the Critical Points 2}
  \item $H(u)$ is finite and constant on $\omega(u^0)$.
\end{enumerate}
\label{Properties of the Critical Points}
\end{lemma}

\begin{proof}
\begin{enumerate}
  \item $\omega(u^0) \neq \emptyset$ follows directly from the boundedness of $\{ u^k\}$. Let \( u^* \) be a limit point of the sequence \( \{u^k\} \), 
  which means that there exists a subsequence \( \{u^{k_q}\}_{q \in \mathbb{N}} \) such that \( u^{k_q} \to u^{*} \) as $q \to \infty$. 
  Since \( H(u) \), i.e. \( H(w,z) \), is a continuous function,
  \begin{align}
    \lim_{q \to \infty} H(u^{k_q}) = H(u^{*}).  \notag
  \end{align}
  Meanwhile, $ A^{k} = \nabla_w H(w^{k}, z^{k}) $ and $ B^{k} \in \partial_z H(w^{k}, z^{k}) $ ,
  so $ \left(A^{k},  B^{k}\right) \in \partial H(w^{k}, z^{k}) = \partial H(u^{k}) $. 
  According to  Theorem \ref{Subgradient Lower Bound}, as 
  $k \to \infty, \left\| A^{k} \right\|_2 \to 0 $ and $B^{k}=0 $, hence $\left( A^{k}, B^{k}\right) \to 0$.
  It is also true that $\partial H(u)$ is closed, which can be found in Remark 1 of \cite{bolte2014proximal}, hence $ 0 \in H(u^{*})$. 
  This proves that $\omega(u^0) \subset \operatorname{crit} H$. 
  \item As proved in Lemma 5 (ii) of \cite{bolte2014proximal}. 
  \item As proved in Lemma 5 (iii) of \cite{bolte2014proximal}.
  \item According to (\ref{monotonicity}), $\{H(u^{k})\}$ is a monotone sequence. We also have $H(x)\geq 0$. Let $H_0$ be the infimum of $\{H(u^{k})\}$. 
  Assume that $u^{*} \in \operatorname{crit} H$. There exists a subsequence \( \{u^{k_q}\}\) converging to $u^{*}$ as $q \to \infty$. 
  Based on the monotonicity  of $\{H(u^{k})\}$ with respect to $k$ and the continuity of \( H(u) \) with respect to $u$, we have the fact that 
  \( \{H(u^{k_q})\}\) converges to $H_0$ and \( H(u^{*})=H_0\).
\end{enumerate}
\end{proof}
Note that the assumptions of Lemma \ref{Properties of the Critical Points} are not the same as those of Lemma 5 in \cite{bolte2014proximal}.
The reason why Property \ref{Properties of the Critical Points 1} and Property \ref{Properties of the Critical Points 2} can be proved in the same way of 
Lemma 5 (ii) and (iii) in \cite{bolte2014proximal} is describe in Remark 5 of \cite{bolte2014proximal}.

A more detailed discussion of this part can be found in Lemma 5 of \cite{bolte2014proximal}.

\subsection{KL Property and Convergence of the Entire Sequence}

Up to this point, we have established that the sequence \(\{u_k\}\) admits convergent subsequences. 
However, the full convergence of the entire sequence—i.e., the Cauchy property—remains to be shown. 
To this end, we utilize the Kurdyka-Łojasiewicz (KL) property, 
which provides a powerful and well-established analytical framework for proving the global convergence of iterative algorithms. 
This methodology has been widely adopted in both nonconvex optimization theoretical algorithm analysis \cite{bolte2014proximal, attouch2010proximal,attouch2013convergence}
and practical applications \cite{linglobally}.

In the following, we proceed to apply the KL property to provide the proof of Theorem \ref{Apply KL Property}.

As a preliminary step, we introduce the definition of the KL property \cite{attouch2010proximal, attouch2013convergence}. 
\begin{definition} \label{KL Property}
\textbf{(KL Property)}
Let \( \sigma: \mathbb{R}^d \to (-\infty, +\infty] \) be a proper lower semicontinuous function. We say that \( \sigma \) 
satisfies the KL property at a point \( \bar{u} \in \mathrm{dom}(\partial \sigma) \left(\overset{\text{def}}{=} \{u| \partial \sigma(u) \neq \emptyset\} \right) \) 
if there exist constants \( \eta > 0 \), a neighborhood \( \mathcal{U} \) of \( \bar{u} \), 
and a continuous concave function \( \phi: [0, \eta) \to \mathbb{R}_{+} \) such that:
\( \phi(0) = 0 \),
\( \phi \) is \( C^1 \) on \( (0, \eta) \) and is continuous at \(0 \), with \( \phi' > 0 \) on \( (0, \eta) \),
and for all \( u \in \mathcal{U} \) satisfying
\begin{align}
  H(\bar{u}) < H(u) < H(\bar{u}) + \eta,  \notag
\end{align}
  the following inequality holds:
\begin{align}
  \phi'\big(\sigma(u) - \sigma(\bar{u})\big) \cdot \mathrm{dist}\left(0, \partial \sigma(u)\right) \geq 1,  \notag
\end{align}
where $ \mathrm{dist}(0,\partial \sigma(u))$ is defined as (\ref{dist}) and denotes the distance from the point $0$ to the set $\partial \sigma(u)$. 
If \( \sigma \) satisfies the KL property at each point of \( \mathrm{dom}(\partial \sigma) \), then \( \sigma \) is called a KL function.
\end{definition}
For a detailed discussion on the development of the KL property, please refer to Remark 4 in \cite{attouch2010proximal}.

We now show that
\begin{align}
    H(w, z) &= \frac{1}{2}\|y-Xw-z\|_2^2 + \delta \left\|\frac{z}{S(y-Xw)}\right\|_1  \notag \\
    &=\frac{1}{2} \sum_{i=1}^{m} (y_i-x_iw-z_i)^2 + \delta \sum_{i=1}^{m} \frac{\left|z_i\right|_1}{S(y_i-x_iw)} \notag 
\end{align}
satisfies the KL property. $H(w,z)$ is also denoted by $H(u)$, where $u = [w;z ]$.

Directly verifying the KL property of a function by Definition \ref{KL Property} is relatively difficult; it is often easier to prove by using the definition of semi-algebraic functions.
Specifically, we verify that $H(u)$ is a semi-algebraic function. 
Since semi-algebraic functions are KL functions (as described in Theorem 3 of \cite{bolte2014proximal}), it follows that $H(u)$ also satisfies the KL property.

We first present the definition of semi-algebraic functions.
\begin{definition}
\textbf{(Semialgebraic Functions)} \cite{attouch2010proximal, attouch2013convergence} A subset of $\mathbb{R}^d$ is called semialgebraic if it can be written as:
\begin{align}
\bigcup_{i=1}^{q} \left\{ u \in \mathbb{R}^d : P_i(u) = 0,\, Q_i(u) < 0 , i=1,2,\cdots ,p\right\}, \notag
\end{align}
where $P_i, Q_i$ are all real polynomial functions.
A function \( \sigma : \mathbb{R}^d \to \mathbb{R} \cup \{+\infty\} \) is semialgebraic if its graph is a semialgebraic subset of \( \mathbb{R}^{n+1} \).  
\end{definition}
We then state two usefull properties \cite{attouch2010proximal}.
\begin{itemize}
    \item Finite sums and products of semialgebraic functions are semialgebraic;
    \item Compositions of semialgebraic functions or mappings are semialgebraic;
\end{itemize}
By the above properties, we know that it suffices to prove that \( \frac{1}{S(x)} \) and \( |u| \) are semi-algebraic functions to conclude 
that \( H(u) = H(w,z) \) is a semi-algebraic function. We let 
\begin{align}
  U_1 = \left\{ (u,v) \in \mathbb{R}^{2} : v - \frac{1}{ 2 \sqrt{\delta} } u^2 - \frac{ \sqrt{\delta} }{2} = 0,\, u- \sqrt{\delta}<0, -u - \sqrt{\delta}<0 \right\}, \notag
\end{align}
\begin{align}
  U_2 = \left\{ (u,v) \in \mathbb{R}^{2} : u - v = 0,\, -u + \sqrt{\delta} \leq 0 \right\}, \notag
\end{align}
\begin{align}
  U_3 = \left\{ (u,v) \in \mathbb{R}^{2} : u + v = 0,\, u + \sqrt{\delta} \leq 0 \right\}. \notag
\end{align}
\begin{align}
  U_4 = \left\{ (u,v) \in \mathbb{R}^{2} : u - v = 0,\, -u \leq 0 \right\}, \notag
\end{align}
\begin{align}
  U_5 = \left\{ (u,v) \in \mathbb{R}^{2} : u + v = 0,\, u \leq 0 \right\}. \notag
\end{align}
Then \( \left\{ \left( u, \frac{1}{S(u)}\right)| u \in \mathbb{R} \right\}, \) can be written as:
\begin{align}
   \left\{ \left. \left( u, \frac{1}{S(u)}\right) \right| u \in \mathbb{R} \right\} = U_1 \cup U_2 \cup U_3, \notag
\end{align}
and \( \left\{ \left( u, |u|\right) | u \in \mathbb{R}\right\} \) can be written as:
\begin{align}
  \left\{ \left( u, |u|\right) | u \in \mathbb{R}\right\} = U_4 \cup U_5 . \notag
\end{align}
Therefore $H(u)$ is a semi-algebraic function and, moreover, a KL function.

Since \( \omega(u^0) \) is a set that may contain more than one point, the pointwise KL property described in Definition \ref{KL Property} is not sufficient to establish Theorem \ref{Apply KL Property}.
Therefore, it is necessary to introduce a lemma to characterize the uniformized KL property.
\begin{lemma}\textbf{(Uniformized KL property)} \label{Uniformized KL property} \cite{bolte2014proximal}
\( \Omega \) is a compact set and \( \sigma : \mathbb{R}^d \to (-\infty, \infty] \) is a proper and lower semicontinuous function. Suppose that \( \{ \sigma(u)| u \in \Omega \}\) only contains a constant and \( \sigma \) satisfies the KL property at each point of \( \Omega \). Then, there exist \( \varepsilon > 0 \), \( \eta > 0 \), and a continuous concave function \( \phi: [0, \eta) \to \mathbb{R}_{+} \) satisfying: 
\( \phi(0) = 0 \),
\( \phi \) is \( C^1 \) on \( (0, \eta) \) and is continuous at \(0 \), with \( \phi' > 0 \) on \( (0, \eta) \), such that for all \( \bar{u} \in \Omega \) and all \( u\) in the set
\begin{align}
\left\{ \left. u \in \mathbb{R}^d \right| \operatorname{dist}(u, \Omega) < \varepsilon \right\} \cap \left\{ \left. u \in \mathbb{R}^d \right| \sigma(\bar{u}) < \sigma(u) < \sigma(\bar{u}) + \eta \right\}, \notag 
\end{align}
we have
\begin{align}
\phi'\left( \sigma(u) - \sigma(\bar{u}) \right) \cdot \operatorname{dist}(0, \partial \sigma(u)) \geq 1. \notag 
\end{align}
\end{lemma}
\begin{proof}
    Please refer to Lemma 6 of \cite{bolte2014proximal}.
\end{proof}
This lemma mainly ensures that the choices of \( \phi \) and \( \eta \) are uniform with respect to \( \bar{u} \).

\subsection{Proof of Theorem \ref{Apply KL Property}}
With all the preparatory results in place, we next prove Theorem \ref{Apply KL Property} by an approach similar to that in \cite{bolte2014proximal, attouch2013convergence, liu2020optimization}.

\begin{proof}\textbf{(Proof of Theorem \ref{Apply KL Property})}
    According to the proof of Lemma \ref{Properties of the Critical Points}, there exists a subsequence \( u^{k_q} \to u^*\) as $k\to \infty$ and $\lim_{k \to \infty} H(u^{k}) = H(u^{*})$. The conclusion clearly holds when there exists a $k^*$ such that \( H(u^{*}) = H(u^{k^*}) \), 
    which because the monotonicity of $H(u)$ and the fact that $ \left\|u^{k+1} - u^{k} \right\|^2_2 \leq \frac{1}{\rho_1} \left( H \left( u^{k+1} \right) - H\left(u^{k} \right) \right)= 0$ if $k \geq k^*$ (Theorem \ref{Subgradient Lower Bound}). 
    Therefore, we only consider the case where \( H(u^{*}) < H(u^{k}) \) for all $k$.

    Based on properties of the critical points (Lemma \ref{Properties of the Critical Points}), we can obtain that for any $\varepsilon, \eta > 0$, there exists a
    a sufficiently large integer \( K \) such that for any $ k>K$, 
    \begin{align}
        \operatorname{dist}(u^k, \Omega) < \varepsilon, H(u^*) < H(u^k) < H(u^*) + \eta.  \notag 
    \end{align}
    Therefore, it can be obtained from Lemma \ref{Uniformized KL property} that when \( k \) is sufficiently large, the sequence \( \{u_k\} \) satisfies the uniformized KL property. Since \( \omega(u^0)\) is a nonempty compact set and \( H(u) \) is constant on \( \omega(u^0) \), as long as we set \( \sigma = H(u), \Omega = \omega(u^0) \) and \( \bar{u} = u^* \) in Lemma \ref{Uniformized KL property}, we obtain that for any $k>K$:
    \begin{align} 
    \phi'\left( H(u^k) - H(u^*) \right) \cdot \operatorname{dist}(0, \partial H(u^k)) \geq 1.  \label{KLunequal}
    \end{align}
    According to Theorem \ref{Subgradient Lower Bound}:
    \begin{align} 
    \operatorname{dist}(0, \partial H(u^k)) \leq \left\|[A^{k}; B^{k}]\right\|_2 \leq \rho_2 \left\|[w^{k}; z^{k}] - [w^{k-1}; z^{k-1}] \right\|_2 = \rho_2 \left\|u^{k} - u^{k-1} \right\|_2  \label{KLbound}
    \end{align}
    By substituting (\ref{KLbound}) into (\ref{KLunequal}), we obtain:
    \begin{align}
    \phi'\left( H(u^k) - H(u^*) \right) \geq \frac{1}{\rho_2 \left\|u^{k} - u^{k-1} \right\|_2}. \label{phiequal1}
    \end{align}
    Moreover, since \( \phi \) is a concave function, 
    \begin{align}
        &\phi \left( H(u^k) - H(u^*) \right) - \phi \left( H(u^{k+1}) - H(u^*) \right) \notag \\
        \geq & \phi'\left( H(u^k) - H(u^*) \right) \left( H(u^k) - H(u^{k+1})\right). \label{phiequal2}
    \end{align}
    To simplify the notation, we let:
    \begin{align}
        \Delta_{p,q} = \phi \left( H(u^p) - H(u^*) \right) - \phi \left( H(u^q) - H(u^*) \right),p,q \in \mathbb{N}. \notag 
    \end{align}
    and 
    \begin{align}
        C = \frac{\rho_1}{\rho_2} > 0.
    \end{align}
    According to (\ref{phiequal1}) and (\ref{phiequal2}), we have:
    \begin{align}
        \Delta_{k,k+1} &\geq \phi'\left( H(u^k) - H(u^*) \right) \left( H(u^k) - H(u^{k+1})\right) \notag \\
        &\geq \frac{1}{\rho_2 \left\|u^{k} - u^{k-1} \right\|_2} \cdot \rho_1 \| u^{k+1} - u^{k} \|_2^2 \notag \\
        &\geq \frac{\| u^{k+1} - u^{k} \|_2^2}{C \left\|u^{k} - u^{k-1} \right\|_2}, \notag 
    \end{align}
    which ie equal to:
    \begin{align}
        \| u^{k+1} - u^{k} \|_2 \leq \sqrt{C \Delta_{k,k+1} \left\|u^{k} - u^{k-1} \right\|_2}.  \notag 
    \end{align}
    Using an arithmetic–geometric mean inequality: $ 2\sqrt{ab} \leq a + b, \forall a,b >0$, we let $a =  \left\|u^{k} - u^{k-1} \right\|_2$ and $b = C \Delta_{k,k+1}$, 
    \begin{align} \label{ARGM}
        2\| u^{k+1} - u^{k} \|_2 \leq  C \Delta_{k,k+1} + \left\|u^{k} - u^{k-1} \right\|_2.
    \end{align}
    For any $k>K$, replace $k$ in (\ref{ARGM}) and sum over $i = K+1,K+1,\cdots k$:
    \begin{align}
        2 \sum_{i=K+1}^{k} \| u^{k+1} - u^{k} \|_2 &\leq  C \sum_{i=K+1}^{k}\Delta_{k,k+1} + \sum_{i=K+1}^{k}\left\|u^{k} - u^{k-1} \right\|_2 \label{deltasum1}\\
        &\leq  C \Delta_{K+1,k+1} + \sum_{i=K+1}^{k}\left\| u^{k+1} - u^{k} \right\|_2 + \left\| u^{K+1} - u^{K} \right\|_2. \label{deltasum2},
    \end{align}
    where (\ref{deltasum1}) $ \implies$ (\ref{deltasum2}) because $\Delta_{p,q} + \Delta_{q,r} = \Delta_{p,r}$.
    We cancel terms on both sides of the above inequality and obtain:
    \begin{align}
        \sum_{i=K+1}^{k} \| u^{k+1} - u^{k} \|_2 &\leq  C \Delta_{K+1,k+1} + \left\| u^{K+1} - u^{K} \right\|_2 \notag \\
        &\leq  C \left(\phi \left( H(u^{K+1}) - H(u^*) \right) - \phi \left( H(u^{k+1}) - H(u^*) \right) \right) + \left\| u^{K+1} - u^{K} \right\|_2 \notag \\
        &\leq  C \phi \left( H(u^{K+1}) - H(u^*) \right) + \left\| u^{K+1} - u^{K} \right\|_2. \label{limitedB}
    \end{align}
    The right-hand side of Inequality (\ref{limitedB}) is bounded and independent of \(k\). By the definition of the convergence of a series, we obtain:
    \begin{align}
        \sum_{k=1}^{\infty} \|u^{k+1} - u^{k} \|_2 < + \infty \label{infintelength}
    \end{align}
    which means that $\sum_{k=K}^{\infty} \|u^{k+1} - u^{k} \|_2 \to 0$ as $K \to \infty$, and:
    \begin{align} 
    \sum_{k=1}^{\infty} \left\| [w^{k+1}; z^{k+1}] - [w^{k}; z^{k}] \right\|_2 < + \infty.  \notag
    \end{align}
    Additionally, for any \( q > p > K \), we have the following:
    \begin{align}
        u^{q} - u^{p}  = \sum_{k=p}^{q-1} (u^{k+1} - u^{k}). \notag
    \end{align}
    By the triangle inequality, we have:
    \begin{align}
        \left\| u^{q} - u^{p} \right\|_2  = \left\| \sum_{k=p}^{q-1} (u^{k+1} - u^{k}) \right\|_2 \leq \sum_{k=p}^{q-1}\left\| (u^{k+1} - u^{k}) \right\|_2. \notag
    \end{align}
    Combining with (\ref{infintelength}), it is easy to conclude that $\{u^k\}$ is a Cauchy sequence.
\end{proof}

\section*{Appendix B: Proof of Theorem \ref{upper bound}}

\subsection{Proof of a Preliminary Lemma}

We first denote the objective function (\ref{objective}) as $H(\delta,w,z)$ with $\delta$ being also regarded as a parameter, and present an important theorem as a prelude.

\begin{theorem} \textbf{Berge's Maximum Theorem.} \label{Berge's Maximum Theorem}
Let \( \Lambda \subset \mathbb{R}^p \) and \( U \subset \mathbb{R}^q \) be both non-empty, where $p,q \in \mathbb{N}$.
Assume \( f : \Lambda \times U \rightarrow \mathbb{R} \) is a continuous function and that  
\( \varphi : \Lambda \rightarrow 2^U \) is a continuous correspondence that is compact and non-empty for all \( \lambda \in \Lambda \).  
For all \( \delta \in \Lambda \), define
\begin{align}
h(\delta) = \max_{u \in \varphi(\delta)} f(\delta,u)
\quad \text{and} \quad
\Gamma(\delta) = \{ u \in \varphi(\delta) | h(\delta) = f(\delta,u) \}.
\end{align}
Then \( \Gamma(\delta) \) is non-empty for all \( \lambda \in \Lambda \), \( \Gamma(\delta) \) and is upper hemicontinuous, and \( h(\delta) \) is continuous.
\end{theorem}

\begin{remark} \label{continuous}
    By setting \( \varphi(\delta) \stackrel{\text{def}}{=} U \),\( \Lambda = \left\{ \delta|\delta > 0 \right\} \), \( p = 1 \), \( q = m+n \), \( f = -H(w, z) \), and \( u = [w;z] \), we can invoke Theorem \ref{Berge's Maximum Theorem} to conclude that \( \Gamma(\delta) = \{ [w;z] \in U | \min_{[w;z] \in U}H(\delta,w,z) = \min_{[w;z] \in \mathbb{R}^{m+n}}H(\delta,w,z) = H(\delta,w,z) \} = \{ [ w^*_{\delta}; z^*_{\delta} ]\} \) is upper semicontinuous.
    According to the assumptions in Theorem \ref{upper bound}, for any \( \delta > 0 \), the set \( \Gamma(\delta) \) admits a unique point. Hence, \( \Gamma(\delta) \) can be regarded as a function on \( \Lambda \), then it is a continuous function since upper hemicontinuous functions are continuous.
\end{remark}

Next, we prove an additional lemma to address two extreme cases respectively.
\begin{lemma} \label{deltalimitopt}
Assume that \( \Gamma(\delta) \)is a function with respect to \( \delta \), as described in Remark \ref{continuous}. Then we can obtain that:
\begin{itemize}
    \item[(1)] \( \lim_{\delta \to 0}\Gamma(\delta) = [w^*_{0};y-Xw^*_{0}] \),
    \item[(2)] \( \lim_{\delta \to +\infty} \Gamma(\delta) = [w^*_{\infty};0] \),
\end{itemize}
where \( w^*_{0}\) is a global optimal solution of 
\begin{align} \label{equalL1}
\min_{w \in \mathbb{R}^n} \left\| \frac{y-Xw}{S(y-Xw)} \right\|_1,
\end{align}
and \( w^*_{\infty}\) is a global optimal solution of 
\begin{align}
\min_{w \in \mathbb{R}^n} \|y-Xw\|_2^2.
\end{align}
\end{lemma}

\begin{proof}
We only prove (1); (2) follows similarly. 

Recall Optimization Problem (\ref{SARMOPT2}) and the definition of \( \Gamma(\delta) \):
\begin{align}
\Gamma(\delta) &= \underset{w \in \mathbb{R}^n, z \in \mathbb{R}^m}{\arg\min} H(\delta,w,z) \notag \\
&= \underset{w \in \mathbb{R}^n, z \in \mathbb{R}^m}{\arg\min} \frac{1}{2}\|y-Xw-z\|_2^2 + \delta \left\|\frac{z}{S(y-Xw)}\right\|_1.
\end{align} 
Suppose that $\bar{w}$ is a global optimal solution of (\ref{equalL1}), $\{\delta^k\}$ is a monotonically decreasing sequence converging to zero and  $\Gamma(\delta^k) = \left[ w^*_{\delta^k},z^*_{\delta^k} \right]$.
$H(\delta,w,z)$ is continuous with respect to $\delta,w,z$.
By the optimality of $\bar{w}$, we have:
\begin{align}
    \left\| \frac{y-X\bar{w}}{S(y-X\bar{w})} \right\|_1 \leq \left\| \frac{y-Xw}{S(y-Xw)} \right\|_1, \forall w \in \mathbb{R}^n.
\end{align}
Furthermore, based on the definition of \( w^*_{\delta^k},z^*_{\delta^k} \), we have that for any $w \in \mathbb{R}^n, z \in \mathbb{R}^m$:
\begin{align}
    &\frac{1}{2}\|y-Xw^*_{\delta^k}-z^*_{\delta^k}\|_2^2 + \delta^k \left\|\frac{z^*_{\delta^k}}{S(y-Xw^*_{\delta^k})}\right\|_1 \notag \\
    \leq&\frac{1}{2}\|y-Xw-z\|_2^2 + \delta^k \left\|\frac{z}{S(y-Xw)}\right\|_1. \notag
\end{align}
We let $w = \bar{w}, z = y-X\bar{w}$, therefore the following inequality holds:
 \begin{align}
    &\frac{1}{2}\|y-Xw^*_{\delta^k}-z^*_{\delta^k}\|_2^2 + \delta^k \left\|\frac{z^*_{\delta^k}}{S(y-Xw^*_{\delta^k})}\right\|_1 \notag \\
    \leq&\frac{1}{2}\|y-X\bar{w}-(y-X\bar{w})\|_2^2 + \delta^k \left\|\frac{y-X\bar{w}}{S(y-X\bar{w})}\right\|_1 \notag \\
    =& \delta^k \left\|\frac{y-X\bar{w}}{S(y-X\bar{w})}\right\|_1. \label{inequalstar}
\end{align}
By rearranging the above inequality, we obtain:
 \begin{align} \label{gammalimit}
    &\|y-Xw^*_{\delta^k}-z^*_{\delta^k}\|_2^2 \leq 2\delta^k \left( \left\|\frac{y-X\bar{w}}{S(y-X\bar{w})}\right\|_1 - \left\|\frac{z^*_{\delta^k}}{S(y-Xw^*_{\delta^k})}\right\|_1 \right).
\end{align}
 Suppose $\left[w^*_{0}, z^*_{0}\right]$ is a limit point of sequence $\left\{\left[ w^*_{\delta^k},z^*_{\delta^k} \right] \right\}$. Since \( U \) is compact, $\left[w^*_{0}, z^*_{0}\right]$ can be attained. By the continuity of \(\Gamma(\delta)\), we have: $\lim_{k\to \infty}\left[ w^*_{\delta^k},z^*_{\delta^k} \right] = [w^*_{0}, z^*_{0}]$. In (\ref{gammalimit}), let $k \to \infty$. Then based on the continuity of $\|y-Xw-z\|_2^2$ and $\left\|\frac{z}{S(y-Xw)}\right\|_1$, and the fact that $\delta^k \to 0$, we have $\|y-Xw^*_{0}-z^*_{0}\|_2^2 = 0$. Therefore $z^*_{0} = y-Xw^*_{0}$, which means that $w^*_{0}$ is a feasible solution to (\ref{equalL1}). According to (\ref{inequalstar}):
\begin{align}
    \delta^k \left\|\frac{z^*_{\delta^k}}{S(y-Xw^*_{\delta^k})}\right\|_1
    \leq \delta^k \left\|\frac{y-X\bar{w}}{S(y-X\bar{w})}\right\|_1.  \label{deltainequal}
\end{align}
Removing the shared coefficient $\delta^k$ and taking the limit on both sides of the inequality (\ref{deltainequal}), we obtain:
\begin{align}
    \left\|\frac{y-Xw^*_{0}}{S(y-Xw^*_{0})}\right\|_1 = \left\|\frac{z^*_{0}}{S(y-Xw^*_{0})}\right\|_1
    \leq \left\|\frac{y-X\bar{w}}{S(y-X\bar{w})}\right\|_1. 
\end{align}
From the optimality of $w^*_{0}$, we have:
\begin{align}
    \left\|\frac{y-Xw^*_{0}}{S(y-Xw^*_{0})}\right\|_1 = \left\|\frac{y-X\bar{w}}{S(y-X\bar{w})}\right\|_1,
\end{align}
which means that $w^*_{0}$ is also a global optimal solution of (\ref{equalL1}).
\end{proof}

\subsection{Proof of the Error Bound}

With all the preliminaries in place, we now proceed to prove Theorem \ref{upper bound}.
\begin{proof} \textbf{(Proof of Theorem \ref{upper bound})}
According to Lemma \ref{deltalimitopt}, as \(\delta \to 0\), we have \( z^*_0 = y - X w^*_0 \) and \( \frac{1}{2}\|y-Xw^*_0-z^*_0\|_2^2 =0 < \varepsilon \). 
Moreover, based on the assumption \label{condition2} and Lemma \ref{deltalimitopt}, as \(\delta \to \infty\), \( \frac{1}{2}\|y-Xw^*_{\infty}-z^*_{\infty}\|_2^2 \geq \varepsilon\)
Then, by the continuity of \( \Gamma(\delta) \), we can deduce that there exists a lower bound \( \delta^{\mathrm{lb}} > 0 \) such that for all \( \delta \geq \delta^{\mathrm{lb}} \), \( [w^*_{\delta}; z^*_{\delta}] \) satisfies \( \frac{1}{2}\|y-Xw^*_{\delta}-z^*_{\delta}\|_2^2 = \varepsilon({\delta})\geq \varepsilon \).

We next show the optimal solution of 
\begin{equation}
\begin{aligned} \label{penaltyformopt}
\min_{w \in \mathbb{R}^n, z \in \mathbb{R}^m} \frac{1}{2}\|y-Xw-z\|_2^2 + \delta \left\|\frac{z}{S(y-Xw)}\right\|_1
\end{aligned} 
\end{equation}
is also an optimal solution of 
\begin{equation}
\begin{aligned} \label{equalformoptimization}
    \min_{w \in \mathbb{R}^n, z \in \mathbb{R}^m} \left\|\frac{z}{S(y-Xw)}\right\|_1 \\
    s.t.\frac{1}{2}\|y-Xw-z\|_2^2 \leq \varepsilon({\delta}).
\end{aligned} 
\end{equation}
$\left[ w^*_{\delta}, z^*_{\delta}\right]$ is the optimal solution of (\ref{penaltyformopt}) and we define $\left[ w^{\circ}_{\delta}, z^{\circ}_{\delta}\right]$ is an optimal solution of (\ref{penaltyformopt}). Due to the optimality of $\left[ w^{\circ}_{\delta}, z^{\circ}_{\delta}\right]$ and the fact that $\frac{1}{2}\|y-Xw-z\|_2^2 \leq \varepsilon({\delta})=\frac{1}{2}\|y-Xw^*_{\delta}-z^*_{\delta}\|_2^2$, we have $ \frac{1}{2}\|y-Xw^{\circ}_{\delta}-z^{\circ}_{\delta}\|_2^2 + \delta \left\|\frac{z^{\circ}_{\delta}}{S(y-Xw^{\circ}_{\delta})}\right\|_1 \leq \frac{1}{2}\|y-Xw^*_{\delta}-z^*_{\delta}\|_2^2 + \delta \left\|\frac{z^*_{\delta}}{S(y-Xw^*_{\delta})}\right\|_1$.
Moreover, due to the optimality of \( \left[ w^*_{\delta}, z^*_{\delta} \right] \), the inequality above holds with equality. Therefore $\left[ w^*_{\delta}, z^*_{\delta}\right]$ is an optimal solution (\ref{equalformoptimization}).

For convenience, we fix a certain value of \( \delta \geq \delta^{\mathrm{lb}} \) and omit the subscript \(\delta\) in the notation hereafter.

Suppose $[w^*; z^*]$ is the global solution of Optimization Problem (\ref{penaltyformopt}) with given $\delta$ and $\lambda \in (1,+\infty)$.
By the optimality condition of $z^*$, we obtain:
\begin{align}
z_i^* &=\underset{z_i \in \mathbb{R}}{\arg\min} \frac{1}{2}(y_i - x_i^T w^*-z_i)^2 + \delta \left| \frac{z_i}{S(y_i - x_i^T w^*)}\right| \notag \\
    & =     \begin{cases}
        0, & \text{if } | y_i - x_i^T w^* |  \leq \sqrt{\delta} \\
         y_i - x_i^T w^* - \frac{\delta}{y_i - x_i^T w^*} , & \text{if } | y_i - x_i^T w^*|  > \sqrt{\delta},
    \end{cases} \label{zoptimal}
\end{align}
We then define several index sets:
\begin{align}
    &\mathcal{I} = \left\{ i \in \mathbb{N}| |y_i-x_i^Tw^*|\geq \sqrt{\delta}\right\}, \notag\\
    &\mathcal{I}_a = \left\{ i \in \mathbb{N}| |y_i-x_i^Tw^*|\geq \lambda \sqrt{\delta}\right\}, \notag\\
    &\mathcal{I}_b = \left\{ i \in \mathbb{N}| \sqrt{\delta} \leq |y_i-x_i^Tw^*|< \lambda\sqrt{\delta}\right\}, \notag
\end{align}
where $I = I_a \cup I_b$. Then based on (\ref{zoptimal}) we can obtain that:
\begin{align}
    \left\|\frac{z^*}{S(y-Xw^*)}\right\|_1 &= \sum_{i=1}^{m} \left| \frac{z_i^*}{S(y_i - x_i^T w^*)}\right| \notag\\
    &= \sum_{i \in \mathcal{I}} \left| \frac{y_i - x_i^T w^* - \frac{\delta}{y_i - x_i^T w^*}}{\left| y_i - x_i^T w^* \right| }\right| \notag\\
    &= \sum_{i \in \mathcal{I}} \left| 1 - \frac{\delta}{(y_i - x_i^T w^*)^2}\right| \notag\\
    &\geq \sum_{i \in \mathcal{I}_a} \left( 1 - \frac{\delta}{(y_i - x_i^T w^*)^2}\right) \label{L1inequal1}
\end{align}
On the other hand, because $\frac{1}{2}\|y-Xw^{true}-z^{true}\|_2^2 = \varepsilon \leq \varepsilon({\delta})$, $\left[ w^{true},z^{true}\right]$ satisfies the constraint of (\ref{equalformoptimization}).
Note that $\left[ w^*, z^*\right]$ is also an optimal solution (\ref{equalformoptimization}), hence, by the optimality of $\left[ w^*, z^*\right]$, we obtain that: 
\begin{align}
    \left\|\frac{z^*}{S(y-Xw^*)}\right\|_1 &= \sum_{i=1}^{m} \left| \frac{z_i^*}{S(y_i - x_i^T w^*)}\right| \notag\\
    &\leq \sum_{i=1}^{m} \left| \frac{z_i^{true}}{S(y_i - x_i^T w^{true})}\right| \notag\\
    &=\sum_{i=1}^{m} \left| \frac{z_i^{true}}{S(z_i^{true} + e_i) }\right| \notag\\
    &=\sum_{i \in \left\{i|z_i^{true}\neq0 \right\}} \left| \frac{z_i^{true}}{S(z_i^{true}) }\right| \notag\\
    &\leq k^{true}.\label{L1inequal2}
\end{align}
Furthermore, combining (\ref{L1inequal1}), (\ref{L1inequal1}) and the fact that for any $ i \in \mathcal{I}_a$, \( \left|y_i - x_i^T w^*\right|\geq \lambda \sqrt{\delta}\), we obtain that:
\begin{align}
    & \sum_{i \in \mathcal{I}_a} \left( 1 - \frac{1}{\lambda^2}\right) \leq \sum_{i \in \mathcal{I}_a} \left( 1 - \frac{\delta}{(y_i - x_i^T w^*)^2}\right) \leq \left\|\frac{z^*}{S(y-Xw^*)}\right\|_1 \leq k^{true} \\
    \implies& \left( 1 - \frac{1}{\lambda^2}\right) \left| \mathcal{I}_a\right| = \sum_{i \in \mathcal{I}_a} \left( 1 - \frac{1}{\lambda^2}\right) \leq k^{true} \\
    \implies& \left| \mathcal{I}_a\right| \leq \frac{k^{true}}{ \left( 1 - \frac{1}{\lambda^2}\right) } 
\end{align}

The optimality of \([w^*; z^*]\) implies that it satisfies $\nabla_w H(w,z)=0$. From Assumption (\ref{modified attacked}), multiplying by \(X^T\) yields: $X^Ty = X^TX w^{true}  + X^Te + X^Tz^{true}$. 
We further denote 
$ z^*_a = [(z^*_a)_i]  = 
\begin{cases}
(z^*)_i, & \text{if } i \in \mathcal{I}_a \\ 
0, & \text{if } i \notin \mathcal{I}_a
\end{cases}$,
and 
$\begin{cases}
 z^*_b = [(z^*_b)_i]  = 
(z^*)_i, & \text{if } i \in \mathcal{I}_b \\ 
0, & \text{if } i \notin \mathcal{I}_b
\end{cases}$.

Therefore, the following relationship holds for any \(\lambda \in (1,+\infty)\):
\begin{align}
    &X^T (y - Xw^* - z^*) - \delta X^T \left[ \frac{\left|z_i^*\right|S'(y_i - x_i^T w^*)}{S^2(y_i - x_i^T w^*)} \right] = 0 \notag\\
    \implies &X^TX(w^{true} - w^*) + X^T(z^{true} - z^*_a) = X^T z^*_b + \delta X^T \left[ \frac{\left|z_i^*\right| S'(y_i - x_i^T w^*)}{S^2(y_i - x_i^T w^*) } \right] - X^Te \notag
\end{align}
Taking the \(\ell_2\)-norm on both sides and applying the triangle inequality, we have:
\begin{align}
    \left\|\left[ I,X^T\right] 
    \begin{bmatrix}
    w^{true} - w^*  \\
    z^{true} - z^*_a
    \end{bmatrix}
    \right\|_2 & \leq \left\|X^T\right\|_2\left\|z^*_b\right\|_2 + \delta \left\|X^T\right\|_2\left\| \left[ \frac{\left|z_i^*\right|S'(y_i - x_i^T w^*)}{S^2(y_i - x_i^T w^*)} \right] \right\|_2  + \varepsilon \notag \\
 & \leq \left\|z^*_b\right\|_2 + \delta \sqrt{\sum_{i\in \mathcal{I}}  \frac{\left(z_i^*\right)^2}{S^4(y_i - x_i^T w^*)} }  + \varepsilon \notag\\
 & \leq \left\|z^*_b\right\|_2 + \sqrt{|\mathcal{I}|\delta}  + \varepsilon \notag\\
 & \leq \delta \sqrt{\sum_{i\in \mathcal{I}_b}  \left|\frac{z_i^* }{y_i - x_i^T w^*}\right|^2} + \sqrt{|\mathcal{I}|\delta}  + \varepsilon \notag\\
 & \leq \delta \sqrt{\sum_{i \in \mathcal{I}_b} \left( 1 - \frac{\delta}{(y_i - x_i^T w^*)^2}\right)^2 } + \sqrt{|\mathcal{I}|\delta}  + \varepsilon \label{lambdanomo1}\\
 & \leq \left( \lambda - \frac{1}{\lambda}\right)\sqrt{|\mathcal{I}_b|\delta} + \sqrt{|\mathcal{I}|\delta}  + \varepsilon \label{lambdanomo2}, \\
 &\leq \left( \lambda - \frac{1}{\lambda} + 1 \right ) \cdot \sqrt{m\delta} + \varepsilon, \label{boundcond1} 
\end{align}
where (\ref{lambdanomo1}) $\implies$ (\ref{lambdanomo2}) due to the definition of $\mathcal{I}_b$.

Applying the R.I.P. condition (Definition \ref{RIP}), we obtain:
\begin{align}
    &\sqrt{1- \mu_{s} }  \left\|    
    \begin{bmatrix}
    w^{true} - w^*  \\
    z^{true} - z^*_a
    \end{bmatrix}\right\|_2 
    \leq
    \left\|\left[ I,X^T\right] 
    \begin{bmatrix}
    w^{true} - w^*  \\
    z^{true} - z^*_a
    \end{bmatrix}
    \right\|_2 \label{boundcond2}.
\end{align}
We also have:
\begin{align}
      \left\|w^* -  w^{true}\right\|_2  \leq    
    \left\|
    \begin{bmatrix}
    w^{true} - w^*  \\
    z^{true} - z^*_a
    \end{bmatrix}
    \right\|_2. \label{boundcond3}
\end{align}
Therefore, combine (\ref{boundcond1}), (\ref{boundcond2}) and (\ref{boundcond3}), we can obtain:
\begin{align}
    &\sqrt{1- \mu_{s} }  \left\|w^* -  w^{true}\right\|_2  \leq \left( \lambda - \frac{1}{\lambda} + 1 \right ) \cdot \sqrt{m\delta} + \varepsilon \notag\\
    \implies& \left\|w^* -  w^{true}\right\|_2 \leq \frac{\left( \lambda - \frac{1}{\lambda} + 1 \right ) \cdot \sqrt{m\delta} + \varepsilon}{\sqrt{1- \mu_{s} }}, \notag
\end{align}
where
\begin{align}
s = \left(\frac{\lambda^2}{ \lambda^2 - 1 } + 1 \right) k^{true} + n.\notag
\end{align}
Note that the inequality holds for any \(\lambda>1\).
\end{proof}

\bibliographystyle{unsrt}  
\bibliography{references}

\end{document}